%% file: ted.tex
\documentclass[letterpaper,11pt]{article}
\usepackage[utf8]{inputenc}
\usepackage[margin=1in]{geometry}
\usepackage{authblk}

\usepackage{amsmath}
\usepackage{amsthm}
\usepackage{amssymb}
\usepackage{thmtools}
\usepackage{thm-restate}
\usepackage{enumitem}
\usepackage{dsfont}
\usepackage{xspace}
\usepackage{comment}
\usepackage{xparse}
\usepackage[ruled,noend,resetcount,linesnumbered]{algorithm2e}
\usepackage{tikz}
\usepackage[draft]{fixme}
\usepackage[colorlinks]{hyperref}
\usepackage[capitalize]{cleveref}
\usepackage{subcaption}

\fxsetup{mode=multiuser,theme=color, layout=inline}
\FXRegisterAuthor{tk}{atk}{\color{blue} {\bf Tomasz}}
\FXRegisterAuthor{hs}{ahs}{\color{orange} {\bf Hamed}}
\FXRegisterAuthor{jg}{ajg}{\color{red} {\bf Jake}}

\newcommand{\op}{\texttt{(}}
\newcommand{\cl}{\texttt{)}}

\newcommand{\str}[2][]{\mathsf{P}_{#1}(#2)}
\newcommand{\T}{\mathsf{T}}
\newcommand{\F}{\mathsf{F}}
\newcommand{\G}{\mathsf{G}}
\newcommand{\Sub}{\mathsf{sub}}
\newcommand{\B}{\mathcal{B}}
\newcommand{\lab}{\lambda}
\newcommand{\pr}{\mathsf{parent}}
\newcommand{\hgt}{\mathsf{height}}
\newcommand{\chld}{\mathsf{children}}
\newcommand{\vV}{\bar{V}}
\newcommand{\sub}{\subseteq}
\newcommand{\sm}{\setminus}

\newcommand{\Oh}{\mathcal{O}}
\newcommand{\tOh}{\tilde{\Oh}}

\newcommand{\ed}{\mathsf{ed}}
\newcommand{\ted}{\mathsf{ted}}

\newcommand{\per}{\mathsf{per}}

\newcommand{\brkp}{\mathsf{B}}
\newcommand{\poly}{\mathrm{poly}}

\newcommand{\rot}{\mathsf{rot}}

\newcommand{\ga}{\mathsf{GA}}
\renewcommand{\aa}{\mathsf{A}}
\newcommand{\ta}{\mathsf{TA}}
\newcommand{\mtch}{\mathsf{M}}

\newcommand{\A}{\mathcal{A}}

\newcommand{\dd}{\mathinner{.\,.\allowbreak}}
\newcommand{\Exp}{\mathbb{E}}

\newcommand{\rev}[1]{\overline{#1}}
\newcommand{\Zz}{\mathbb{Z}_{\ge 0}}
\newcommand{\Zp}{\mathbb{Z}_{+}}

\DeclareMathOperator{\lcm}{lcm}
\newcommand{\ORS}{\mathsf{ORS}}
\newcommand{\hu}{\hat{u}}
\newcommand{\hv}{\hat{v}}
\newcommand{\hM}{\hat{M}}
\newcommand{\onto}{\to}

\newcommand{\hlab}{\hat{\lab}}
\newcommand{\hq}{\hat{q}}

\setlist[enumerate]{nosep, topsep=1ex}
\setlist[itemize]{nosep, topsep=1ex}
\setlist[description]{nosep,topsep=1ex}

\newtheorem{theorem}{Theorem}[section]
\newtheorem{corollary}[theorem]{Corollary}
\newtheorem{proposition}[theorem]{Proposition}
\newtheorem{lemma}[theorem]{Lemma}

\newtheorem{claim}[theorem]{Claim}
\newtheorem{observation}[theorem]{Observation}

\theoremstyle{definition}
\newtheorem{definition}[theorem]{Definition}

\theoremstyle{remark}

\title{$\tOh(n+\poly(k))$-time Algorithm for Bounded Tree Edit Distance}

\author[1]{Debarati Das}
\author[2]{Jacob Gilbert}
\author[2]{MohammadTaghi Hajiaghayi}
\author[3]{Tomasz Kociumaka}
\author[4]{Barna~Saha}
\author[2]{Hamed Saleh}

\affil[1]{Pennsylvania State University, United States}
\affil[ ]{\texttt{debaratix710@gmail.com}}
\affil[2]{University of Maryland, United States}
\affil[ ]{\texttt{jgilber8@umd.edu}\; \texttt{hajiaghayi@gmail.com}\; \texttt{hamed@cs.umd.edu}}
\affil[3]{Max Planck Institute for Informatics, Germany}
\affil[ ]{\texttt{tomasz.kociumaka@mpi-inf.mpg.de}}
\affil[4]{University of California, San Diego, United States}
\affil[ ]{\texttt{barnas@ucsd.edu}}

\date{}

\begin{document}

\maketitle

\begin{abstract}
\input{src/abstract}
\end{abstract}

\section{Introduction}
\input{src/intro}

\section{Technical Overview}\label{sec:overview}
\input{src/overview}

\section{Preliminaries}
\label{sec:prelim}
\input{src/prelim}

\section{Labeling Refinements}
\input{src/refinements}

\section{Periodicity Reduction in Strings} 
\input{src/string_per}

\section{Horizontal Periodicity Reduction}\label{sec:hor}
\input{src/horiz_per}

\section{Vertical Periodicity Reduction}\label{sec:ver}
\input{src/vert_per}

\section{Full Periodicity Reduction}\label{sec:fur}
\input{src/full_per}

\section{Main Algorithm}
\input{src/main_alg}

\section{Conclusion}
\input{src/conclusion}

\section*{Acknowledgment}
\input{src/acknowledgment}

\bibliographystyle{alphaurl}
\bibliography{ted}
\end{document}

%% file: src/abstract.tex
Computing the \emph{edit distance} of two strings is one of the most basic problems in computer
science and combinatorial optimization. \emph{Tree edit distance} is a natural generalization of edit distance in which the task is to compute a measure of dissimilarity between two (unweighted) rooted trees with node labels. Perhaps the most notable recent application of tree edit distance is in NoSQL big databases, such as MongoDB, where each row of the database is a JSON document represented as a labeled rooted tree and finding dissimilarity between two rows is a basic operation.
Until recently, the fastest algorithm for tree edit distance ran in cubic time (Demaine, Mozes, Rossman, Weimann; TALG'10);
however, Mao (FOCS'21) broke the cubic barrier for the tree edit distance problem using fast matrix multiplication.

Given a parameter $k$ as an upper bound on the distance,
an $\Oh(n+k^2)$-time algorithm for edit distance has been known since the 1980s due to works of Myers (Algorithmica'86) and Landau and Vishkin (JCSS'88).
The existence of an $\tOh(n+\poly(k))$-time algorithm for tree edit distance has been posed as open question, e.g., by Akmal and Jin (ICALP'21), who give a state-of-the-art $\tOh(nk^2)$-time algorithm. 
In this paper, we answer this question positively.

%% file: src/intro.tex
Computing the \emph{edit distance} of two strings is one of the most fundamental problems in theoretical computer
science and combinatorial optimization studied since the 1960s. \emph{Tree edit distance} is a natural generalization of edit distance in which the task is to compute a measure of dissimilarity between two (unweighted) rooted trees with node labels in which every insertion, deletion, or
relabeling operation has unit cost. 
Tree edit distance, first introduced by Selkow~\cite{SELKOW1977184}, extends the
applications of edit distance to areas such as computational biology (e.g., analysis of RNA molecules, where the secondary structure of RNA is represented as a rooted tree)~\cite{10.1016/j.tcs.2004.12.030,10.5555/262228,10.1137/0213024,DBLP:journals/bioinformatics/ShapiroZ90,10.5555/279082.279125}, structured data analysis (e.g., XML)~\cite{10.5555/1315451.1315465,DBLP:conf/vldb/Chawathe99,10.1145/1613676.1613680}, image analysis~\cite{10.1016/S0167-8655(97)00179-7}, and
compiler optimization~\cite{10.1145/1644015.1644017}. Perhaps the most notable recent application of tree edit distance is in NoSQL big databases, such as MongoDB~\cite{MongoDB}, where each row of the database is a JSON document represented as a labeled rooted tree, and finding dissimilarity between two rows is a basic operation.

The computational aspect of  tree edit distance is also widely studied.
Tai~\cite{10.1145/322139.322143} gave the first solution for tree edit distance that runs in
time $\Oh(n^6)$, where $n$ is the total number of nodes in both trees. 
This running time
was later improved in a series of works to $\Oh(n^4)$~\cite{10.1137/0218082},
$\Oh(n^3\log n)$~\cite{10.5555/647908.740125}, and $\Oh(n^3)$~\cite{10.1145/1644015.1644017}. 
Very recently, Mao~\cite{Xiao21} broke the cubic barrier by showing an $\Oh(n^{
2.9546})$-time algorithm via a reduction to max-plus product of bounded-difference matrices.
In a follow-up preprint, D\"{u}rr further improved the running time to $\Oh(n^{2.9149})$~\cite{Duerr2022}.

Boroujeni, Ghodsi, Hajiaghayi, and Seddighin~\cite{DBLP:conf/stoc/BoroujeniGHS19} presented a $(1+\epsilon)$-approximation algorithm for tree
edit distance that runs in $\tOh(n^2)$ time.\footnote{The $\tOh(\cdot)$ notation suppresses polylogarithmic factors.} They also obtained an $\Oh(\sqrt{n})$-factor approximation algorithm that runs in $\tOh(n)$ time. 
Very recently, Seddighin and Seddighin~\cite{DBLP:conf/innovations/SeddighinS22} gave an $\Oh(n^{1.99})$-time $(3+\epsilon)$-approximation algorithm
for tree~edit~distance.

The problem has also been considered from the lower-bound perspective. 
In particular, Bringmann, Gawrychowski, Mozes, and Weimann~\cite{10.1145/3381878} proved that the cubic running time barrier
for weighted tree edit distance cannot be broken unless APSP
admits a truly subcubic time solution and weighted $k$-clique
admits an $\Oh(n^{k-\epsilon})$-time solution. The existence of such a lower bound
was previously conjectured by Abboud~\cite{Amir14} in a collection of open problems in fine-grained complexity.

In contrast to
tree edit distance, approximation algorithms for string edit distance have
been subject to many studies~\cite{10.1109/FOCS.2010.43,10.1145/1536414.1536444,10.1109/FOCS.2004.14,10.5555/1109557.1109644,10.1145/3313276.3316371,10.5555/874063.875596, 10.5555/279082.279125,DBLP:journals/jacm/BoroujeniEGHS21,10.1145/3422823,DBLP:conf/focs/AndoniN20}. After a series of recent developments~\cite{DBLP:journals/jacm/BoroujeniEGHS21,10.1145/3422823,GRS:20,KS20,BR19},
the current best bound is an algorithm by Andoni and Nosatzki~\cite{DBLP:conf/focs/AndoniN20}, for any constant $\epsilon > 0$,
provides a constant-factor approximation of edit distance between strings of length $n$ in time~$\Oh(n^{1+\epsilon})$.

Given a parameter $k$ as an upper bound on the  distance,
an $\Oh(n+k^2)$-time algorithm for edit distance is known since the 1980s, due to  Myers~\cite{DBLP:journals/algorithmica/Meyers86} and Landau and Vishkin~\cite{LV88}, who combined suffix trees with an elegant greedy algorithm. 
The existence of such an $\tOh(n+\poly(k))$-time algorithm for tree edit distance (even for unlabeled trees) remained open despite persistent effort from researchers in the field. 
In particular, this question was posed by Mao~\cite{Xiao21} and Akmal and Jin~\cite{DBLP:conf/icalp/AkmalJ21}.
The current fastest algorithm for the problem runs in $\tOh(nk^2)$ time~\cite{DBLP:conf/icalp/AkmalJ21}, improving on the previous results of $\Oh(nk^3)$ time by Touzet~\cite{10.1007/11496656_29}.
In this paper, we answer this open question affirmatively by providing an $\tOh(n+k^{15})$-time algorithm for bounded tree edit distance. 

The previous algorithms for computing tree edit distance~\cite{10.1145/322139.322143,10.1137/0218082,10.5555/647908.740125,10.1145/1644015.1644017,Xiao21} are dynamic-programming-based procedures, and the solutions for bounded tree edit distance~\cite{10.1007/11496656_29,DBLP:conf/icalp/AkmalJ21} are obtained by appropriately pruning the set of states in earlier general-purpose algorithms.
Our strategy is very different: the main effort is to greedily match all but $\Oh(\poly(k))$ nodes of the input trees
so that any polynomial-time algorithm can be used to solve the residual instances of the problem.
The greedy approach is sufficiently powerful only for trees avoiding certain synchronized periodicity,
and therefore we start with a preprocessing step that eliminates appropriate periodic structures.
Although such an approach is fairly simple to realize for strings,
implementing it on trees requires several novel components of their own interest.
This is because, so far, the underlying techniques have not been used on trees,
and this setting brings many challenges absent in the context of strings.
For example, periodicity in trees comes in two flavors, which we name \emph{vertical} and \emph{horizontal}, and we provide efficient procedures detecting both kinds.
Another obstacle is that, whereas greedily matching two characters yields two independent instances of the string edit distance problem,
greedily matching two nodes does not produce two independent instances of the tree edit distance problem,
and thus we need a dedicated algorithm to optimally extend a partial alignment to a complete alignment of two trees.
For more details, see the technical overview in \cref{sec:overview}.

Finally, a problem related to (tree) edit distance is the \emph{Dyck edit distance} problem, in which, given a sequence of $n$ parentheses, the task is to find the minimum number of edits (character insertions, deletions, and substitutions) needed to make the sequence well-balanced. The Dyck edit distance has numerous applications~\cite{h:78,k:12} and has been subject to many theoretical studies designing exact~\cite{BGSW19,CDX22,BO16,F:22,Duerr2022} and approximation algorithms~\cite{s:14,DKS21}. It is also known that this problem is at least as hard as Boolean matrix multiplication~\cite{abw:15}. Though Dyck edit distance problem  has a different flavor than tree edit distance (since the goal is to find the minimum number of edits to completion, i.e., to a target of well-balanced parentheses), the existence of an $\tOh(n+\poly(k))$-time algorithm for the bounded version was still a very important open problem (motivated by fixing hierarchical data files, such as XML and JSON). 
Backurs and Onak~\cite{BO16} solved the open problem by providing an exact algorithm for Dyck edit distance that runs in $\Oh(n+k^{16})$ time, which has recently been improved by Fried, Golan, Kociumaka, Kopelowitz, Porat, and Starikovskaya~\cite{F:22} to run in $\Oh(n+k^5)$ time,
and further to $\tOh(n+k^{4.5442})$ using fast matrix multiplication~\cite{F22a,Duerr2022}.
Dyck edit distance falls under the umbrella of a general language edit distance problem~\cite{BGSW19},
which, however, is at least as hard as Boolean matrix multiplication already for $k=0$.

\subsection*{Our Result}

While in edit distance, the goal is to transform a string $S$ into
another string $S'$, in tree edit distance the goal is to transform a tree $T$ into another tree $T'$ using the least number of
edit operations. 
In the most common version of the problem, it is assumed that both trees $T$ and $T'$ are rooted and
that there is a left-to-right order between the children of any node. 
Moreover, each node has a label (independent of its the degree).
The elementary operations are node deletion, node insertion, and node
relabeling. In node deletion, we remove a node $v$ and replace it with all
of its children, preserving their order. The reverse of node deletion
is node insertion, which allows us to select a consecutive set of
siblings and bring them under a new node $v$ which appears at the
previous position of the relocated nodes. A node relabeling simply
modifies the label of an existing node.

In fact, we solve a slightly more general problem of computing the edit distance between two labeled forests,
which are defined as sequences of labeled, rooted, and ordered trees. For two labeled forests $F$ and $G$,
we denote their tree edit distance with $\ted(F,G)$.
Moreover, for a threshold $k$, we denote with $\ted_{\le k}(F,G)$ a value equal to $\ted(F,G)$ (if it is at most $k$)
or $\infty$ (otherwise). The main result of this paper is summarized in the following theorem:

\begin{restatable}{theorem}{main}\label{thm:main}
    There exists a randomized algorithm that, given forests $F,G$ of total size $n$ and an integer $k\in \Zp$,
    computes $\ted_{\le k}(F,G)$ in $\Oh(n\log n + k^{15}\log k \log n)$ time correctly with high probability.
\end{restatable}

%% file: src/overview.tex
Given two strings $X,Y\in \Sigma^{\le n}$ and a threshold $k\in \Zz$, the classic Landau--Vishkin algorithm~\cite{LV88} computes $\ed_{\le k}(X,Y)$ in time $\Oh(n+k^2)$.
The algorithm uses dynamic programming: for each $i\in [0 \dd k]$ and $j\in [-k \dd k]$, it computes an index $d_{i,j}=\max\{x: \ed(X[0 \dd x),Y[0 \dd x+j))\le i\}$.
In terms of the standard quadratic-size DP table, $d_{i,j}$ can be interpreted as the row of the farthest cell on the $j$th diagonal that can be reached with cost at most $i$.
After a linear-time reprocessing of $X,Y$, each value $d_{i,j}$ can be computed in $\Oh(1)$ time; thus, the algorithm takes $\Oh(n+k^2)$ time in total.
From the definition of $d_{i,j}$, we can observe that the alignments produced by the Landau--Vishkin algorithm satisfy the following greedy property:
if a prefix $X[0\dd x)$ is aligned to a prefix $Y[0\dd y)$ and the characters $X[x]=Y[y]$ match,
then these two characters are also aligned (instead of being deleted).
This greedy matching strategy is crucial in achieving $\Oh(n+\poly(k))$ time complexity as it allows the algorithm to focus on $\Oh(\poly(k))$
mismatches and quickly slide through the bulk of the input strings.

A natural question is whether a similar greedy strategy can be applied in the context of the edit distance of two labeled forests $F$ and $G$.
While there could be several ways to formalize such a greedy property, all definitions should capture the following scenario:
if the leftmost roots of $F$ and $G$ have the same label, we should be able to greedily match these roots.
Unfortunately, this is not the case, as illustrated in \cref{fig:nogreedy}.

\begin{figure}[ht]
    \centering
    \input{figs/nogreedy}
    \caption{An example of simple greedy strategy where the alignment matches the leftmost roots, both having labels $a$. However, the unique optimal alignment 
    deletes the entire tree $T_1$ and substitutes the label at the node of $T_2$ from $b$ to $a$.}\label{fig:nogreedy}
\end{figure}
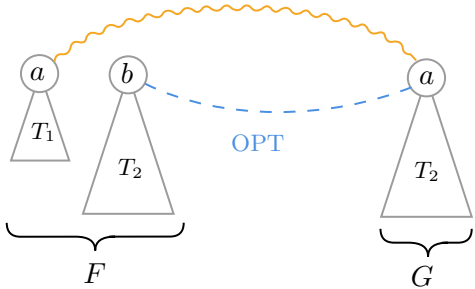

However, it is fairly obvious that if the entire leftmost trees of $F$ and $G$ are the same, then these trees can be matched greedily.
In general, as proved in \cref{lem:optIsGreedy}, if, while constructing an alignment, we are given a chance to match a node $u$ in $F$ 
with a node $v$ in $G$ (or to delete one or both of these nodes), then we can greedily match the two nodes 
provided that the entire subtrees rooted at $u$ and $v$, respectively, are identical.
We call alignments following this principle \emph{greedy alignments}.
Unfortunately, these subtrees can be large.
Thus, first we need to design a process that, for any node $u \in V_F$ (or $v\in V_G$), can efficiently encode the information (e.g., structure, labeling) of the subtree rooted at $u$. We do it using \emph{Look-Ahead Labeling} discussed next.

\paragraph{Look-Ahead Labelling:}
Given unlabeled forests $\F,\G$ and a labeling $\lab : V_\F, V_\G \to \Sigma$ (here, $V_\F$ stands for the node set of $\F$) and a parameter $d$, the depth-$d$ look-ahead labeling $\lambda'=L(\lambda,d)$ satisfies the following property for any pair of nodes $u,v\in V_F\cup V_G$: $\lambda'(u)=\lambda'(v)$ if and only if $\ted(\Sub_{\le d}(u),\Sub_{\le d}(v))=0$. Here, $\Sub_{\le d}(u)$ denotes the $\lambda$-labeled subtree rooted at $u$ trimmed to depth $d$.

To construct $\lambda'$, we first define the parentheses representation of $\F$ and $\G$ with respect to labeling $\lambda$.
For forest $\F$ with labeling $\lambda$,
   this is denoted as $\str[\lab]{\F}$ and can be defined using the following recursion:
    If $\F$ consists of trees $\T_1,\ldots,\T_m$, 
    then $\str[\lab]{\F}=\bigodot_{i=1}^m \str[\lab]{\T_i}$, where $\str[\lab]{\T_i} = \op_{\lab(v_i)} \cdot \str[\lab]{\F_i} \cdot \cl_{\lab(v_i)}$
    and $\cdot$ as well as $\bigodot$ denote concatenation. Here, $v_i$ is the root of $T_i$ and $\F_i$ is the forest obtained from $T_i$ by removal of $v_i$.
To create $\lambda'$ in linear time, we replace the label of each node $u\in V_\F\cup V_\G$ with a Karp--Rabin fingerprint of $P_\lambda(\Sub_{\le d}(v))$. Note by setting $d$ larger than the heights of $\F$ and $\G$, we can create a labeling that for each node encodes the entire subtree rooted at that node. The details can be found in Section~\ref{sec:lookahead}.

\vspace{3mm}

As an additional challenge, it turns out that this greedy approach is beneficial only if the input trees $F,G$ avoid certain periodic structures,
which may be present in the input forests.
Thus, we first design an $\tOh(n)$-time algorithm that, given arbitrary labelled forests $F,G$, produces a pair of forests $F',G'$ such that $\ted_{\le k} (F,G)=\ted_{\le k} (F',G')$ and both $F',G'$ avoid synchronized \emph{horizontal} and \emph{vertical} $k$-periodicity, the two types of periodicity which we describe in the following paragraphs. 
Moreover, with an $\tOh(n+k^3)$-time post-processing, we also compute an alignment $\tilde{\A}$ mapping $\str{F'}$ to $\str{G'}$ such that if we consider any optimal alignment $\B$ between $F'$ and $G'$ (and identify it with the corresponding alignment of $\str{F'}$ and $\str{G'}$), then $\tilde{\A}$ and $\B$ differ at $\Oh(k^4)$ positions. 
Next, we discuss the key ideas of the periodicity reduction process. We remark that this is the first algorithm that provides an efficient way of reducing periodicity in trees or forests, and it can be of independent interest. 

\paragraph{Periodicity Reduction:}
One of the most important novel contributions of our algorithm is the definition and identification of periodicity in forests. 
An integer $p$ is \emph{period} of a string $S$ if $S[i] = S[i + p]$ holds for all $i \in [0 \dd n-p)$; we then call $\frac{|S|}{p}$ the \emph{exponent} of this period of $S$. If $p \leq k$, then $S$ is called $k$-periodic. Ideally, we would like to identify and reduce the exponent of periodic regions of forests $F$ and $G$ somehow. Given two strings $S$ and $T$, if there exist periodic fragments $Q = S[\alpha_s \dd \beta_s) =  T[\alpha_t \dd \beta_t)$ such that $|\alpha_s - \alpha_t| \leq 2k$, we say that there are $k$-synchronized occurrences of $Q$ in $S$ and $T$. Synchronized periodicity in strings has previously been studied for string edit distance algorithms that take advantage of such periodicity when constructing optimal alignments (e.g., in~\cite{FOCS}), but never in forests. Unfortunately, reducing periodic fragments in the parentheses representations of forests does not preserve tree edit distance. 
In particular, if a periodic substring is unbalanced, then any edit to an opening parenthesis in the periodic region also needs to be applied to the corresponding closing parenthesis, which may lie outside the periodic region. 
We propose and identify two new types of ``balanced'' periodicity in forests, \emph{horizontal} and \emph{vertical} periodicity.  By constructing forests $F'$ and $G'$ which avoid $k$-synchronized horizontal and vertical $k$-periodicity in $F$ and $G$ without affecting the total edit distance between the forests, we prove that we also avoid any periodicity in the underlying parentheses representations of $F$ and $G$ (with an appropriate labeling).


In \cref{sec:hor}, the first type of periodicity we identify and reduce is \emph{horizontal periodicity}. If a sequence of subtrees of a given node repeats periodically, then we consider the repeated forest as a horizontal period. Therefore, horizontal periodicity appears in the parentheses representation of $F$ and $G$ as a balanced periodic fragment. Since horizontal periods are already balanced, any edits can be applied locally within the horizontally periodic fragment. By computing all maximal periodic fragments in the parentheses representations of $F$ and $G$, we can find and reduce all horizontally periodic regions at nearby locations in both forests. In order to avoid interactions between overlapping horizontally periodic regions, we are careful to only reduce horizontal periodicity with large exponent and small period. \cref{alg:SyncReductions} gives detailed pseudocode for the $\tilde{\Oh}(n)$-time construction of forests $F'$ and $G'$ which avoid $k$-synchronized horizontal $k$-periodicity from input forests $F$ and $G$.

In \cref{sec:ver}, the second type of periodicity we identify and reduce is \emph{vertical periodicity}. Vertical periodicity occurs in a forest when there is some path in which the subtrees to the left and right of the path repeat periodically at subsequent levels of the path (see Fig.~\ref{fig:vertical_diverge}). To find all vertical periods in a forest~$F$, we identify nodes whose corresponding opening and closing parenthesis are each contained in periodic substrings of the parentheses representation $\str[\lab]{\F}$. Once we find all vertical periods of $F$ and $G$, we use orthogonal range queries (in two dimensions) to find nodes $u \in V_F$ and $v \in V_G$ which are both contained in a vertically periodic region and whose opening parentheses and closing parentheses are in similar locations of  $\str[\lab]{\F}$ and  $\str[\lab]{\G}$. Once we find all such pairs of nodes, we can reduce the exponent of the periodic path as in the case of horizontal periods. Algorithm \ref{alg:VertSyncReductions} performs these vertical periodic reductions in $\tilde{\Oh}(n)$-time without introducing any new horizontal $k$-periodicity, and it outputs forests $F'$ and $G'$ which avoid both $k$-synchronized horizontal and vertical $k$-periodicity.

In \cref{sec:fur}, we prove that since $F'$ and $G'$ avoid $k$-synchronized horizontal and vertical $k$-periodicity of large exponent, then $\str[\lab]{\F'}$ and  $\str[\lab]{\G'}$ avoids all $k$-synchronized $2k$-periodicity of large exponent, where $\lab$ is an appropriate refinement of the depth-$8k$ look-ahead labeling. We show that any periodic substring of a forest's parentheses representation must be either horizontal or vertical period depending on whether the period string is balanced or unbalanced. At the end of the section, using an $\Oh(n + k^3)$-time DP algorithm described in \cref{sec:prelim}, we compute an alignment $\tilde{\A}$ of $\str[\lab]{F'}$ and $\str[\lab]{G'}$ that differs from any optimal tree alignment of $F'$ and $G'$ on at most $k^4$ characters.

\vspace{3mm}

Given the two trees $F'$ and $G'$ that avoid synchronized horizontal $k$-periodicity and synchronized vertical $k$-periodicity, in Section~\ref{thm:smallh} we design an algorithm using the greedy strategy that computes $\ted_{\le k}(F',G')$ in time $\tOh(n+h^2k^7)$, where $h$ is an upper bound on the height of $F'$ and $G'$. We provide a brief sketch of the algorithm in the following.

\paragraph{Tree Edit Distance of Shallow Forests:}
Here, we first show that if we consider any optimal alignment $\A$ between $F'$ and $G'$, then, for all but $2hk$ pairs of nodes $(u,v)$ where $u\in V_{F'}$ and $v\in V_{G'}$, $\A$ matches $u$ with $v$ and the subtree rooted at $u$ is identical as the one rooted at $v$. 
If we assume that $F'$ and $G'$ avoid synchronized horizontal $k$-periodicity,
this allows constructing a large set of matching pairs $M\sub \{(u,v):u\in V_F, v\in V_G\}$ (of size $|V_F|-\Oh(h^2k^4)$) that is common to all optimal greedy alignments; i.e., each pair $(u,v)\in M$ is matched by every greedy alignment and the subtrees rooted at $u$ and $v$ are at distance $0$. The time required to construct $M$ is $\Oh(n+hk^2)$.

Given this partial matching $M$ common to every optimal greedy alignment, our next objective is to extend it in order to construct an optimal alignment between $F'$ and $G'$. An analogous task is relatively straightforward for strings: given a partial matching, we can first partition the strings along this matching and then independently compute the optimal distance between subsequent pieces. 
This strategy fails for trees, but we still provide a linear-time algorithm that, given a pair of forests $F'$ and $G'$ and a non-crossing partial matching $M$ between them, constructs forests $F''$ and $G''$ such that $\ted^M_{\le k}(F',G')=\ted_{\le k}(F'',G'')$ and $|F''|+|G''|=\Oh(k(|F'|+|G'|-2|M|))$.
Here, $\ted^{M}_{\le k}(F',G')$ is the minimum of the cost of all alignments $\A$ between $F',G'$ such that each pair of nodes in $M$ is also matched by $\A$ (or $\infty$ if that cost exceeds $k$).

Thus, following this construction and using the fact that the partial matching $M$ is common to every optimal greedy alignment between $F',G'$ (thus $\ted^M_{\le k}(F',G')=\ted_{\le k}(F',G')$) and $|M|=|V_F|-\Oh(h^2k^4)$, in linear time we can construct trees $F'',G''$ such that $\ted_{\le k}(F',G')=\ted_{\le k}(F'',G'')$ and $|F''|+|G''|=\Oh(h^2k^5)$. Next using the $\tOh(nk^2)$ time algorithm of~\cite{DBLP:conf/icalp/AkmalJ21}, we compute $\ted_{\le k}(F'',G'')$ in time $\tOh(h^2k^7)$. Below, we discuss how to extend a partial matching to a complete alignment.

\paragraph{Partial Forest Matching:}
As mentioned earlier, here, given a pair of forests $F'$ and $G'$ and a non-crossing partial matching $M$, our objective is to construct forests $F''$ and $G''$ such that $\ted^M_{\le k}(F',G')=\ted_{\le k}(F'',G'')$ and $|F''|+|G''|=\Oh(k(|F'|+|G'|-2|M|))$. Also, if the height of $F''$ (or $G''$) is $h>2$, then there is a length $h-2$ top down-path in $F'$ avoiding nodes from $M$. Thus, if $M$ is large and hits all long paths, then this significantly reduces both the size and the depth of the input forests while preserving their distance. 

We start by partitioning $F'$ and $G'$ along the matching pairs in $M$, creating forests $\hat F$ and $\hat G$, and a set of partial matching pairs $\hat M$ between $\hat F,\hat G$. Our aim is to ensure that each node appearing in a matching pair of $\hat{M}$ is a leaf node of some tree in $\hat F$ (or $\hat G$) and $\ted^M(F',G')=\ted^{\hat M}(\hat F,\hat G)$.
For this, we mark all the nodes that are present in $M$ and assign each node to the nearest marked proper ancestor. This creates a partition of $V_F$ and of $V_G$. Next for each subset $S$ of nodes in the partition, we create a tree by deleting all the nodes in $V_F\setminus S$ from $F$, and we add the tree to $\hat F$ (while preserving the order of the nodes as present in $F$). Similarly create $\hat G$ from $G'$.
Note that, by construction, each node present in $M$ appears as a leaf node of some tree in $\hat F$ (or in $\hat G$).
We also copy all the matching pairs from $M$ to $\hat M$. Lastly to ensure $\ted^M(F',G')=\ted^{\hat M}(\hat F,\hat G)$,
between each pair of consecutive trees in $\hat F$ and their corresponding trees in $\hat G$,
we insert a pair of single-node trees whose roots are also added to $\hat M$.
Now, each node present in $\hat M$ appears as a leaf node of some tree in $\hat F$ (or in $\hat G$) and, if height of $\hat F$ (or $\hat G$) is $h>1$, then there is a length $h-1$ top down path in $F'$ (or $G'$) avoiding nodes from $M$.

In the second step we further reduce $\hat M$ by removing redundant matching pairs. The idea here is that, for each matching pair $(\hat u, \hat v)\in \hat M$, if their immediate left siblings $u$ and $v$ (respectively) are also matched by $\hat M$, i.e., $(u,v)\in \hat M$, then we can discard $(\hat u, \hat v)$ to create a new set $\bar M$ and create forests $\bar F$ and $\bar G$ by deleting all the nodes present in some discarded matching pair.
This is because $(u,v)$ serves as a representative of $(\hat u, \hat v)$ and thus reintroducing them would not violate the non-crossing property of the matching set, and we can ensure $\ted^{\bar M}(\bar F,\bar G)=\ted^{\hat M}(\hat F,\hat G)=\ted^M(F',G')$. We also show $|\hat M|\le 2/5(|\hat F|+|\hat G|+1)$. 

Lastly, using $\bar F$, $\bar G$, and $\bar M$, we construct forests $F'',G''$ such that $\ted^{\bar M}_{\le k}(\bar F,\bar G)=\ted_{\le k}(F'',G'')$.
To construct $F''$ and $G''$, we attach a gadget containing $k+1$ uniquely-labeled children to each node present in $\bar M$.
As the cost of an optimal alignment of $F''$ and $G''$ is bounded by $k$, it should match at least one node from each gadget; following this, we can show the alignment will indeed match all the nodes from the gadget and from set  $\bar M$ thus proving $\ted^{\bar M}_{\le k}(\bar F,\bar G)=\ted_{\le k}(F'',G'')$.
Moreover, from the above construction and using the bound of $|\hat M|$, we can show $|F''|+|G''|=\Oh(\min(|F'|+|G'|+k|M|,k(|F'|+|G'|-2|M|)))$.
Also, by the gadget construction, we ensure the height of $F''$ (or $G''$) is at most the height of $\bar F$ (or $\bar G$) plus one. Thus, if the height of $F''$ (or $G''$) is $h>2$, then there is a length $h-2$ top down path in $F'$ (or $G'$, respectively) avoiding $M$. We provide the details in Section~\ref{lem:partial}.

\vspace{3mm}

To summarize, given arbitrary labelled forests $F,G$ and a threshold $k$, we first design an algorithm that in time $\tOh(n)$ produces a pair of forests $F',G'$ such that $\ted_{\le k}(F,G)=\ted_{\le k}(F',G')$ and $F',G'$ avoid synchronized horizontal and vertical $k$-periodicity. 
At $\tOh(n+k^3)$ time, this algorithm also computes an alignment $\tilde{\A}$ between $\str{F'}$ and $\str{G'}$ such that, any optimal alignment $\B$ between $F'$ and $G'$, differs from $\tilde{\A}$ at $\Oh(k^4)$ positions. Next, using the greedy strategy and the partial matching technique, we design an algorithm that can compute $\ted_{\le k} (F',G')$ in time $\tOh(n+h^2k^7)$, given that the heights of $F'$ and $G'$ is at most $h$.
However, the challenge is that the height $h$ can be much larger than $k$. To solve this issue, instead of directly computing the distance between $F',G'$, we design \cref{sec:sampl}, which first creates a partial matching $\tilde{M}$ using random sampling within alignment $\tilde{\A}$, and then extends this partial matching to construct an optimal alignment between $F'$ and $G'$. The random sampling constructs $\tilde{M}$ in such a way so that, given $F', G'$, and $\tilde{M}$, the partial matching reduction generates forests $F'', G''$ where $\ted_{\le k}^{\tilde{M}}(F',G')=\ted_{\le k}(F'',G'')$  and the heights of $F''$ and $G''$ are bounded by $O(k^4)$. Thus, one can compute the distance in time $\tOh(n+k^{15})$. We now give a sketch of the sampling technique.

\paragraph{Level Sampling}
Set the height threshold $h=O(k^4)$. Select the sampling parameter $r$ from $[0 \dd h)$ uniformly at random. Mark all the nodes in $F'$ and $G'$ at depth congruent to $r$ modulo $h$. Next, we create a partial matching $\tilde{M}$ as follows: for each marked node $u\in V_{F'}$ (or in $V_{G'}$), add $(u,v)$ to $\tilde{M}$, where $v\in V_{G'}$ (or $v\in V_{F'}$) and $\tilde{\A}$ matches the parentheses representing $u$ with the parentheses representing $v$. Apply partial matching reduction on $F'$, $G'$ and $\tilde{M}$ to create forests $F'', G''$ given (i) every marked node is present in $\tilde{M}$ (ii) $|\tilde{M}|=O(n/k^4)$. Using the bound on $|\tilde{M}|$, we claim $|F''|+|G''|=O(n)$. Moreover, neither of $F'',G''$ has a top-down length-$h$ path avoiding $\tilde{M}$. Hence, the height of $F'',G''$ is at most $h+1=O(k^4)$. 
Thus, we can compute $\ted_{\le k}(F'',G'')$ using the shallow forest algorithm in time $\tOh(n+\poly(k))$. Note that the partial matching reduction ensures $\ted_{\le k}^{\tilde{M}}(F',G')=\ted_{\le k}(F'',G'')$. As $\ted (F',G')\le k$ and any optimal alignment between $F'$ and $G'$ and $\tilde{\A}$ differs in at most $O(k^4)$ matching pairs, following our sampling strategy we further can argue that $\ted_{\le k}(F',G')=\ted_{\le k}^{\tilde{M}}(F',G')$ holds with constant probability. To ensure concentration, we repeat this $\Theta(\log n)$ times and report the minimum distance computed.

%% file: figs/nogreedy.tex
\tikzset{every picture/.style={line width=0.75pt}} 

\begin{tikzpicture}[x=0.75pt,y=0.75pt,yscale=-1,xscale=1]

\draw  [color={rgb, 255:red, 155; green, 155; blue, 155 }  ,draw opacity=1 ] (173.08,57.25) -- (187.77,101.14) -- (158.4,101.14) -- cycle ;
\draw  [color={rgb, 255:red, 155; green, 155; blue, 155 }  ,draw opacity=1 ] (217.53,57.88) -- (240.38,128.01) -- (194.69,128.01) -- cycle ;
\draw  [color={rgb, 255:red, 155; green, 155; blue, 155 }  ,draw opacity=1 ][fill={rgb, 255:red, 255; green, 255; blue, 255 }  ,fill opacity=1 ] (163.42,57.25) .. controls (163.42,52.09) and (167.63,47.9) .. (172.81,47.9) .. controls (178,47.9) and (182.2,52.09) .. (182.2,57.25) .. controls (182.2,62.41) and (178,66.59) .. (172.81,66.59) .. controls (167.63,66.59) and (163.42,62.41) .. (163.42,57.25) -- cycle ;
\draw  [color={rgb, 255:red, 155; green, 155; blue, 155 }  ,draw opacity=1 ][fill={rgb, 255:red, 255; green, 255; blue, 255 }  ,fill opacity=1 ] (208.14,57.88) .. controls (208.14,52.72) and (212.35,48.53) .. (217.53,48.53) .. controls (222.72,48.53) and (226.92,52.72) .. (226.92,57.88) .. controls (226.92,63.04) and (222.72,67.23) .. (217.53,67.23) .. controls (212.35,67.23) and (208.14,63.04) .. (208.14,57.88) -- cycle ;
\draw  [color={rgb, 255:red, 155; green, 155; blue, 155 }  ,draw opacity=1 ] (367.93,59.4) -- (390.78,129.53) -- (345.09,129.53) -- cycle ;
\draw  [color={rgb, 255:red, 155; green, 155; blue, 155 }  ,draw opacity=1 ][fill={rgb, 255:red, 255; green, 255; blue, 255 }  ,fill opacity=1 ] (358.54,59.4) .. controls (358.54,54.24) and (362.75,50.05) .. (367.93,50.05) .. controls (373.12,50.05) and (377.32,54.24) .. (377.32,59.4) .. controls (377.32,64.56) and (373.12,68.75) .. (367.93,68.75) .. controls (362.75,68.75) and (358.54,64.56) .. (358.54,59.4) -- cycle ;
\draw [color={rgb, 255:red, 245; green, 166; blue, 35 }  ,draw opacity=1 ]   (180.4,51.14) .. controls (180.6,48.52) and (181.84,47.29) .. (184.12,47.46) .. controls (186.35,47.76) and (187.79,46.62) .. (188.44,44.04) .. controls (188.88,41.73) and (190.26,40.82) .. (192.59,41.33) .. controls (194.88,41.94) and (196.13,41.24) .. (196.34,39.23) .. controls (197.19,36.93) and (198.8,36.15) .. (201.16,36.88) .. controls (203.49,37.69) and (204.92,37.08) .. (205.43,35.07) .. controls (206.64,32.83) and (208.44,32.17) .. (210.83,33.09) .. controls (212.59,34.26) and (214.16,33.76) .. (215.54,31.59) .. controls (216.38,29.61) and (218.01,29.15) .. (220.43,30.22) .. controls (222.16,31.5) and (223.84,31.09) .. (225.47,28.99) .. controls (226.49,27.06) and (227.87,26.76) .. (229.6,28.11) .. controls (232,29.34) and (233.76,29.01) .. (234.87,27.12) .. controls (236.04,25.24) and (237.83,24.96) .. (240.26,26.27) .. controls (241.93,27.72) and (243.39,27.53) .. (244.64,25.69) .. controls (245.93,23.86) and (247.77,23.66) .. (250.17,25.09) .. controls (251.8,26.61) and (253.3,26.48) .. (254.65,24.7) .. controls (256.04,22.93) and (257.54,22.83) .. (259.16,24.4) .. controls (261.49,25.95) and (263.38,25.87) .. (264.82,24.15) .. controls (266.29,22.44) and (267.81,22.4) .. (269.36,24.04) .. controls (271.63,25.69) and (273.53,25.69) .. (275.05,24.03) .. controls (276.6,22.38) and (278.11,22.41) .. (279.59,24.12) .. controls (281.03,25.83) and (282.91,25.91) .. (285.23,24.35) .. controls (286.84,22.77) and (288.33,22.86) .. (289.72,24.63) .. controls (291.07,26.4) and (292.56,26.52) .. (294.17,24.99) .. controls (296.56,23.56) and (298.39,23.75) .. (299.67,25.57) .. controls (300.91,27.4) and (302.72,27.64) .. (305.09,26.28) .. controls (306.77,24.85) and (308.19,25.07) .. (309.34,26.94) .. controls (311.13,28.95) and (312.87,29.27) .. (314.56,27.89) .. controls (316.27,26.54) and (317.96,26.9) .. (319.64,28.97) .. controls (320.57,30.9) and (322.22,31.31) .. (324.57,30.19) .. controls (326.3,28.93) and (327.89,29.38) .. (329.33,31.53) .. controls (330.06,33.5) and (331.59,34) .. (333.91,33.01) .. controls (335.66,31.86) and (337.12,32.4) .. (338.28,34.62) .. controls (339.31,36.84) and (340.97,37.54) .. (343.24,36.73) .. controls (345.55,35.98) and (346.85,36.61) .. (347.12,38.62) .. controls (347.76,40.87) and (349.19,41.69) .. (351.42,41.07) .. controls (354.14,40.86) and (355.65,41.9) .. (355.95,44.17) .. controls (356.05,46.38) and (357.36,47.5) .. (359.88,47.52) -- (362.75,50.59) ;
\draw [color={rgb, 255:red, 74; green, 144; blue, 226 }  ,draw opacity=1 ] [dash pattern={on 4.5pt off 4.5pt}]  (226,62.34) .. controls (264.8,81.54) and (318.7,80.9) .. (361.5,63.7) ;
\draw   (156.25,132.59) .. controls (156.26,137.26) and (158.59,139.59) .. (163.26,139.59) -- (190.88,139.56) .. controls (197.55,139.55) and (200.88,141.88) .. (200.89,146.55) .. controls (200.88,141.88) and (204.21,139.55) .. (210.88,139.54)(207.88,139.54) -- (238.51,139.51) .. controls (243.18,139.5) and (245.51,137.17) .. (245.5,132.5) ;
\draw   (344.75,135.59) .. controls (344.75,140.26) and (347.08,142.59) .. (351.75,142.59) -- (358.75,142.59) .. controls (365.42,142.59) and (368.75,144.92) .. (368.75,149.59) .. controls (368.75,144.92) and (372.08,142.59) .. (378.75,142.59)(375.75,142.59) -- (383.75,142.59) .. controls (388.42,142.59) and (390.75,140.26) .. (390.75,135.59) ;

\draw (204.4,96.6) node [anchor=north west][inner sep=0.75pt]  [font=\footnotesize]  {$ \begin{array}{l}
T_{2}\\
\end{array}$};
\draw (160.4,78.2) node [anchor=north west][inner sep=0.75pt]  [font=\footnotesize]  {$ \begin{array}{l}
T_{1}\\
\end{array}$};
\draw (166.4,52.6) node [anchor=north west][inner sep=0.75pt]    {$a$};
\draw (212.3,50.2) node [anchor=north west][inner sep=0.75pt]    {$b$};
\draw (353.8,98.12) node [anchor=north west][inner sep=0.75pt]  [font=\footnotesize]  {$ \begin{array}{l}
T_{2}\\
\end{array}$};
\draw (362.7,54.72) node [anchor=north west][inner sep=0.75pt]    {$a$};
\draw (268,86.6) node [anchor=north west][inner sep=0.75pt]  [color={rgb, 255:red, 74; green, 144; blue, 226 }  ,opacity=1 ] [align=left] {{\footnotesize \textcolor[rgb]{0.29,0.56,0.89}{OPT}}};
\draw (193.5,151.5) node [anchor=north west][inner sep=0.75pt]    {${F}$};
\draw (358.5,152.84) node [anchor=north west][inner sep=0.75pt]    {${G}$};

\end{tikzpicture}

%% file: src/prelim.tex
\subsection{Strings}
A \emph{string} $Y\in \Sigma^n$ is a sequence of $|Y|:=n$ characters from an \emph{alphabet} $\Sigma$.
For $i\in [0\dd n)$, we denote the $i$th character of $Y$ with $Y[i]$. The \emph{reverse} of a string $Y$ is $\rev{Y}:=Y[n-1] Y[n-2] \cdots Y[0]$.
We say that a string $X$ \emph{occurs} as a \emph{substring} of a string $Y$ if $X=Y[i]\cdots Y[j-1]$ holds for some indices $0\le i \le j \le |Y|$.
We denote the underlying \emph{occurrence} of $X$ as $Y[i\dd j)$.
Formally, $Y[i\dd j)$ is a \emph{fragment} of $Y$ that can be represented using a reference to $Y$ as well as its endpoints $i,j$.
The fragment $Y[i\dd j)$ can be alternatively denoted as $Y[i\dd j-1]$, $Y(i-1\dd j-1]$, or $Y(i-1\dd j)$.
A fragment of the form  $Y[0\dd j)$ is a \emph{prefix} of $Y$, whereas a fragment of the form $Y[i\dd n)$ is a \emph{suffix} of $Y$.

An integer $p\in [1\dd n]$ is a \emph{period} of a string $Y\in \Sigma^n$ if $Y[i]=Y[i+p]$ holds for all $i\in [0\dd n-p)$.
In this case, the prefix $Y[0\dd p)$ is called a \emph{string period} of $Y$.
By $\per(Y)$ we denote the smallest period of $Y$. 
The exponent of a string $Y$ is defined as $\exp(Y):=\frac{|Y|}{\per(Y)}$,
and we say that a string $Y$ is \emph{periodic} if $\exp(Y)\ge 2$.

For a string $Y$ and an integer $m\ge 0$, we define the $m$th power of $Y$, denoted $Y^m$, as the concatenation of $m$ copies of~$Y$. 
For a string $Y\in \Sigma^n$, we define a \emph{forward rotation} $\rot(Y)=Y[1] \cdots Y[n-1]Y[0]$.
In general, a \emph{cyclic rotation} $\rot^s(Y)$ with \emph{shift} $s\in \mathbb{Z}$ is obtained by iterating $\rot$ or the inverse operation $\rot^{-1}$.
A non-empty string $Y\in \Sigma^n$ is \emph{primitive} if it is distinct from its non-trivial rotations, i.e., if $Y=\rot^s(Y)$ holds only when $s$ is a multiple of $n$.
A string $Y$ is primitive if and only if it cannot be expressed as $Y=X^m$ for some string $X$ and integer $m>1$.

\subsection{Edit-distance Alignments}
The \emph{edit distance} (also known as the \emph{Levenshtein distance}) $\ed(X, Y)$ between two strings $X$ and $Y$ is defined as the smallest number of character insertions, deletions, and substitutions required to transform $X$ to $Y$. 

Equivalently, the edit distance $\ed(X,Y)$ can be defined as the cost of the cheapest \emph{alignment} $\A : X \onto Y$.
There are many ways to formally define an alignment; all of them, however, need to identify which characters of $X$
are \emph{deleted} and which characters of $Y$ are \emph{inserted}. The remaining characters are \emph{aligned} 
(either \emph{substituted} or \emph{matched}) in the left-to-right order.
Since parts of this work rely on the techniques of~\cite{FOCS}, we chose to stick to their formalization.
\begin{definition}
  A sequence $\A = (x_t,y_t)_{t=0}^m$ is an \emph{alignment}
  of a string $X\in \Sigma^*$ onto a string $Y\in\Sigma^*$, denoted $\A : X \onto Y$, if $(x_0,y_0)=(0,0)$,
  $(x_m,y_m)=(|X|,|Y|)$, and $(x_{t+1},y_{t+1})\in \{(x_t+1,y_t+1),(x_t+1,y_t),(x_t,y_t+1)\}$ for $t\in [0\dd m)$.
\end{definition}
Given an alignment $\A = (x_t,y_t)_{t=0}^m : X \onto Y$, for every $t\in [0\dd m)$:
\begin{itemize}
  \item If $(x_{t+1},y_{t+1})=(x_t+1,y_t)$, we say that $\A$ \emph{deletes} $X[x_t]$.
  \item If $(x_{t+1},y_{t+1})=(x_t,y_t+1)$, we say that $\A$ \emph{inserts} $Y[y_t]$.
  \item If $(x_{t+1},y_{t+1})=(x_t+1,y_t+1)$, we say that $\A$ \emph{aligns} $X[x_t]$ and $Y[y_t]$, denoted $X[x_t] \sim_\A Y[y_t]$. If~additionally $X[x_t]= Y[y_t]$, we say that $\A$ \emph{matches} $X[x_t]$ and $Y[y_t]$, denoted $X[x_t] \simeq_\A Y[y_t]$.
  Otherwise, we say that $\A$ \emph{substitutes} $X[x_t]$ for $Y[y_t]$.
\end{itemize}

The \emph{cost} of an edit distance alignment $\A$ is the total number characters that $\A$ deletes, inserts, or substitutes.
We denote the cost by $\ed_\A(X,Y)$.
The cost of an alignment $\A=(x_t,y_t)_{t=0}^m$ is at least its \emph{width}, defined as $\max_{t=0}^m |x_t-y_t|$.
Observe that $\ed(X,Y)$ can be defined as the minimum cost of an alignment of $X$ and~$Y$,
that is, $\ed(X,Y) = \min_{\A : X \onto Y} \ed_{\A}(X,Y)$.
An alignment of $X$ and $Y$ is \emph{optimal} if its cost is equal to $\ed(X, Y)$.

Given an alignment $\A=(x_t,y_t)_{t=0}^m$ of $X,Y\in \Sigma^+$, we partition the elements $(x_t,y_t)$ of $\A$
into \emph{matches} (for which $X[x_t]\simeq_{\A} Y[y_t]$) and \emph{breakpoints} (the remaining elements).
We denote the set of matches and breakpoints by $\mtch_{\A}(X,Y)$ and $\brkp_{\A}(X,Y)$, respectively.
Observe that $|\brkp_{\A}(X,Y)|=1+\ed_\A(X,Y)$.

We say that a set $M\sub \Zz\times \Zz$ is \emph{non-crossing} if 
there are no distinct pairs $(x,y),(x',y')\in M$ with $x\le x'$ and $y\ge y'$.
Observe that, for every alignment $\A$ of $X,Y$, the set $\{(x,y)\in [0\dd |X|)\times [0\dd |Y|): X[x]\sim_{\A} Y[y]\}$
is non-crossing.
We further say that a non-crossing set $M\sub [0\dd |X|)\times [0\dd |Y|)$ 
is a \emph{non-crossing matching} of strings $X,Y$ if $X[x]=Y[y]$ holds for every $(x,y)\in M$.
Note that, for every alignment $\A$ of $X,Y$, the set $\mtch_{\A}(X,Y)$ is a non-crossing matching.

Given an alignment $\A= (x_t,y_t)_{t=0}^m$ of $X$ and $Y$, for every $\ell,r\in [0\dd m]$ with $\ell\le r$,
we say that $\A$ \emph{aligns} $X[x_\ell\dd x_{r})$ to $Y[y_\ell\dd y_{r})$, denoted $X[x_\ell\dd x_{r})\sim_\A Y[y_{\ell}\dd y_{r})$. If there is no breakpoint $(x_t,y_t)$ with $t\in [\ell\dd r)$,
we further say that $\A$ \emph{matches} $X[x_\ell\dd x_{r})$ to $Y[y_\ell\dd y_{r})$, denoted $X[x_\ell\dd x_{r})\simeq_\A Y[y_{\ell}\dd y_{r})$. 

\begin{definition}[Greedy alignment~\cite{FOCS}]\label{def:greedy}
    We say that an alignment $\A$ of two strings $X, Y\in \Sigma^*$ is \emph{greedy} if  $X[x] \ne Y[y]$ holds for every $(x,y)\in \brkp_{\A}(X,Y)$
    such that $x\ne |X|$ and $y\ne |Y|$.
\end{definition}
As proven in~\cite{FOCS}, any two strings $X,Y\in \Sigma^*$ have an optimal alignment that is also greedy.

For strings $X,Y\in \Sigma^*$ and integers $k\ge w\ge 0$, we denote by $\aa_{k,w}(X,Y)$ the family of alignments of $\A:X\onto Y$
whose cost is at most $k$ and whose width is at most $w$. Furthermore, we define $\ga_{k,w}(X,Y)$ as the set consisting of 
all greedy alignments in $\aa_{k,w}(X,Y)$.
The following result can be obtained using a straightforward adaptation of the Landau--Vishkin algorithm~\cite{LV88}.
The only difference compared to the baseline implementation is that the DP table
is artificially restricted to diagonals $[-w\dd w]$.

\begin{lemma}\label{lem:build_greedy}
Given strings $X,Y\in\Sigma^{\le n}$ and integers $k,w\in \Zz$,
one can in $\Oh(n+kw)$ time decide whether $\aa_{k,w}(X,Y)\ne \emptyset$.
If the answer is positive, the algorithm reports a witness alignment $\A \in \ga_{k,w}(X,Y)$.
\end{lemma} 

\subsection{Trees and Forests}\label{sec:pr2}

\begin{definition}[Forest]
We say that a forest $\F$ is a (possibly empty) sequence of non-empty rooted ordered trees $\T_1,\ldots,\T_m$. Formally, each such tree $\T_i$ consists of a root node $v_i$ and a forest $\F_i$ representing the descendants of $v_i$.
We write $V_\F$ to denote the node set of $\F$.
\end{definition}

For a node $v\in V_{\F}$ in a forest $\F$, we denote the parent of $v$ by $\pr_\F(v)$.
If $v$ is a root, then we set $\pr_\F(v)=\bot_{\F}$; consistently with~\cite{Xiao21,DBLP:conf/icalp/AkmalJ21},
we refer to $\bot_{\F}$ as the \emph{virtual root} of $\F$. We denote $\vV_\F = V_\F \cup \{\bot_{\F}\}$
as the set of nodes in $\F$ including the virtual root.
For each node $v\in \vV_\F$, we denote the sequence of children of $v$ by $\chld_\F(v)$.
Observe that the forest $\F$ is uniquely specified by the lists $\chld_\F(v)$ for $v\in \vV_\F$
and these lists, in total, contain every node $u\in V_\F$ exactly once.

We define the \emph{height} of a forest $\F$, denoted $\hgt(\F)$, as the maximum number of nodes on the root-to-leaf path in $\F$.
Thus, an empty forest has height $0$, whereas a forest consisting of roots only has height $1$.

\begin{definition}[Labelled forest]
A \emph{labeled forest} $F$ consists of a forest $\F$ and a labeling function $\lab : V_{\F}\to \Sigma$
assigning labels (from an integer \emph{alphabet} $\Sigma=[0\dd \sigma)$) to nodes of $\F$.
\end{definition}

In this work, we typically consider two labeled forests $F,G$.
We then assume that the underlying unlabeled forests $\F,\G$ have disjoint node sets
and a joint labeling $\lab : V_{\F}\cup V_{\G}\to \Sigma$.

\begin{definition}[Tree edit distance]\label{def:treen}
The following operations are jointly called \emph{edits} of a labeled forest $F$:
\begin{description}
\item[Substitution (relabeling)] Change the label of a node in $F$ with another symbol from $\Sigma$.
\item[Deletion] Delete a node $v$ from $F$. As a result, the children of $v$ become the children of $\pr_F(v)$ while preserving their order.
Formally, we modify $\chld_F(\pr_F(v))$ by replacing  $v$ with $\chld_F(v)$.
\item[Insertion] Insert a labeled node $v$ into $F$ so that deletion of $v$ produces $F$ again.
 Formally, an insertion is specified by a $u\in \vV_\F$ (which will become the parent of $v$), a (possibly empty) fragment of $\chld_F(u)$ (identifying nodes which will become the children of $v$), and the label of the new node.
\end{description}
The tree edit distance $\ted(F,G)$ is defined as the minimum number of edits (listed above) required to transform one forest to the other.
\end{definition}

Just like the string edit distance, the tree edit distance admits an equivalent definition 
in terms of \emph{alignments} $\A : F \onto G$. 
Such an alignment must specify which nodes in $F$ are \emph{deleted} and which nodes in $G$ are \emph{inserted}.
The remaining nodes are \emph{aligned} (either \emph{relabeled} or \emph{matched}).
Unlike for string edit distance, where the alignment must only be consistent with the left-to-right ordering,
here, the alignment must the consistent with both the pre-order and the post-order forest traversal.

\subsection{Parenthesis Representation of Forests}
Aiming to adapt string techniques into the forest setting, we map each forest to a string through 
the \emph{parentheses representation} (closely related to the \emph{bi-order traversal} considered in~\cite{Xiao21}).
\begin{definition}[Parenthesis representation]
    For every unlabeled forest $\F$, we define the \emph{parentheses representation} of $\F$,
    denoted $\str{\F}\in \{\op,\cl\}^*$, using the following recursion:
    If $\F$ consists of trees $\T_1,\ldots,\T_m$, where each $\T_i$ consists of a root node $v_i$ and a forest $\F_i$ representing the descendants of $v_i$,
    then $\str{\F}=\bigodot_{i=1}^m \str{\T_i}$, where $\str{\T_i} = \op \cdot \str{\F_i} \cdot \cl$ and $\cdot$ as well as $\bigodot$ denote concatenation.

    For a labelling $\lab : V_\F \to \Sigma$, we further set $\str[\lab]{\F}\in (\{\op,\cl\}\times \Sigma)^*$ 
    so that $\str[\lab]{\F}=\bigodot_{i=1}^m \str[\lab]{\T_i}$
    and $\str[\lab]{\T_i} = \op_{\lab(v_i)} \cdot \str[\lab]{\F_i} \cdot \cl_{\lab(v_i)}$.
    For a labeled forest $F=(\F,\lab)$, we also write $\str{F}:=\str[\lab]{\F}$.
\end{definition}

Observe that $\str{\F}$ forms a well-parenthesized sequence and there is a bijection
between nodes of $\F$ and the pairs of matching parentheses in $\str{\F}$.
We define functions  $o_\F,c_\F : V_\F \to [0\dd 2|\F|)$
so that the opening and the closing parenthesis representing $u$
are located at positions $o_\F(u)$ and $c_\F(u)$, respectively.

\begin{definition}\label{def:ta}
    We say that $\A : \str{\F} \onto \str{\G}$ is a \emph{tree alignment} of forests $\F$ and $\G$, denoted $\A : \F \onto \G$
    if the following \emph{consistency} conditions are satisfied for each $u\in V_\F$:
    \begin{itemize}
        \item either $\A$ deletes both $\str{\F}[o_\F(u)]$ and $\str{\F}[c_\F(u)]$, or
        \item there exists $v\in V_\G$ such that $\str{\F}[o_\F(u)] \sim_\A \str{\G}[o_\G(v)]$ and $\str{\F}[c_\F(u)] \sim_\A \str{\G}[c_\G(v)]$.
    \end{itemize}
    We denote the set of all tree alignments $\A: \F \to \G$ with $\ta(\F,\G) \sub \aa(\str{\F},\str{\G})$.
\end{definition}


\begin{definition}\label{def:treep}
Consider a joint labeling $\lab$ of forests $\F,\G$ and an alignment $A\in \ta(\F,\G)$. 
We denote
$\ted_{\lab,\A}(\F,\G) := \frac12 \ed_{\A}(\str[\lab]{\F},\str[\lab]{\G})$.
In terms of the underlying labeled forests $F,G$, we express this as
$\ted_\A(F,G) := \frac12 \ed_{\A}(\str{F},\str{G})$.
\end{definition}

We observe that \cref{def:treen,def:treep} are equivalent as each tree edit operation can be represented using two string edit operations. For example, 
a deletion of a node $v\in V_F$ corresponds to deleting the characters at positions $o_F(u)$ and $c_F(u)$ in $P(F)$. 
\begin{observation}
For any two forests $F,G$, we have $\ted(F,G)=\min_{\A \in \ta(F,G)} \ted_\A(F,G)$.
\end{observation}

Furthermore, for an integer $k\in \Zz$, we denote $\ta_k(F,G)=\{\A \in \ta(F,G) : \ted_{\A}(F,G)\le k\}$.
We also denote
\[\ted_{\le k}(F,G) = \begin{cases}
    \ted(F,G) & \text{if }\ted(F,G)\le k,\\
    \infty & \text{otherwise.}
\end{cases}
\]

%% file: src/refinements.tex
Consider two joint labelings $\lab,\lab' : V_{\F}\cup V_{\G} \to \Sigma$ of forests $\F,\G$.
We say that $\lab$ is a \emph{refinement} of $\lab'$ if $\lab'(u)=\lab'(v)$ holds for any two nodes $u,v\in V_\F\cup V_\G$
such that $\lab(u)=\lab(v)$.
If simultaneously $\lab$ is a refinement of $\lab'$ and $\lab'$ is a refinement of $\lab$, we say that $\lab$ is equivalent to $\lab'$.

\subsection{Look-Ahead Refinement}\label{sec:lookahead}

For a node $v$ of an unlabeled forest $\F$, let $\Sub(v)$ denote the subtree of $\F$ rooted at $v$.
Further, for an integer $d\in\Zp$, let $\Sub_{<d}(v)$ denote the subtree of $\Sub(v)$ consisting of nodes at distance less than $d$ from the root $v$.

\newcommand{\look}[2]{\ensuremath{\mathsf{L}(#1,#2)}}

\begin{definition}[Look-ahead refinement]
Let $\lab$ be a joint labeling of forests $\F,\G$.
We say that $\lab'$ is a \emph{depth-$d$ look-ahead refinement} of $\lab$
if, for any $u,v\in V_{\F}\cup V_{\G}$, we have $\lab'(u)=\lab'(v)$
if and only if $\str[\lab]{\Sub_{<d}(u)}=\str[\lab]{\Sub_{<d}(v)}$.
We use $\look{\lab}{d}$ to denote a depth-$d$ look-ahead refinement of $\lab$
(chosen arbitrarily up to equivalence of labelings).
\end{definition}

\begin{lemma}\label{lem:boundLookahead}
    Consider forests $\F,\G$, their joint labeling $\lab$, an alignment $\A \in \ta(\F,\G)$,
    and an integer $d\in \Zp$.
    Then, \[\ted_{\look{\lab}{d},\A}(\F,\G)\le d\cdot \ted_{\lab,\A}(\F,\G).\]
\end{lemma}
\begin{proof}
Let the \emph{level} of a node in $V_\F\cup V_\G$ denote its distance to the root.
We \emph{mark} some nodes of $\F,\G$ and prove that 
$\ted_{\look{\lab}{d},\A}(\F,\G)$ is bounded by the number of marks.
\begin{enumerate}
    \item If $\A$ deletes a node $u\in V_\F$ at level $\ell$
    or aligns such a node to a node $v\in V_\G$ with $\lab(u)\ne \lab(v)$, then we mark all ancestors of $u$ at levels in $(\ell-d\dd \ell]$.
    \item If $\A$ inserts a node $v\in V_\G$ at level $\ell$, then we mark all ancestors of $v$ at levels in $(\ell-d\dd \ell]$.
\end{enumerate}
Clearly, the total number of marks does not exceed $d\cdot \ted_{\lab,\A}(\F,\G)$.

It remains to prove that each unit of $\ted_{\look{\lab}{d},\A}(\F,\G)$ can be charged to a marked node.
If $\A$ deletes a node in $V_\F\cup V_\G$, then this node is already marked.
We shall prove that if $\A$ aligns a node $u\in V_\F$ with a node $v\in V_\G$ such that $\look{\lab}{d}(u)\ne \look{\lab}{d}(v)$, then either $u$ or $v$ is marked.
For a proof by contradiction, suppose that this is not the case.
In particular, this means that $\A$ does not delete nor insert any node in $\Sub_{<d}(u)$ and $\Sub_{<d}(v)$.
Consequently, these trees are isomorphic, i.e., $\str{\Sub_{<d}(u)}= \str{\Sub_{<d}(v)}$,
and $\A$ aligns nodes according to this isomorphism.
Moreover, $\lab$ assigns the same label to any two nodes matched by this isomporhism (otherwise, $u$ would have been marked), and thus $\str[\lab]{\Sub_{<d}(u)}= \str[\lab]{\Sub_{<d}(v)}$,
contradicting the assumption that $\look{\lab}{d}(u)\ne \look{\lab}{d}(v)$.
\end{proof}

\begin{lemma}\label{lem:buildLookahead}
Given forests $\F,\G$ of total size $n$, their joint labeling $\lab$, and an integer $d\in \Zz$, a labeling equivalent to $\look{\lab}{d}$
can be constructed in $\Oh(n)$ using a randomized algorithm correct with high probability.
\end{lemma}
\begin{proof}
We replace the label of each node $v\in V_{\F}\cup V_{\G}$ with the Karp--Rabin fingerprint~\cite{KR87} of $\str[\lab]{\Sub_{<d}(v)}$.
Choosing a sufficiently large prime number $p = n^{\Theta(1)}$ and a uniformly random seed $r\in [0\dd p)$,
we can guarantee that, with high probability, the fingerprints will have no collisions among the $n$ strings  $\str[\lab]{\Sub_{<d}(v)}$.

As far as the implementation is concerned, for each node $v$ at level $\ell$, we first identify all its descendants $v_1,\ldots,v_m$ at level $\ell+d$.
Such lists can be produced in linear time by a traversal of $F$ and $G$ that maintains the ancestors of the currently visited node 
in an array indexed by the level of the ancestor. When visiting a node at level $\ell\ge d$, we add it to the list of descendants of the ancestor at level $\ell-d$.
Next, we note that $\str[\lab]{\Sub_{<d}(v)}$ can be obtained from  $\str[\lab]{\Sub(v)}$ by cutting out the fragments representing $\str[\lab]{\Sub(v_i)}$. Thus, given the list $v_1,\ldots,v_m$, we can represent $\str[\lab]{\Sub_{<d}(v)}$ as the concatenation of $m+1$ fragments of $\str[\lab]{\Sub(v)}$.
The fingerprint of each fragment of $\str[\lab]{F}$ and $\str[\lab]{G}$ can be constructed in $\Oh(1)$ time and the Karp--Rabin fingerprints support concatenations in $\Oh(1)$ time.
Consequently, the desired fingerprints can be constructed in $\Oh(n)$ time in total.
\end{proof}

Next, we show that there exists an optimum alignment that 
becomes greedy if we refine the labeling using look-ahead $d\ge \max\{\hgt(F),\hgt(G)\}$.

\newcommand{\bx}{\bar{x}}
\newcommand{\by}{\bar{y}}

\begin{proposition}\label{lem:optIsGreedy}
    Let $\lab$ be a joint labeling of unlabeled forests $\F,\G$ of height at most $h$
    and let $\hlab = \look{\lab}{h}$.
    Then, there exists a tree alignment $\A \in \ta(\F,\G)$ of optimum cost $\ted_{\lab,\A}(\F,\G)=\ted_{\lab}(\F,\G)$ such that $\A \in \ga(\str[\hlab]{\F},\str[\hlab]{\G})$.
    \end{proposition}
    
    \begin{proof}
    Consider a tree alignment $\A\in \ta(\F,\G)$ of optimum cost $\ted_{\lab,\A}(\F,\G)=\ted_{\lab}(\F,\G)$.
    We perform the following steps repeatedly for $2|\F|$ rounds: at round $i$, we construct a tree alignment $\A_i\in \ta(\F,\G)$  such that $\ted_{\lambda,\A_i}(F,G)=\ted_{\lambda,\A}(F,G)$ and $\str[\hlab]{\F}[x]\neq \str[\hlab]{\G}[y]$ for every $(x,y)\in \brkp_{\A_i}(\str[\hlab]{\F},\str[\hlab]{\G}) \cap ([0\dd i)\times [0\dd 2|\G|))$. Trivially, $\A_0 := \A$ satisfies this condition for $i=0$.
    
    \vspace{3mm}
    \noindent
    \textbf{Round $i$.}
    At round $i$, we construct $\A_{i}$ given $\A_{i-1}= (x_t, y_t)_{t\in [0\dd m]}$. If $\str[\hlab]{\F}[x]\neq \str[\hlab]{\G}[y]$ already holds for every $(x,y)\in \brkp_{\A_{i-1}}(\str[\hlab]{\F},\str[\hlab]{\G}) \cap ([0\dd i)\times [0\dd 2|\G|))$, then we simply set $\A_{i} := \A_{i-1}$.
    Otherwise, let $(x,y)$ be the leftmost breakpoint with $\str[\hlab]{\F}[x]= \str[\hlab]{\G}[y]$.
    By the assumptions on $\A_{i-1}$, we must have $x=i-1$.

    \vspace{2mm}
    \noindent
    \textbf{Case 1.}
    First, we consider the case where $(x,y)=(o_\F(u),o_\G(v))$ for some $(u,v)\in V_\F\times V_\G$.
    Let $(\bx,\by)=(c_\F(u),c_\G(v))$ and note that, due to $\str[\hlab]{\F}[x]= \str[\hlab]{\G}[y]$ and $\hlab = \look{\lab}{h}$,
    we have $\str[\hlab]{\F}[x\dd \bx] = \str[\hlab]{\G}[y\dd \by]$.
    We define $q:=\min\{t : x_t > \bx \text{ and } y_t > \by \}$ and create $\A_i$ so that it:
    \begin{itemize}
        \item aligns $\str[\hlab]{\F}[0\dd x)$ against $\str[\hlab]{\G}[0\dd y)$ in the same way as $\A_{i-1}$ did,
        \item matches $\str[\hlab]{\F}[x\dd \bx]$ against $\str[\hlab]{\G}[y\dd \by]$,
        \item deletes $\str[\hlab]{\F}(\bx\dd x_q)$ and inserts $\str[\hlab]{\G}(\by\dd y_q)$ (at least one of these fragments is empty),
        \item aligns $\str[\hlab]{\F}[x_q\dd 2|\F|)$ against $\str[\hlab]{\G}[y_q\dd 2|\G|)$ in the same way as $\A_{i-1}$ did.
     \end{itemize}
    Formally, $\A_i$ is defined as follows, where $\odot$ denotes concatenation of sequences and $p\in [0\dd m]$ is such that $(x,y)=(x_{p},y_{p})$:
    \[\A_i= (x_t,y_t)_{t\in [0\dd p)}\odot (x+\delta,y+\delta)_{\delta\in [0\dd \bx-x]} \odot (\hat{x},y_q)_{\hat{x} \in (\bx\dd x_q)} \odot (x_q,\hat{y})_{\hat{y}\in (\by\dd y_q)}\odot (x_t,y_t)_{t\in [q\dd m]}.\]

    First, we show $\A_i\in \ta(\F,\G)$. For this, take a node $u'\in V_\F$ and consider several cases:
    \begin{itemize}
        \item{\boldmath$o_\F(u')\in [x\dd \bx]$ or $c_\F(u')\in [x\dd \bx]$.} In this case, $u'\in \Sub(u)$ and thus both $o_\F(u')\in [x\dd \bx]$ and $c_\F(u')\in [x\dd \bx]$. Moreover, due to $\str[\hlab]{\F}[x\dd \bx]\simeq_{\A_i} \str[\hlab]{\G}[y\dd \by]$, in this case, $\A_i$ matches $u'$ with the corresponding node of the subtree $\Sub(v)$ identical to $\Sub(u)$.
        \item {\boldmath$o_\F(u'),c_\F(u') \in (\bx\dd x_q)$.} In this case, $\A_i$ deletes both $\str[\hlab]{\F}[o_\F(u')]$ and $\str[\hlab]{\F}[c_\F(u')]$.
        \item {\boldmath$o_\F(u')\in (\bx\dd x_q)$ and $c_\F(u')\in [x_q\dd 2|\F|)$.} In this case, $\A_i$ deletes $\str[\hlab]{\F}[o_\F(u')]$ and handles $\str[\hlab]{\F}[c_\F(u')]$ in the same way as $\A_{i-1}$ did. Thus, we must prove that $\A_{i-1}$ deleted $u'$.
        For a proof by contradiction, suppose that $\A_{i-1}$ aligned $u'$ with some node $v'\in V_\G$. By the non-crossing property of $\A_{i-1}$,
        due to $(x,y),(x_q,y_q)\in \A_{i-1}$, we must have $o_\G(v')\in [y\dd y_q)$ and $c_\G(v')\in [y_q\dd 2|\G|)$. Moreover, since $(\bx\dd x_q)\ne \emptyset$, we must have $(\by\dd y_q)=\emptyset$,
        which means that $o_\G(v')\in [y\dd \by]=[y\dd y_q)$. Consequently, $v'\in \Sub(v)$ and hence $c_\G(v')\in [y\dd \by]=[y\dd y_q)$;
        this contradicts $c_\G(v')\in [y_q\dd 2|\G|)$.
        \item {\boldmath$o_\F(u')\in [0\dd x)$ and $c_\F(u')\in (\bx\dd x_q)$.} In this case, $\A_i$ deletes $\str[\hlab]{\F}[c_\F(u')]$ and handles $\str[\hlab]{\F}[o_\F(u')]$ in the same way as $\A_{i-1}$ did. Thus, we must prove that $\A_{i-1}$ deleted $u'$.
        For a proof by contradiction, suppose that $\A_{i-1}$ aligned $u'$ with some node $v'\in V_\G$. By the non-crossing property of $\A_{i-1}$,
        due to $(x,y),(x_q,y_q)\in \A_{i-1}$, we must have $o_\G(v')\in [0\dd y)$ and $c_\G(v')\in [y\dd y_q)$. Moreover, since $(\bx\dd x_q)\ne \emptyset$, we must have $(\by\dd y_q)=\emptyset$,
        which means that $c_\G(v')\in [y\dd \by]=[y\dd y_q)$. Consequently, $v'\in \Sub(v)$ and hence $o_\G(v')\in [y\dd \by]=[y\dd y_q)$;
        this contradicts $o_\G(v')\in [0\dd y)$.
        \item {\boldmath$o_\F(u'),c_\F(u')\in [0\dd x)\cup [x_q\dd 2|\F|)$.} In this case, $\A_i$ handles $u'$ in the same way as $\A_{i-1}$ did.
    \end{itemize}
    This case analysis above is exhaustive and, in all cases, the two parentheses corresponding to $u'$ are handled consistently.
    Consequently, we indeed have $\A_i\in \ta(\F,\G)$.

    Next we argue that $\ted_{\lambda,\A_i}(\F,\G)\le \ted_{\lambda,\A_{i-1}}(\F,\G)$
    or, equivalently, $\ed_{\A_i}(\str[\lab]{\F},\str[\lab]{G})\le \ed_{\A_{i-1}}(\str[\lab]{\F},\str[\lab]{G})$.
    The two alignments differ only on $\str[\lab]{\F}[x\dd x_q)$ and $\str[\lab]{\G}[y\dd y_q)$.
    The contribution of these fragments to the cost of $\A_i: \str[\lab]{\F}\leadsto \str[\lab]{\G}$ equals $|(\bx\dd x_q)|+|(\by\dd y_q)| = \big| |(\bx\dd x_q)|-|(\by\dd y_q)| \big|
    = \big||[x\dd x_q)-|[y\dd y_q)|\big|$. Due to $(x,y),(x_q,y_q)\in \A_{i-1}$,  the contribution of these fragments to the cost of $\A_{i-1}:\str[\lab]{\F}\leadsto \str[\lab]{\G}$ must be at least that large.

    Finally, we note that the breakpoints of $\A_i$ in $[0\dd i)\times [0\dd 2|\G|)$ satisfy the greedy property.
    For breakpoints to the left of $(x,y)$, this follows by definition of $(x,y)$.
    Moreover, due to $\str[\hlab]{\F}[x]\simeq \str[\hlab]{\G}[y]$, the pair $(x,y)$ is not a breakpoint in $\A_i$;
    in particular, the subsequent pair $(x+1,y+1)\in \A_i$ satisfies $x+1=i$.

    \newsavebox{\imagebox}

    \begin{figure}[ht]
        \centering
        \begin{subfigure}[b]{0.4\textwidth}
            \centering\usebox{\imagebox}
            \input{figs/labeling_refinements1}
            \caption{Case 1}
        \end{subfigure}    
        \begin{subfigure}[b]{0.4\textwidth}
            \centering\raisebox{\dimexpr\ht\imagebox-\height}{
            \input{figs/labeling_refinements2}}
            \caption{Case 2}
        \end{subfigure}
        \caption{Example of the construction of alignment $\A_i$ from alignment $\A_{i-1}$}
        \label{fig:my_label}
    \end{figure}
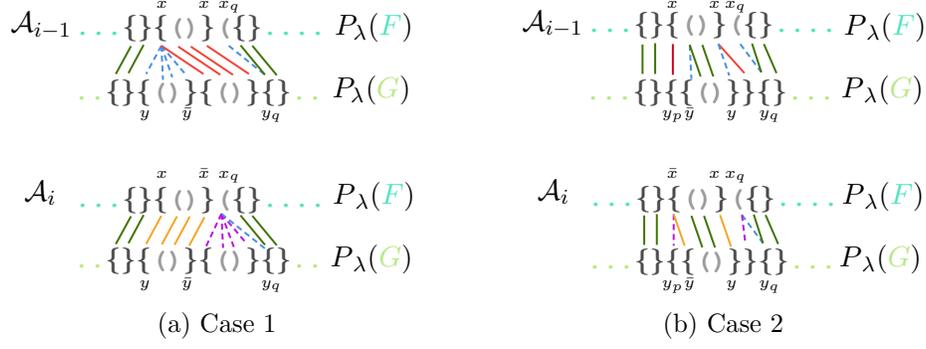
    
    \noindent
    \textbf{Case 2.}
    Next, we consider the case where $(x,y)=(c_\F(u),c_\G(v))$ for some $(u,v)\in V_\F\times V_\G$. Define $(\bx,\by)=(o_\F(u),o_\G(v))$
    and note that, due to $\str[\hlab]{\F}[x]= \str[\hlab]{\G}[y]$ and $\hlab = \look{\lab}{h}$,
    we have $\str[\hlab]{\F}[\bx\dd x] = \str[\hlab]{\G}[\by\dd y]$.
    Define $q = \min\{t : x_t > x \text{ and } y_t > y\}$
    and $p = \max\{t : x_t \le \bx \text{ and } y_t \le \by\}$.
    We create $\A_i$ so that it:
    \begin{itemize}
        \item aligns $\str[\hlab]{\F}[0\dd x_{p})$ against $\str[\hlab]{\G}[0\dd y_{p})$ in the same way as $\A_{i-1}$ did,
        \item deletes $\str[\hlab]{\F}[x_{p}\dd \bx)$ and inserts $\str[\hlab]{\G}[y_{p}\dd \by)$ (at least one of these fragments is empty),
        \item matches $\str[\hlab]{\F}[\bx\dd x]$ against $\str[\hlab]{\G}[\by\dd y]$,
        \item deletes $\str[\hlab]{\F}(x\dd x_q)$ and inserts $\str[\hlab]{\G}(y\dd y_q)$ (at least one of these fragments is empty),
        \item aligns $\str[\hlab]{\F}[x_q\dd 2|\F|)$ against $\str[\hlab]{\G}[y_q\dd 2|\G|)$ in the same way as $\A_{i-1}$ did.
     \end{itemize}
    Formally, $\A_i$ is defined as follows:    
    \begin{align*}
       \A_{i} =  (x_t,y_t)_{t\in[0\dd p)}&\odot (\hat{x},\by)_{\hat{x}\in [x_{p}\dd \bx)} \odot (\bx,\hat{y})_{\hat{y}\in [y_{p}\dd \by)} \odot (\bx+\delta,\by+\delta)_{\delta\in [0\dd x-\bx]}\\             &\odot (\hat{x},y_q)_{\hat{x} \in (x\dd x_q)} \odot (x_q,\hat{y})_{\hat{y} \in (y\dd y_q)} \odot (x_t,y_t)_{t\in [q\dd m]}.
       \end{align*}
    
    First, we show $\A_i\in \ta(\F,\G)$. For this, take a node $u'\in V_\F$ and consider several cases:
       \begin{itemize}
           \item{\boldmath$o_\F(u')\in [\bx\dd x]$ or $c_\F(u')\in [\bx\dd x]$.} In this case, $u'\in \Sub(u)$ and thus both $o_\F(u')\in [\bx\dd x]$ and $c_\F(u')\in [\bx\dd x]$. Moreover, due to $\str[\hlab]{\F}[\bx\dd x]\simeq_{\A_i} \str[\hlab]{\G}[\by\dd y]$, in this case, $\A_i$ matches $u'$ with the corresponding node of the subtree $\Sub(v)$ identical to $\Sub(u)$.
           \item {\boldmath$o_\F(u'),c_\F(u') \in [x_{p}\dd \bx)\cup (x\dd x_q)$}. In this case, $\A_i$ deletes both $\str[\hlab]{\F}[o_\F(u')]$ and $\str[\hlab]{\F}[c_\F(u')]$.
           \item {\boldmath$o_\F(u')\in [x_{p}\dd \bx)$ and $c_\F(u')\in [x_q\dd 2|\F|)$.} In this case, $\A_i$ deletes $\str[\hlab]{\F}[o_\F(u')]$ and handles $\str[\hlab]{\F}[c_\F(u')]$ in the same way as $\A_{i-1}$ did. Thus, we must prove that $\A_{i-1}$ deleted $u'$.
           For a proof by contradiction, suppose that $\A_{i-1}$ aligned $u'$ with some node $v'\in V_\G$. By the non-crossing property of $\A_{i-1}$,
           due to $(x_{p},y_{p}),(x,y),(x_q,y_q)\in \A_{i-1}$, we must have $o_\G(v')\in [y_{p}\dd y)$ and $c_\G(v')\in [y_q\dd 2|\G|)$. Moreover, since $[x_{p}\dd \bx)\ne \emptyset$, we must have $[y_{p}\dd \by)=\emptyset$, which means that $o_\G(v')\in [\by\dd y)=[y_{p}\dd y)$. Consequently, $v'\in \Sub(v)$ and hence $c_\G(v')\in [s\dd y]=[y_{p}\dd y]$; this contradicts  $c_\G(v')\in [y_q\dd 2|\G|)$ due to $y_q>y$.
           \item {\boldmath$o_\F(u')\in (x\dd x_q)$ and $c_\F(u')\in [x_q\dd 2|\F|)$.} In this case, $\A_i$ deletes $\str[\hlab]{\F}[o_\F(u')]$ and handles $\str[\hlab]{\F}[c_\F(u')]$ in the same way as $\A_{i-1}$ did. Thus, we must prove that $\A_{i-1}$ deleted $u'$.
           For a proof by contradiction, suppose that $\A_{i-1}$ aligned $u'$ with some node $v'\in V_\G$. By the non-crossing property of $\A_{i-1}$,
           due to $(x_{p},y_{p}),(x,y),(x_q,y_q)\in \A_{i-1}$, we must have $o_\G(v')\in [y\dd y_q)$ and $c_\G(v')\in [y_q\dd 2|\G|)$. Moreover, since $(x\dd x_q)\ne \emptyset$, we must have $(y\dd y_q)=\emptyset$,
           which means that $o_\G(v')=y$; this is a contradiction because $y=c_\G(v)$.
           \item {\boldmath$o_\F(u')\in [0\dd x_{p})$ and $c_\F(u')\in [x_{p}\dd \bx)$.} In this case, $\A_i$ deletes $\str[\hlab]{\F}[c_\F(u')]$ and handles $\str[\hlab]{\F}[o_\F(u')]$ in the same way as $\A_{i-1}$ did. Thus, we must prove that $\A_{i-1}$ deleted $u'$.
           For a proof by contradiction, suppose that $\A_{i-1}$ aligned $u'$ with some node $v'\in V_\G$. By the non-crossing property of $\A_{i-1}$,
           due to $(x_{p},y_{p}),(x,y),(x_q,y_q)\in \A_{i-1}$, we must have $o_\G(v')\in [0\dd y_{p})$ and $c_\G(v')\in [y_{p}\dd y)$. Moreover, since $[x_{p}\dd \bx)\ne \emptyset$, we must have $[y_q\dd \by)=\emptyset$,
           which means that $o_\G(v')\in [\by\dd y)=[y_p\dd y)$. Consequently, $v'\in \Sub(v)$ and hence $o_\G(v')\in [\by \dd y]=[y_{p}\dd y]$;
           this contradicts $o_\G(v')\in [0\dd y_{p})$.
           \item {\boldmath$o_\F(u')\in [0\dd x_p)$ and $c_\F(u')\in (x\dd x_q)$.} In this case, $\A_i$ deletes $\str[\hlab]{\F}[c_\F(u')]$ and handles $\str[\hlab]{\F}[o_\F(u')]$ in the same way as $\A_{i-1}$ did. Thus, we must prove that $\A_{i-1}$ deleted $u'$.
           For a proof by contradiction, suppose that $\A_{i-1}$ aligned $u'$ with some node $v'\in V_\G$. By the non-crossing property of $\A_{i-1}$,
           due to $(x_{p},y_{p}),(x,y),(x_q,y_q)\in \A_{i-1}$, we must have $o_\G(v')\in [0\dd y_{p})$ and $c_\G(v')\in [y\dd y_q)$. Moreover, since $(x\dd x_q)\ne \emptyset$, we must have $(y\dd y_q)=\emptyset$,
           which means that $c_\G(v')=y$. Consequently, $v=v'$ and hence $o_\G(v')=\by$; this contradicts $o_\G(v')\in [0\dd y_{p})$ due to $y_{p}\le \by$.
           \item {\boldmath$o_\F(u'),c_\F(u')\in [0\dd x_{p})\cup [x_q\dd 2|\F|)$.} In this case, $\A_i$ handles $u'$ in the same way as $\A_{i-1}$ did.
       \end{itemize}
       This case analysis above is exhaustive and, in all cases, the two parentheses corresponding to $u'$ are handled consistently.
       Consequently, we indeed have $\A_i\in \ta(\F,\G)$.

       Next we argue that $\ted_{\lambda,\A_i}(\F,\G)\le \ted_{\lambda,\A_{i-1}}(\F,\G)$
       or, equivalently, $\ed_{\A_i}(\str[\lab]{\F},\str[\lab]{G})\le \ed_{\A_{i-1}}(\str[\lab]{\F},\str[\lab]{G})$.
       The two alignments differ only in how on $\str[\lab]{\F}[x_{p}\dd x_q)$ and $\str[\lab]{\G}[y_{p}\dd y_q)$.
       The contribution of these fragments to the cost of $\A_i:\str[\lab]{\F}\leadsto \str[\lab]{\G}$ equals $|[x_{p}\dd \bx)|+|[y_{p}\dd \by)|+|(x\dd x_q)|+|(y\dd y_q)| =  \big| |[x_{p}\dd \bx)|-|[y_{p}\dd \by)| \big|+ \big| |(x\dd x_q)|-|(y\dd y_q)| \big|
    = \big||[x_{p}\dd x)-|[y_{p}\dd y)|\big|+\big||[x\dd x_q)-|[y\dd y_q)|\big|$. Due to $(x_{p},y_{p})(x,y),(x_q,y_q)\in \A_{i-1}$,  the contribution of these fragments to the cost of $\A_{i-1}:\str[\lab]{\F}\leadsto \str[\lab]{\G}$ must be at least that large.

    Finally, we shall prove that $\A_i$ is greedy up to index $i-1$.
    For breakpoints to the left of $(x_p,y_p)$, this follows by definition of $(x,y)$.
    Moreover, none of the pairs $(\bx+\delta,\by+\delta)_{\delta\in [0\dd x-\bx]}$ is a breakpoint and, in particular, the subsequent pair $(x+1,y+1)\in \A_i$ satisfies $x+1=i$.
    Hence, it suffices to consider breakpoints of the form $(\hat{x},\by)$ for some $\hat{x}\in [x_{p}\dd \bx)$ or $(\bx,\hat{y})$ for some ${\hat{y}\in [y_{p}\dd \by)}$.
    These two cases are symmetric, so let us consider $(\bx,\hat{y})$ for some $\hat{y}\in [y_{p}\dd \by)$.
    For a proof by contradiction, suppose that $\str[\hlab]{\F}[\bx]=\str[\hlab]{\G}[\hat{y}]$; consequently, $\str[\hlab]{\F}[\bx\dd x]=\str[\hlab]{\G}[\hat{y} \dd \hat{y}+x-\bx]$
    because $\hlab = \look{\lab}{h}$.
    Moreover, due to $[y_{p}\dd \by)\ne \emptyset$, so we must have $[x_p \dd \bx)=\emptyset$, i.e., $x_p=\bx$.
    Define $r = \min\{t \in [p\dd q) : y_t - x_t \ge \hat{y}-\bx\}$.
    Note that $y_p - x_p \le \hat{y}-\bx$ and $y-x = \by - \bx > \hat{y}-\bx$;
    hence, $x_r \in [x_p\dd x]=[\bx\dd x]$ and $y_r - x_r = \hat{y}-\bx$.
    Since all breakpoints to the left of $(x,y)$ satisfy the greedy property and since  $\str[\hlab]{\F}[\bx\dd x]=\str[\hlab]{\G}[\hat{y} \dd \hat{y}+x-\bx]$,
    we conclude that $\str[\hlab]{\F}[x_r\dd x]\simeq_{\A_{i-1}}\str[\hlab]{\G}[y_r\dd \hat{y}+x-\bx]$.
    In particular, $\str[\hlab]{\F}[x]\simeq_{\A_{i-1}} \str[\hlab]{\G}[\hat{y}+x-\bx]$, which contradicts $(x,y)\in \A_{i-1}$ by the non-crossing property of $\A_{i-1}$ because $\hat{y}+x-\bx < \by + x - \bx = y$.

    \vspace{3mm}
    
    After $2|\F|$ rounds, we construct a tree alignment $\A_{2|\F|}\in \ta(\F,\G)$  such that $\ted_{\lambda,\A_{2|\F|}}(F,G)=\ted_{\lambda,\A}(F,G)$ and  $\str[\hlab]{\F}[x]\neq \str[\hlab]{\G}[y]$ holds for every $(x,y)\in \brkp_{\A_{2|\F|}}(\str[\hlab]{\F},\str[\hlab]{\G}) \cap ([0\dd 2|\F|-1]\times [0\dd 2|\G|))$. Thus, $\A_{2|\F|}\in \ta(\F,\G)$ is an optimum tree alignment and $\A_{2|\F|}\in \ga(\str[\hlab]{\F},\str[\hlab]{\G})$.
    \end{proof}

\subsection{Compatibility Refinement}

\begin{definition}[Compatibility]
    Let $\lab$ be a joint labeling of forests $\F,\G$, and let $w\in \Zz$.
    We say that nodes $u\in V_{\F}$ and $v\in V_{\G}$ are $w$-\emph{compatible}
    if $\lab(u)=\lab(v)$, $|o_\F(u)-o_\G(v)|\le w$, and $|c_\F(u)-c_\G(v)|\le w$.
    \end{definition}
    
    \newcommand{\hF}{\hat{F}}
    \newcommand{\hG}{\hat{G}}
    \newcommand{\bF}{\bar{F}}
    \newcommand{\bG}{\bar{G}}
    \newcommand{\bM}{\bar{M}}
    \newcommand{\hQ}{\hat{Q}}
    \newcommand{\CR}[2]{\mathsf{C}(#1,#2)}

    \begin{definition}[Compatibilty refinement]
    Let $\lab$ be a joint labeling of forests $\F,\G$, and let $w\in \Zz$.
    We say that $\lab'$ is a \emph{$w$-compatibility refinement} of $\lab$
    if, for any $u,v\in V_{\F}\cup V_{\G}$, we have $\lab'(u)=\lab'(v)$
    if and only if $(u,v)$ belongs to the transitive closure of the $w$-compatibility relation.
    We use $\CR{\lab}{w}$ to denote a \emph{$w$-compatibility refinement} of $\lab$ (chosen arbitrarily up to equivalence of labelings).
    \end{definition}
    
    \begin{lemma}\label{lem:buildC}
    Given forests $\F,\G$ of total size $n$, their joint labeling $\lab$,
    and an integer $w$, a $w$-compatibility refinement labeling $\lab' = \CR{\lab}{w}$ 
    can be constructed in $\Oh(n\log n)$ time.
    \end{lemma}
    \begin{proof}
    We process nodes according to the original label.
    For each label $\ell$, we use a BFS-based algorithm to compute connected components with respect to the $w$-compatibility relation. While doing so, we maintain two instances of a data structure for dynamic orthogonal range reporting~\cite{M06},
    one with points $(o_\F(u),c_\F(u))$ for all unvisited nodes $u\in V_\F$ with label $\ell$,
    and one with points $(o_\G(v),c_\G(v))$ for all unvisited nodes $v\in V_\G$ with label $\ell$.
    While we do not have direct access to edges incident to any vertex, an orthogonal range reporting query allows
    listing all unvisited neighbors of any vertex. Each listed node is immediately visited (and thus the corresponding point is removed from the range reporting data structure), so, in total, we have $\Oh(n)$ queries producing output of total size $\Oh(n)$. Thus, the algorithm works in $\Oh(n\log n)$ time.
    \end{proof}
    
\begin{observation}\label{obs:boundCompatibility}
    Consider a joint labeling $\lab$ of forests $\F,\G$,
     an integer $w\in \Zz$, and an alignment $\A\in \ta(\F,\G)$ of width at most $w$.
     Then, $\ted_{\CR{\lab}{w},\A}(\F,\G)=\ted_{\lab,\A}(\F,\G)$.
\end{observation}

%% file: figs/labeling_refinements1.tex
\tikzset{every picture/.style={line width=0.75pt}} 

\begin{tikzpicture}[x=0.75pt,y=0.75pt,yscale=-.85,xscale=.85]

\draw [color={rgb, 255:red, 65; green, 117; blue, 5 }  ,draw opacity=1 ]   (121.75,120.59) -- (112.75,137.09) ;
\draw [color={rgb, 255:red, 65; green, 117; blue, 5 }  ,draw opacity=1 ]   (129.25,119.84) -- (120.25,136.59) ;
\draw [color={rgb, 255:red, 74; green, 144; blue, 226 }  ,draw opacity=1 ] [dash pattern={on 2.25pt off 1.5pt on 2.25pt off 1.5pt}]  (139.75,120.34) -- (130.75,137.09) ;
\draw [color={rgb, 255:red, 74; green, 144; blue, 226 }  ,draw opacity=1 ] [dash pattern={on 2.25pt off 1.5pt on 2.25pt off 1.5pt}]  (139.75,120.34) -- (140.75,140.59) ;
\draw [color={rgb, 255:red, 74; green, 144; blue, 226 }  ,draw opacity=1 ] [dash pattern={on 2.25pt off 1.5pt on 2.25pt off 1.5pt}]  (139.75,120.34) -- (146.25,139.09) ;
\draw [color={rgb, 255:red, 74; green, 144; blue, 226 }  ,draw opacity=1 ] [dash pattern={on 2.25pt off 1.5pt on 2.25pt off 1.5pt}]  (139.75,120.34) -- (154.75,138.09) ;
\draw [color={rgb, 255:red, 242; green, 68; blue, 59 }  ,draw opacity=1 ]   (139.75,120.34) -- (164.75,137.09) ;
\draw [color={rgb, 255:red, 242; green, 68; blue, 59 }  ,draw opacity=1 ]   (149.75,120.84) -- (174.75,137.59) ;
\draw [color={rgb, 255:red, 242; green, 68; blue, 59 }  ,draw opacity=1 ]   (157.25,120.84) -- (182.25,137.59) ;
\draw [color={rgb, 255:red, 242; green, 68; blue, 59 }  ,draw opacity=1 ]   (166.75,120.34) -- (191.75,137.09) ;
\draw [color={rgb, 255:red, 74; green, 144; blue, 226 }  ,draw opacity=1 ] [dash pattern={on 2.25pt off 1.5pt on 2.25pt off 1.5pt}]  (179.75,120.34) -- (201.25,137.59) ;
\draw [color={rgb, 255:red, 65; green, 117; blue, 5 }  ,draw opacity=1 ]   (201.25,137.59) -- (186.75,120.09) ;
\draw [color={rgb, 255:red, 65; green, 117; blue, 5 }  ,draw opacity=1 ]   (208.25,137.84) -- (193.75,120.34) ;
\draw [color={rgb, 255:red, 65; green, 117; blue, 5 }  ,draw opacity=1 ]   (121.75,221.59) -- (112.75,238.09) ;
\draw [color={rgb, 255:red, 65; green, 117; blue, 5 }  ,draw opacity=1 ]   (129.25,220.84) -- (120.25,237.59) ;
\draw [color={rgb, 255:red, 245; green, 166; blue, 35 }  ,draw opacity=1 ]   (139.75,221.34) -- (130.75,238.09) ;
\draw [color={rgb, 255:red, 189; green, 16; blue, 224 }  ,draw opacity=1 ] [dash pattern={on 2.25pt off 1.5pt on 2.25pt off 1.5pt}]  (175.39,220.13) -- (176.39,240.38) ;
\draw [color={rgb, 255:red, 189; green, 16; blue, 224 }  ,draw opacity=1 ] [dash pattern={on 2.25pt off 1.5pt on 2.25pt off 1.5pt}]  (175.39,220.13) -- (182.79,239.84) ;
\draw [color={rgb, 255:red, 189; green, 16; blue, 224 }  ,draw opacity=1 ] [dash pattern={on 2.25pt off 1.5pt on 2.25pt off 1.5pt}]  (175.39,220.13) -- (190.39,237.88) ;
\draw [color={rgb, 255:red, 74; green, 144; blue, 226 }  ,draw opacity=1 ] [dash pattern={on 2.25pt off 1.5pt on 2.25pt off 1.5pt}]  (175.39,220.13) -- (201.39,240.88) ;
\draw [color={rgb, 255:red, 65; green, 117; blue, 5 }  ,draw opacity=1 ]   (201.25,238.59) -- (186.75,221.09) ;
\draw [color={rgb, 255:red, 65; green, 117; blue, 5 }  ,draw opacity=1 ]   (208.25,238.84) -- (193.75,221.34) ;
\draw [color={rgb, 255:red, 245; green, 166; blue, 35 }  ,draw opacity=1 ]   (149.04,221.63) -- (140.04,238.38) ;
\draw [color={rgb, 255:red, 245; green, 166; blue, 35 }  ,draw opacity=1 ]   (156.61,221.77) -- (147.61,238.52) ;
\draw [color={rgb, 255:red, 245; green, 166; blue, 35 }  ,draw opacity=1 ]   (165.32,222.2) -- (156.32,238.95) ;
\draw [color={rgb, 255:red, 189; green, 16; blue, 224 }  ,draw opacity=1 ] [dash pattern={on 2.25pt off 1.5pt on 2.25pt off 1.5pt}]  (175.39,220.13) -- (166.21,239.7) ;


\draw (86.67,100.33) node [anchor=north west][inner sep=0.75pt]   [align=left][font=\small] {{\fontfamily{pcr}\selectfont \textbf{\textcolor[rgb]{0.31,0.89,0.76}{...}\textcolor[rgb]{0.29,0.29,0.29}{\{\}\{}\textcolor[rgb]{0.61,0.61,0.61}{()}\textcolor[rgb]{0.29,0.29,0.29}{\}}\textcolor[rgb]{0.61,0.61,0.61}{(}\textcolor[rgb]{0.29,0.29,0.29}{\{\}}\textcolor[rgb]{0.31,0.89,0.76}{....}}}};
\draw (86,138.33) node [anchor=north west][inner sep=0.75pt]   [align=left][font=\small] {{\fontfamily{pcr}\selectfont \textbf{\textcolor[rgb]{0.72,0.91,0.53}{..}\textcolor[rgb]{0.29,0.29,0.29}{\{\}\{}\textcolor[rgb]{0.61,0.61,0.61}{()}\textcolor[rgb]{0.29,0.29,0.29}{\}\{}\textcolor[rgb]{0.61,0.61,0.61}{()}\textcolor[rgb]{0.29,0.29,0.29}{\}\{\}}\textcolor[rgb]{0.72,0.91,0.53}{..}}}};
\draw (48.67,98.68) node [anchor=north west][inner sep=0.75pt]    {$\mathcal{A}_{i-1}$};
\draw (241,99.34) node [anchor=north west][inner sep=0.75pt]  [color={rgb, 255:red, 80; green, 227; blue, 194 }  ,opacity=1 ]  {$\textcolor[rgb]{0,0,0}{P_{\lambda }(} F\textcolor[rgb]{0,0,0}{)}$};
\draw (239.5,136.34) node [anchor=north west][inner sep=0.75pt]  [color={rgb, 255:red, 80; green, 227; blue, 194 }  ,opacity=1 ]  {$\textcolor[rgb]{0,0,0}{P_{\lambda }(}\textcolor[rgb]{0.72,0.91,0.53}{G}\textcolor[rgb]{0,0,0}{)}$};
\draw (140,94) node [anchor=mid][inner sep=0.75pt]  [font=\tiny]  {$x$};
\draw (165,94) node [anchor=mid][inner sep=0.75pt]  [font=\tiny]  {$\bx$};
\draw (180,94) node [anchor=mid][inner sep=0.75pt]  [font=\tiny]  {$x_q$};
\draw (130,159) node [anchor=mid][inner sep=0.75pt]  [font=\tiny]  {$y$};
\draw (155,159) node [anchor=mid][inner sep=0.75pt]  [font=\tiny]  {$\by$};
\draw (205,159) node [anchor=mid][inner sep=0.75pt]  [font=\tiny]  {$y_q$};
\draw (86.67,202.33) node [anchor=north west][inner sep=0.75pt]   [align=left][font=\small] {{\fontfamily{pcr}\selectfont \textbf{\textcolor[rgb]{0.31,0.89,0.76}{...}\textcolor[rgb]{0.29,0.29,0.29}{\{\}\{}\textcolor[rgb]{0.61,0.61,0.61}{()}\textcolor[rgb]{0.29,0.29,0.29}{\}}\textcolor[rgb]{0.61,0.61,0.61}{(}\textcolor[rgb]{0.29,0.29,0.29}{\{\}}\textcolor[rgb]{0.31,0.89,0.76}{....}}}};
\draw (86,238.33) node [anchor=north west][inner sep=0.75pt]   [align=left][font=\small] {{\fontfamily{pcr}\selectfont \textbf{\textcolor[rgb]{0.72,0.91,0.53}{..}\textcolor[rgb]{0.29,0.29,0.29}{\{\}\{}\textcolor[rgb]{0.61,0.61,0.61}{()}\textcolor[rgb]{0.29,0.29,0.29}{\}\{}\textcolor[rgb]{0.61,0.61,0.61}{()}\textcolor[rgb]{0.29,0.29,0.29}{\}\{\}}\textcolor[rgb]{0.72,0.91,0.53}{..}}}};
\draw (56.67,198.68) node [anchor=north west][inner sep=0.75pt]    {$\mathcal{A}_{i}$};
\draw (140,195) node [anchor=mid][inner sep=0.75pt]  [font=\tiny]  {$x$};
\draw (165,195) node [anchor=mid][inner sep=0.75pt]  [font=\tiny]  {$\bx$};
\draw (180,195) node [anchor=mid][inner sep=0.75pt]  [font=\tiny]  {$x_q$};
\draw (130,260) node [anchor=mid][inner sep=0.75pt]  [font=\tiny]  {$y$};
\draw (155,260) node [anchor=mid][inner sep=0.75pt]  [font=\tiny]  {$\by$};
\draw (205,260) node [anchor=mid][inner sep=0.75pt]  [font=\tiny]  {$y_q$};
\draw (238.3,236.74) node [anchor=north west][inner sep=0.75pt]  [color={rgb, 255:red, 80; green, 227; blue, 194 }  ,opacity=1 ]  {$\textcolor[rgb]{0,0,0}{P}\textcolor[rgb]{0,0,0}{_{\lambda }}\textcolor[rgb]{0,0,0}{(}\textcolor[rgb]{0.72,0.91,0.53}{G}\textcolor[rgb]{0,0,0}{)}$};
\draw (239.8,199.74) node [anchor=north west][inner sep=0.75pt]  [color={rgb, 255:red, 80; green, 227; blue, 194 }  ,opacity=1 ]  {$\textcolor[rgb]{0,0,0}{P}\textcolor[rgb]{0,0,0}{_{\lambda }}\textcolor[rgb]{0,0,0}{(} F\textcolor[rgb]{0,0,0}{)}$};

\end{tikzpicture}

%% file: figs/labeling_refinements2.tex
\tikzset{every picture/.style={line width=0.75pt}} 

\begin{tikzpicture}[x=0.75pt,y=0.75pt,yscale=-.85,xscale=.85]

\draw [color={rgb, 255:red, 65; green, 117; blue, 5 }  ,draw opacity=1 ]   (432.25,119.59) -- (432.25,137.09) ;
\draw [color={rgb, 255:red, 65; green, 117; blue, 5 }  ,draw opacity=1 ]   (439.75,118.84) -- (439.75,137.09) ;
\draw [color={rgb, 255:red, 74; green, 144; blue, 226 }  ,draw opacity=1 ] [dash pattern={on 2.25pt off 1.5pt on 2.25pt off 1.5pt}]  (460.25,119.84) -- (461.25,140.09) ;
\draw [color={rgb, 255:red, 242; green, 68; blue, 59 }  ,draw opacity=1 ]   (477.25,119.34) -- (492.75,136.59) ;
\draw [color={rgb, 255:red, 74; green, 144; blue, 226 }  ,draw opacity=1 ] [dash pattern={on 2.25pt off 1.5pt on 2.25pt off 1.5pt}]  (490.25,119.34) -- (502.75,136.09) ;
\draw [color={rgb, 255:red, 65; green, 117; blue, 5 }  ,draw opacity=1 ]   (502.75,136.09) -- (497.25,119.09) ;
\draw [color={rgb, 255:red, 65; green, 117; blue, 5 }  ,draw opacity=1 ]   (511.75,136.59) -- (504.25,119.34) ;
\draw [color={rgb, 255:red, 208; green, 2; blue, 27 }  ,draw opacity=1 ]   (450.25,119.34) -- (450.25,137.59) ;
\draw [color={rgb, 255:red, 65; green, 117; blue, 5 }  ,draw opacity=1 ]   (460.25,119.84) -- (466.75,138.59) ;
\draw [color={rgb, 255:red, 65; green, 117; blue, 5 }  ,draw opacity=1 ]   (467.75,119.84) -- (474.25,138.59) ;
\draw [color={rgb, 255:red, 74; green, 144; blue, 226 }  ,draw opacity=1 ] [dash pattern={on 2.25pt off 1.5pt on 2.25pt off 1.5pt}]  (477.25,119.34) -- (483.75,138.09) ;
\draw [color={rgb, 255:red, 65; green, 117; blue, 5 }  ,draw opacity=1 ]   (433.25,221.09) -- (433.25,238.59) ;
\draw [color={rgb, 255:red, 65; green, 117; blue, 5 }  ,draw opacity=1 ]   (440.75,220.34) -- (440.75,238.59) ;
\draw [color={rgb, 255:red, 74; green, 144; blue, 226 }  ,draw opacity=1 ] [dash pattern={on 2.25pt off 1.5pt on 2.25pt off 1.5pt}]  (491.25,220.84) -- (503.75,237.59) ;
\draw [color={rgb, 255:red, 65; green, 117; blue, 5 }  ,draw opacity=1 ]   (503.75,237.59) -- (498.25,220.59) ;
\draw [color={rgb, 255:red, 65; green, 117; blue, 5 }  ,draw opacity=1 ]   (512.75,238.09) -- (505.25,220.84) ;
\draw [color={rgb, 255:red, 65; green, 117; blue, 5 }  ,draw opacity=1 ]   (461.25,221.34) -- (467.75,240.09) ;
\draw [color={rgb, 255:red, 65; green, 117; blue, 5 }  ,draw opacity=1 ]   (468.75,221.34) -- (475.25,240.09) ;
\draw [color={rgb, 255:red, 245; green, 166; blue, 35 }  ,draw opacity=1 ]   (478.25,220.84) -- (484.75,239.59) ;
\draw [color={rgb, 255:red, 189; green, 16; blue, 224 }  ,draw opacity=1 ] [dash pattern={on 2.25pt off 1.5pt on 2.25pt off 1.5pt}]  (450.75,220.34) -- (450.75,239.09) ;
\draw [color={rgb, 255:red, 245; green, 166; blue, 35 }  ,draw opacity=1 ]   (450.75,220.34) -- (457.25,239.09) ;
\draw [color={rgb, 255:red, 189; green, 16; blue, 224 }  ,draw opacity=1 ] [dash pattern={on 2.25pt off 1.5pt on 2.25pt off 1.5pt}]  (491.25,220.84) -- (493.25,237.59) ;

\draw (397.17,98.33) node [anchor=north west][inner sep=0.75pt]   [align=left][font=\small] {{\fontfamily{pcr}\selectfont \textbf{\textcolor[rgb]{0.31,0.89,0.76}{...}\textcolor[rgb]{0.29,0.29,0.29}{\{\}\{}\textcolor[rgb]{0.61,0.61,0.61}{()}\textcolor[rgb]{0.29,0.29,0.29}{\}}\textcolor[rgb]{0.61,0.61,0.61}{(}\textcolor[rgb]{0.29,0.29,0.29}{\{\}}\textcolor[rgb]{0.31,0.89,0.76}{....}}}};
\draw (396.5,138.33) node [anchor=north west][inner sep=0.75pt]   [align=left][font=\small] {{\fontfamily{pcr}\selectfont \textbf{\textcolor[rgb]{0.72,0.91,0.53}{...}\textcolor[rgb]{0.29,0.29,0.29}{\{\}\{\{}\textcolor[rgb]{0.61,0.61,0.61}{()}\textcolor[rgb]{0.29,0.29,0.29}{\}\}\{\}}\textcolor[rgb]{0.72,0.91,0.53}{...}}}};
\draw (359.17,97.68) node [anchor=north west][inner sep=0.75pt]    {$\mathcal{A}_{i-1}$};
\draw (475,94) node [anchor=mid][inner sep=0.75pt]  [font=\tiny]  {$x$};
\draw (450,94) node [anchor=mid][inner sep=0.75pt]  [font=\tiny]  {$\bx$};
\draw (487.5,94) node [anchor=mid][inner sep=0.75pt]  [font=\tiny]  {$x_q$};
\draw (450,159) node [anchor=mid][inner sep=0.75pt]  [font=\tiny]  {$y_{p}$};
\draw (460,159) node [anchor=mid][inner sep=0.75pt]  [font=\tiny]  {$\by$};
\draw (485,159) node [anchor=mid][inner sep=0.75pt]  [font=\tiny]  {$y$};
\draw (507.5,159) node [anchor=mid][inner sep=0.75pt]  [font=\tiny]  {$y_q$};
\draw (398.17,201.83) node [anchor=north west][inner sep=0.75pt]   [align=left] [font=\small] {{\fontfamily{pcr}\selectfont \textbf{\textcolor[rgb]{0.31,0.89,0.76}{...}\textcolor[rgb]{0.29,0.29,0.29}{\{\}\{}\textcolor[rgb]{0.61,0.61,0.61}{()}\textcolor[rgb]{0.29,0.29,0.29}{\}}\textcolor[rgb]{0.61,0.61,0.61}{(}\textcolor[rgb]{0.29,0.29,0.29}{\{\}}\textcolor[rgb]{0.31,0.89,0.76}{....}}}};
\draw (397.5,238.83) node [anchor=north west][inner sep=0.75pt]   [align=left] [font=\small] {{\fontfamily{pcr}\selectfont \textbf{\textcolor[rgb]{0.72,0.91,0.53}{...}\textcolor[rgb]{0.29,0.29,0.29}{\{\}\{\{}\textcolor[rgb]{0.61,0.61,0.61}{()}\textcolor[rgb]{0.29,0.29,0.29}{\}\}\{\}}\textcolor[rgb]{0.72,0.91,0.53}{...}}}};
\draw (368.17,198.18) node [anchor=north west][inner sep=0.75pt]    {$\mathcal{A}_{i}$};
\draw (450,195) node [anchor=mid][inner sep=0.75pt]  [font=\tiny]  {$\bx$};
\draw (487.5,195) node [anchor=mid][inner sep=0.75pt]  [font=\tiny]  {$x_q$};
\draw (475,195) node [anchor=mid][inner sep=0.75pt]  [font=\tiny]  {$x$};
\draw (450,260) node [anchor=mid][inner sep=0.75pt]  [font=\tiny]  {$y_{p}$};
\draw (460,260) node [anchor=mid][inner sep=0.75pt]  [font=\tiny]  {$\by$};
\draw (485,260) node [anchor=mid][inner sep=0.75pt]  [font=\tiny]  {$y$};
\draw (507.5,260) node [anchor=mid][inner sep=0.75pt]  [font=\tiny]  {$y_q$};
\draw (550,99.34) node [anchor=north west][inner sep=0.75pt]  [color={rgb, 255:red, 80; green, 227; blue, 194 }  ,opacity=1 ]  {$\textcolor[rgb]{0,0,0}{P}\textcolor[rgb]{0,0,0}{_{\lambda }}\textcolor[rgb]{0,0,0}{(} F\textcolor[rgb]{0,0,0}{)}$};
\draw (548.5,136.34) node [anchor=north west][inner sep=0.75pt]  [color={rgb, 255:red, 80; green, 227; blue, 194 }  ,opacity=1 ]  {$\textcolor[rgb]{0,0,0}{P}\textcolor[rgb]{0,0,0}{_{\lambda }}\textcolor[rgb]{0,0,0}{(}\textcolor[rgb]{0.72,0.91,0.53}{G}\textcolor[rgb]{0,0,0}{)}$};
\draw (547.3,236.74) node [anchor=north west][inner sep=0.75pt]  [color={rgb, 255:red, 80; green, 227; blue, 194 }  ,opacity=1 ]  {$\textcolor[rgb]{0,0,0}{P}\textcolor[rgb]{0,0,0}{_{\lambda }}\textcolor[rgb]{0,0,0}{(}\textcolor[rgb]{0.72,0.91,0.53}{G}\textcolor[rgb]{0,0,0}{)}$};
\draw (548.8,199.74) node [anchor=north west][inner sep=0.75pt]  [color={rgb, 255:red, 80; green, 227; blue, 194 }  ,opacity=1 ]  {$\textcolor[rgb]{0,0,0}{P}\textcolor[rgb]{0,0,0}{_{\lambda }}\textcolor[rgb]{0,0,0}{(} F\textcolor[rgb]{0,0,0}{)}$};

\end{tikzpicture}

%% file: src/string_per.tex
In this section, we analyze the structure of alignments $\A \in \ga_{k,w}(X,Y)$ 
under the assumption that $X,Y$ avoid certain periodic structure. 
This is loosely inspired by the proof of~\cite[Lemma III.10]{FOCS}.

\begin{definition}
    We say that a pattern $P$ has \emph{$s$-synchronized occurrences}
    in strings $X,Y$ if $P = X[x\dd x+|P|) = Y[y\dd y+|P|)$ holds for some positions $x,y$ satisfying $|x-y|\le s$.
\end{definition}
    
\begin{definition}
    Consider integers $s,e,\ell\in \Zp$.
    We say that strings $X,Y\in \Sigma^*$ \emph{avoid $s$-synchronized $e$-powers with root at most $\ell$}
    if there is no non-empty string $Q\in \Sigma^{\le \ell}$ such that $Q^e$ has $s$-synchronized occurrences in $X,Y$.
\end{definition}

\begin{lemma}\label{lem:sheep}
    Consider strings $X,Y$, integers $e,k,w\in \Zp$, and two alignments $\A,\A' \in \aa_{k,w}(X,Y)$.
    If $X,Y$ avoid $w$-synchronized $e$-powers with root at most $2w$, then $|\A \sm \A'|\le 7wke$.
    \end{lemma}
    \begin{proof}
    Let us partition $X$ into individual characters representing deletions or substitutions of $\A$ or $\A'$
    and maximal fragments $X[x\dd x+\ell)$ that $\A$ and $\A'$ match perfectly: $X[x\dd x+\ell)\simeq_{\A} Y[y\dd y+\ell)$
    and $X[x\dd x+\ell)\simeq_{\A'} Y[y'\dd y'+\ell)$ for some $y,y'\in [0\dd |Y|-\ell]\cap [x-w\dd x+w]$.
    Observe that the number of such maximal fragments is at most $2k+1$ and their total length is at least $|X|-2k$.
    
    If $y= y'$, then $(x+\delta,y+\delta)\in \A \cap \A'$ holds for all $\delta\in [0\dd \ell]$.
    Otherwise, $P:=X[x\dd x+\ell)$ has $w$-synchronized occurrences in $X,Y$ (starting at positions $x,y$)
    and $\per(P) \le |y-y'| \le 2w$. Denote $Q:= X[x\dd x+\per(P))$ and observe that, if $\ell \ge 2we$, then $Q^e$ is a prefix of $P$ and thus also has $w$-synchronized occurrences in $X,Y$, contradicting our assumption.
    Hence,  $y\ne y'$ is only possible if $\ell < 2we$.
    In either case, $X[x\dd x+\ell)$ contributes at least $\ell+1-2we$ elements to $\A \cap \A'$.
    
    Overall, the fragments contribute at least $|X|-2k+(1-2we)(2k+1)\ge |X|-4wke -2we +1$ elements to $\A \cap \A'$.
    Consequently,
    $|\A\sm \A'| \le |\A|-|\A \cap \A'| \le |X|+1+k - (|X|-4wke -2we +1)
    = 4wke + 2we+k \le 7wke$.
    \end{proof}
    
    \begin{lemma}[{Compare~\cite[Lemma III.10]{FOCS}}]\label{lem:almostEverythingIsCommon}
        Consider strings $X,Y\in \Sigma^{\le n}$ and integers $e,k,w\in \Zp$.
        If $X,Y$ avoid $w$-synchronized $e$-powers with root at most $2w$, then 
        there exists a set $M$ of size $|M|\ge |X|-15wk^2e$
        such that $M\sub \mtch_{\A}(X,Y)$ holds for every alignment $\A \in \ga_{k,w}(X,Y)$.
        Moreover, such a set $M$ can be constructed in $\Oh(n+kw)$ time.
    \end{lemma}
    \begin{proof}
        Let us fix an alignment $\A \in \ga_{k,w}(X,Y)$ and partition $X$ into individual characters representing deletions or substitutions of $\A$ and maximal fragments $X[x\dd x+\ell)$ that $\A$ matches perfectly: $X[x\dd x+\ell)\simeq_{\A} Y[y\dd y+\ell)$ for some $y\in [0\dd |Y|-\ell]$. Observe that the number of such maximal fragments is at most $k+1$ and their total length it at least $|X|-k$. 
        
        Let us fix a maximal fragment $X[x\dd x+\ell)\simeq_{\A} Y[y\dd y+\ell)$.
        We claim that $(x+\delta,y+\delta)$ can be added to $M$ for every $\delta \in [7wke\dd \ell)$.
        For this, consider another alignment $\A'\in \ga_{k,w}(X,Y)$.
        By \cref{lem:sheep}, we have $|\A \sm \A'|\le 7wke$.
        By the pigeonhole principle, this means that $(x+\delta,y+\delta)\in \A'$ holds for some $\delta \in [0\dd 7wke)$.
        Then, the greedy nature of $\A'$ guarantees that $\A'$ must greedily match $X[x+\delta\dd x+\ell)$ and $Y[y+\delta\dd y+\ell)$
        and, in particular, $X[x+7wke\dd x+\ell)$ and $Y[y+7wke\dd y+\ell)$.
        Hence, $X[x\dd x+\ell)$ contributes at least $\ell-7wke$ pairs to $M$.
        The total contribution of all maximal fragments $X[x\dd x+\ell)$ is $|X|-k-(k+1)\cdot 7wke \ge |X|-15wk^2e$.

        As for the efficient algorithm, we use \cref{lem:build_greedy} to build $\A$ (if one exists)
        in $\Oh(n+kw)$ time.
        The remaining steps of the construction above can be easily implemented in $\Oh(n)$ time.
    \end{proof}

%% file: src/horiz_per.tex

In this section, we discuss periodic sections of the input forests $F$ and $G$.  Specifically, we look at ``horizontal periodicity,'' which refers to the case when the children of a given node are a periodic sequence of subtrees.  If a horizontal period is repeated a large number of times in both $F$ and $G$, we know that any minimum cost alignment of $F$ and $G$ must align the periodic subtrees in a predictable way. We then argue that we can consider two forests $F'$ and $G'$ that are the exact same as $F$ and $G$ with a bound on the number of repetitions of horizontal periodic subtrees.

For ease of analysis, we do these periodic reductions on the parentheses representations of $\str{F}$ and $\str{G}$ with tree alignments. Furthermore, we borrow some useful definitions and results from~\cite{DBLP:journals/siamcomp/BannaiIINTT17} on maximal periodic substrings in strings, called runs.

\begin{definition}[Runs,~\cite{DBLP:journals/siamcomp/BannaiIINTT17}]
    A \emph{run} in a string $S$ is a triple $(i, j, p)$ such that $S[i\dd j)$ is periodic with $\per(S[i\dd j)) = p$ and maximal in length, i.e., $S[i-1] \ne S[i - 1 + p]$ and $S[j] \ne S[j - p]$. We call $\frac{j - i}{p}$ the power of the run.
    Let $\textsf{Runs}(S)$ denote the entire list of runs in $S$.
\end{definition}

\begin{theorem}[Theorem 9 of~\cite{DBLP:journals/siamcomp/BannaiIINTT17}]\label{thm:numruns}

    The number of runs in a string $S$ of length $n$ is less than~$n$.
\end{theorem}

\begin{theorem}[Algorithm 1 of~\cite{DBLP:journals/siamcomp/BannaiIINTT17}]\label{thm:computeruns}
    Given a string $S$ with $n = |S|$, it is possible to compute all runs of $S$ in $\Oh(n)$ time.
\end{theorem}

Using the results of the above theorems, we can find all horizontal periodic substrings of $\str{F}$ and $\str{G}$ in $\Oh(n)$ time. Ideally, we could just iterate through a list of computed runs and reduce the exponent of each run in constant time. However, some runs may be overlapping, and so by reducing one run, we may significantly change the length of a previously reduced run, or even worse introduce a new run to the string. So, by requiring that the period of runs that we reduce is small, we also can guarantee that the overlap of runs will be manageable.

\begin{lemma}[Fact 2.2.4,~\cite{phd}]\label{lem:runoverlap}
    Any two distinct runs $(i_1, j_1, p_1)$ and $(i_2, j_2, p_2)$ in a string $S$ satisfy $|[i_1 \dd j_1) \cap [i_2 \dd j_2)| < p_1 + p_2 - \gcd(p_1, p_2)$.
\end{lemma}

\begin{corollary}\label{cor:boundedperiodoverlap}
    Consider two distinct runs $(i_1, j_1, p_1)$ and $(i_2, j_2, p_2)$ in a string $S$ with periods $p_1,p_2\le 4k$
    and exponents $\frac{j_1-i_1}{p_1}, \frac{j_2-i_2}{p_2}\ge 8k$. 
    If $i_1\le i_2$, then $i_2 - 8k < j_1 < j_2$.
\end{corollary}
\begin{proof}
    By \cref{lem:runoverlap}, we have $|[i_1 \dd j_1) \cap [i_2 \dd j_2)|< p_1 + p_2 - \gcd(p_1, p_2) < 8k$.
    Due to $|[i_2 \dd j_2)| \ge 8kp_2 \ge 8k$, this means that $[i_2 \dd j_2)$ is not contained in $[i_1 \dd j_1)$, i.e., $j_1 < j_2$.
    Moreover, $j_1 - i_2 = \min(j_1,j_2)-\max(i_1,i_2) \leq |[i_1 \dd j_1) \cap [i_2 \dd j_2)|\le 8k$.
\end{proof}

Since horizontal periodic reductions take place on subtrees, we require that any periodic substring we reduce has a balanced string period in the parentheses representation.  For a string $X$, we define a balance function $\sigma(X)$ which is mapped to the minimum number of rotations of $X$ needed to make it balanced. If $X$ cannot be balanced, $\sigma(X) := -1$.  For example, let $X = \texttt{)][()(}$, then $\sigma(X) = 2$ since $\rot^2(X) = \texttt{[()()]}$; let $Y = \texttt{(((((}$, then $\sigma(Y) = -1$ since there is no rotation that balances $Y$. We give the following lemma for computing the $\sigma$ function.

\begin{lemma}
    $\sigma(X)$ can be computed in time $\Oh(|X|)$.
    \label{lem:balanceX}
\end{lemma}

\begin{proof}
    To compute $\sigma(X)$, we run a folklore stack matching algorithm for checking balance on $X$. The algorithm iterates across the characters of $X$ and every opening parenthesis the algorithm comes across gets pushed onto the stack.  When the algorithm reaches a closing parenthesis, it pops an opening parenthesis from the stack and matches the two.  We modify this basic algorithm so that if the stack is empty when a closing parenthesis is reached, it ignores the failed match and continues to the next character. Additionally, the algorithm will keep track of the largest index $m < |X|$ such that $X[m]$ is a closing parenthesis which had no match due to the stack being empty when this parenthesis was reached. Clearly we must rotate $X$ at least $m + 1$ times in order for closing parenthesis $X[m]$ to potentially have an opening parenthesis before it in $X$ to match to. Moreover, since $m$ is the index of the last unmatched closing parenthesis in $X$ then for any index $i > m$ there can only be unmatched opening parenthesis. Therefore, rotating more than $m + 1$ times will only place unmatched opening parenthesis at the end of $X$ with no potential matches afterwards. Thus, we rotate exactly $m+1$ times, and run the stack matching algorithm on $\rot^{m+1}(X)$.  If the rotated string is balanced and all matching parentheses share the same parenthesis type, then $\sigma(X) = m + 1$.  Otherwise, $\sigma(X) := -1$.
\end{proof}

Note that the above lemma shows that checking balance and finding the minimum number of rotations to balance a string only takes linear time. Since we only want to find the balance of periodic substrings of length at most $4k$, then computing $\sigma$ will only take time $\Oh(k)$.

From the above definition and results, we can now give our algorithm to reduce horizontal periodicity seen below in Algorithms~\ref{alg:FilterRuns}, ~\ref{alg:SyncOccur}, and~\ref{alg:SyncReductions}. The goal of \cref{alg:SyncReductions} is to take two input forests $F$ and $G$ and output forests $F'$ and $G'$ which are the same as the input forests except any long horizontal periodic synchronized occurrences are reduced to an exponent of at most $18k$. First, \cref{alg:SyncOccur} uses \cref{alg:FilterRuns} as a subroutine to compute all runs of the parentheses representations of input forests $F$ and $G$, sorted by starting index. Before the algorithm can actually reduce any pair of runs, it must make sure the runs adhere to strict requirements. First, to make sure the overlap between runs that we want to reduce is not too large compared to the run size, we only reduce runs whose period is at most $4k$ in length and whose exponent is larger than $16k$. \cref{alg:FilterRuns} does exactly this and returns a filtered list of runs that meet these criteria. Afterwards, \cref{alg:SyncOccur} iterates over pairs of runs, one from $\str{F}$ and one from $\str{G}$ in sorted order by their starting index.  Note that we only want to reduce periodic substrings who have $2k$-synchronized occurrences; otherwise we cannot align the matching periodic subtrees with less than $k$ edits. If we find a very large run in one input forest with no synchronized occurrence in the other input forest, the algorithm should not reduce these occurrences. 

After checking for these properties in each run, \cref{alg:SyncOccur} does one final check that the string period of the run is a balanced parentheses string, i.e., it actually corresponds to a sequence of subtrees in the original forests $F$ and $G$.  If a run's string period is not balanced, we do not actually have horizontal periodicity and cannot reduce these runs as easily (these runs are instead handled in Sections~\ref{sec:ver} and~\ref{sec:fur}). If a run is not balanced, did not have a small enough period, or did not have a large enough exponent, the run is simply copied over to the corresponding output forest $F'$ or $G'$ without any changes. In the case that \cref{alg:SyncOccur} does find a pair of runs representing $2k$-synchronized occurrences of large horizontal periodicity, a triple representing the pair of runs and their common period and exponent is added to a set $S$.  Then, \cref{alg:SyncReductions} iterates across set $S$ and reduces all found $2k$-synchronized occurrences to only $14k$ repetitions of the string period. \cref{alg:SyncReductions} copies all characters in $\str{F}$ to a new string $\str{F'}$ except any substrings contained in $2k$-synchronized occurrences of horizontal periodicity, in which instead the algorithm skips all but the last $14k$ repetitions of the string period. Once we have the output forests $F'$ and $G'$, it is still necessary to show that reduced runs have not changed the tree edit distance at all nor introduced any new horizontal periodicity. The rest of the section is devoted to the analysis of these three algorithms and their outputs. 

 \begin{algorithm}
    $R \gets \textsf{SortByStartingIndex}(\textsf{Runs}(S))$\;
    $R' \gets \textrm{empty list}$\;
    \For{$(i,j,p) \in R$}{
        \If{$p \leq 4k$ \KwSty{and} $\frac{j-i}{p} \geq 16k$}{
            $R'.\textsf{append}\big((i,j,p)\big)$\;
        }
    }
    \KwRet{$R'$}\;
     \caption{$\mathsf{FilterRuns}(S)$. }
     \label{alg:FilterRuns}
 \end{algorithm}

 \begin{algorithm}
    $R_F \gets \textsf{FilterRuns}(\str{F}), R_G \gets \textsf{FilterRuns}(\str{G})$\;\label{ln:2.1}
    $\ell_F, \ell_G \leftarrow 0$\;
    $S \leftarrow$ empty list\;
    \While{$\ell_F < |R_F|$ \KwSty{and} $\ell_G < |R_G|$}{\label{ln:2.4}
        $(i_F, j_F, p_F) \leftarrow R_F[\ell_F], (i_G, j_G, p_G) \leftarrow R_G[\ell_G]$\;
         $X \leftarrow \str{F}[i_F \dd i_F + p_F), Y \leftarrow \str{G}[i_G \dd i_G + p_G)$\;
        $e \leftarrow \lfloor\frac{|[i_F \dd j_F) \cap [i_G \dd j_G)|}{p_F}\rfloor$\;
         \uIf{$e \geq 16k$ \KwSty{and} $\exists \alpha \leq p_F$, $rot^\alpha(X) = Y$ \KwSty{and} $\sigma(X) \ge 0$}
         {\label{ln:2.8}
            \tcp{Add synchronized occurrences to set $S$}
            $S.append((\max(i_F, i_G), p_F, e-2k))$\;\label{ln:2.9}
         }
        \tcp{Go to a new run}
        \If{$j_F < j_G$}
        {
            $\ell_F \leftarrow \ell_F + 1$\;\label{ln:2.11}
        }
        \Else{
            $\ell_G \leftarrow \ell_G + 1$\;\label{ln:2.13}
        }
     }
     \KwRet{$S$}\;\label{ln:2.18}
     \caption{$\mathsf{SyncOccurrences}(F, G)$. }
     \label{alg:SyncOccur}
 \end{algorithm}

We begin by proving that the runtime of \cref{alg:SyncOccur} is linear. Since we know that computing all runs in a string only takes $\Oh(n)$ time from \cref{thm:computeruns}, the call to \cref{alg:FilterRuns} do not cause any issues. As for the main loop of \cref{alg:SyncOccur}, either the algorithm finds a reason that a pair of runs do not contain $2k$-synchronized occurrences and moves on to the next run, or the algorithm creates a synchronized occurrence triple and adds it to the output list $S$. Therefore, we must be careful that all the checks done in each case do not take superlinear total time.

\begin{lemma}
    Given forests $F$ and $G$ of total size $n$, \cref{alg:SyncOccur} runs in time $\Oh(n)$.
    \label{lem:hoccurRuntime}
\end{lemma}

\begin{proof}
    The runtime of \cref{alg:SyncOccur} is mostly affected by the call to \cref{alg:FilterRuns} for sorting and computation of runs in step \ref{ln:2.1} as well as the number of iterations and work done in each iteration of the $\KwSty{while}$ loop spanning steps~\ref{ln:2.4}--\ref{ln:2.18}. First in step \ref{ln:2.1}, by \cref{thm:computeruns}, computing all runs of both forests can be done in time $\Oh(n)$, and there are less than $2n$ runs total by \cref{thm:numruns}. Therefore, sorting all $\Oh(n)$ runs by their starting index can be done using a radix sort in time $\Oh(n)$. The rest of \cref{alg:FilterRuns} filters runs in linear time to make sure the period is small enough and the exponent is large enough. Due to these filtering steps and Corollary~\ref{cor:boundedperiodoverlap}, we can observe that $ |R_F|, |R_G| = \Oh(n/k)$.  
    
    We look individually at the work done in the subsequent iterations of the $\KwSty{while}$ loop. An important observation is that one of the run counters $\ell_F$ or $\ell_G$ always increments by 1 (steps \ref{ln:2.11} and~\ref{ln:2.13}) per iteration of the loop.
    Line~\ref{ln:2.8} involves three checks.  First, the overlap of the runs is computed, which takes constant time.  Then, we check if a rotation of string period $X$ is equal to a rotation of string period $Y$. Since $p_F \leq 4k$, checking all rotations takes time $\Oh(k)$ using a rolling hash function~\cite{KR87}. The last check just requires computing $\sigma(X)$, which by Lemma~\ref{lem:balanceX} takes time $\Oh(|X|) = \Oh(k)$. The remaining instructions within each iteration of the $\KwSty{while}$ loop take $\Oh(1)$ time, so each iteration costs $\Oh(k)$ time in total. As mentioned earlier, $ |R_F|, |R_G| = \Oh(n/k)$, and hence the entire algorithm is done in linear time.
\end{proof}

Now, we prove that \cref{alg:SyncOccur} indeed finds $2k$-synchronized occurrences of horizontal periodicity with large exponent in $F$ and $G$. Any triple $(i, p, e)$ added to set $S$ by   \cref{alg:SyncOccur} satisfies these requirements as per the checks in step \ref{ln:2.8}. Therefore, the following proof is fairly straightforward and mostly just formalizes this intuition.

\begin{lemma}
    Given forests $F, G$, let $S = \mathsf{SyncOccurrences}(F, G)$.  Then for any $(i, p, e) \in S$, there must be runs $(i_F, j_F, p)$ in $\str{F}$, $(i_G, j_G, p)$ in $\str{G}$ each containing $2k$-synchronized occurrences of string $Q^e$ in $\str{F}$ and $\str{G}$ where $|Q| = p \leq 4k$, $e \geq 14k$, $i = \max(i_F, i_G)$, and $Q$ is balanced.
    \label{lem:SyncOccur}
\end{lemma}

\begin{proof}
    First, we observe that triples $(i, p, e)$ are only inserted to $S$ in step~\ref{ln:2.9} of \cref{alg:SyncOccur}. Let $(i_F, j_F, p_F)$ and $(i_G, j_G, p_G)$ be the runs considered by the algorithm during such an insertion, and note that $p = p_F=p_G \le 4k$ and $i = \max(i_F, i_G)$. 
    Moreover, from the checks done in step~\ref{ln:2.8}, the string periods $X = \str{F}[i_F \dd i_F + p)$ and $Y = \str{G}[i_G \dd i_G + p)$ are the same up to rotation, i.e., $\rot^{\alpha}(X) = Y$ for some $\alpha \in [0\dd p)$. Hence, $\str{F}[i_F + \alpha \dd i_F + \alpha + p) = \str{G}[i_G \dd i_G + p)$. From step~\ref{ln:2.8}, we also know that $e + 2k \geq 16k$,  which implies that the length of the set of overlapping indices between the two runs has a lower bound of
    \[
       \left| [i_F \dd j_F) \cap [i_G \dd j_G) \right|\geq p_F \left\lfloor\tfrac{|[i_F \dd j_F) \cap [i_G \dd j_G)|}{p_F}\right\rfloor = p (e+2k) \geq 16kp.
    \]
      Without loss of generality, we assume that $i_F \leq i_G$. Then, there exists some $m \in \mathbb{N}$ such that $i_G \leq i_F' = i_F + \alpha + pm \leq i_G + p$ and $\str{F}[i_F' \dd i_F' + (e+k)p) = \str{G}[i_G \dd i_G + (e+k)p)$. Clearly, we have $2k$-synchronized occurrences of $Y^{e+k}$, with $e+k\geq 15k$. The check in step \ref{ln:2.8} guarantees that $\sigma(X) \ne -1$, and by extension, $\sigma(Y) \ne -1$. Recall that $\sigma(Y)$ is the minimum number of rotations needed for $Y$ to be a balanced parentheses string, i.e. $Q:=\rot^{\sigma(Y)}(Y) = \str{G}[i_G + \sigma(Y) \dd i_G + \sigma(Y) + p)$ is balanced. Since $\sigma(Y) \leq |Y| = p$, we have that $\str{F}[i_F' + \sigma(Y) \dd i_F' + \sigma(Y) + ep) = \str{G}[i_G + \sigma(Y) \dd i_G + \sigma(Y) + ep)$ are $2k$-synchronized occurrences of $Q^e$ in $\str{F}$ and $\str{G}$ where $|Q| = p \leq 4k$ and $e \geq 14k$.
\end{proof}

\begin{algorithm}
    $S \leftarrow \mathsf{SyncOccurrences}(F, G)$\;\label{ln:3.1}
    $s_F, s_G \leftarrow \varepsilon$\;
    $i \leftarrow 0$\;
    \For{$(i', p', e')\in S$}
    {\label{ln:3.4}
        \tcp{Copy from start of previous synchronized occurrences to start of next synchronized occurrences}
        $s_F \leftarrow s_F \cdot \str{F}[i \dd i')$\;
        $s_G \leftarrow s_G \cdot \str{G}[i \dd i')$\;
        \tcp{Reduce synchronized occurrences to exponent of $14k$}
        $i \leftarrow i' + p'(e' - 14k)$\;\label{ln:3.7}
    }
    $s_F \leftarrow s_F \cdot \str{F}[i \dd |\str{F}|)$\;
    $s_G \leftarrow s_G \cdot \str{G}[i \dd |\str{G}|)$\;
    \KwRet{$s_F, s_G$}\;
    \caption{$\mathsf{SyncReductions}(F, G)$}
    \label{alg:SyncReductions}
\end{algorithm}

Next, we prove the main statement that we need in order to show that the outputted forests $F'$ and $G'$ from \cref{alg:SyncReductions} have the same edit distance as the original forests $F$ and $G$.  Conceptually, the proof just shows that for long synchronized occurrences of a run in both forests, any two minimal cost alignments of $F$ and $G$ must match long segments of these periodic sections together.  In fact, we show that given any minimal cost tree alignment $\A$ of $F$ and $G$, we can build a minimal cost tree alignment of $F'$ and $G'$ that follows $\A$ almost exactly.



\begin{lemma}\label{lem:horizreductions}
    Consider forests $F, G$ such that $\str{F}[\alpha_F \dd \beta_F) = \str{G}[\alpha_G \dd \beta_G) = Q^e$ 
    for a balanced string $Q$ of length $0 < |Q| \leq 4k$, an integer exponent $e \geq 6k$, and indices $\alpha_F,\alpha_G,\beta_F,\beta_G$ satisfying 
    $|\alpha_F - \alpha_G| \leq 2k$. 
    Let $F',G'$ be forests such that $\str{F'} = \str{F}[0 \dd \alpha_F) \cdot Q^{e'} \cdot \str{F}[\beta_F \dd |\str{F}|)$ and 
    $\str{G'} = \str{G}[0 \dd \alpha_G) \cdot Q^{e'} \cdot \str{G}[\beta_G \dd |\str{G}|)$ for some integer exponent $e'\ge 6k$.
    Then, $\ted_{\le k}(F, G) = \ted_{\le k}(F', G')$.
\end{lemma}

\begin{figure}
    \centering
    \input{figs/horizontal_reduction}
    \caption{A $2k$-synchronized horizontal periodicity with $|Q| = 8$ is demonstrated. Each gray and red circle correspond to a balanced period block in $\str{F}$ and $\str{G}$ similar to the zoomed portion. According to Lemma~\ref{lem:horizreductions} we can reduce the period exponent to $6k$ by removing the blue part of $\str{F}$ and $\str{G}$ so that the tree edit distance $\ted_{\le k}(F,G)$ remains unchanged.}
    \label{fig:horizontal_reduction}
\end{figure}
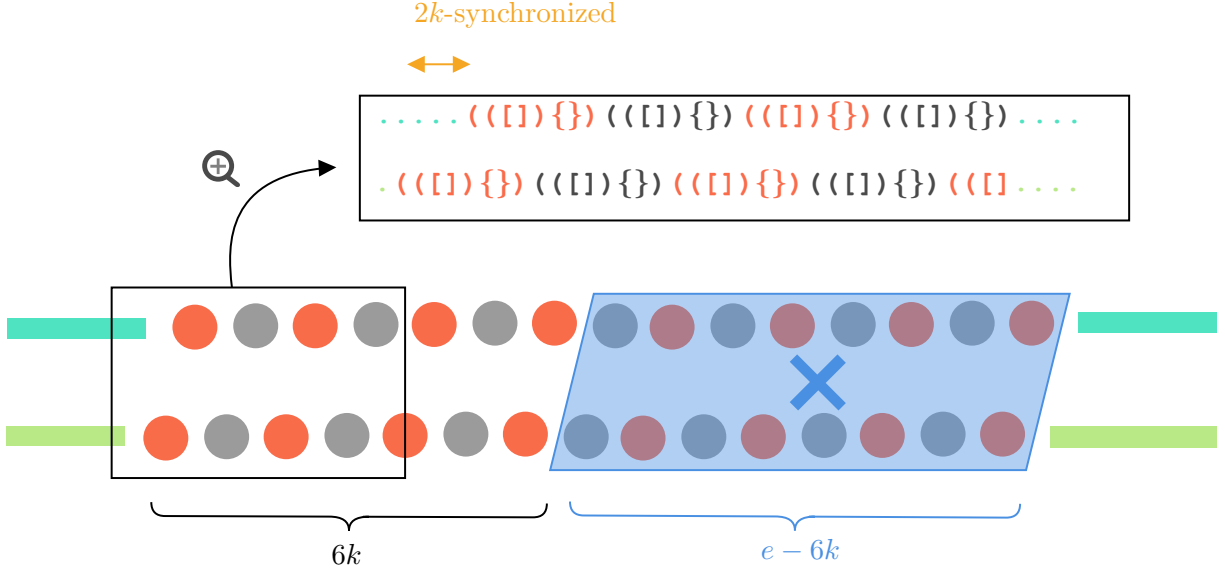

\begin{proof}
    We assume without loss of generality that $Q$ is primitive (otherwise, we replace $Q$ by its primitive root)
    and denote $p := |Q| \leq 4k$. Let $\A$ be an optimal tree alignment such that $\ted(F,G)= \ted_\A(F,G) \le k$. 

    \begin{claim}
        There exist $i_F,i_G\in [0\dd 3k]$ such that $(\alpha_F+i_F\cdot p, \alpha_G+i_G\cdot p)\in \A$ and $j_F,j_G\in [0\dd 3k]$ such that $(\beta_F-j_F\cdot p, \beta_G-j_G\cdot p)\in \A$.
    \end{claim}
    \begin{proof} Let $(x_F,x_G)\in \A$ be the leftmost element of $\A$ such that $x_F \ge \alpha_F$ and $x_G \ge \alpha_G$.
    By symmetry between $F$ and $G$, we assume without loss of generality that $x_F = \alpha_F$.
    Consider the $k+1$ occurrences of $Q$ starting at positions $\alpha_F+i\cdot p$ for $i\in [0\dd k]$.
    Since $Q$ is balanced, the alignment $\A$ (of cost at most $k$) matches at least one of them exactly;
    we can thus define $i_F\in [0\dd k]$ so that $\A$ matches $\str{F}[\alpha_F+i_F\cdot  p\dd \alpha_F+(i_F+1)\cdot p)$ exactly to some fragment $\str{G}[y_G\dd y_G+p)$.
    Due to $(x_F,x_G)\in \A$, the non-crossing property of $\A$ implies that $y_G \ge \alpha_G$.
    Moreover, since $\ted_\A(F, G) \leq k$ and $\str{F}[\alpha_F + i_F \cdot p] \sim_\A \str{G}[y_G]$, we have $y_G \le (\alpha_F+i_F \cdot p)+2k \le \alpha_F + kp+2k \le \alpha_G + kp + 4k \le \alpha_G + 3kp$,
    where the last inequality follows from $p\ge 2$ (recall that $Q$ is balanced, so its length is even).
    Furthermore, since $Q$ is primitive (i.e., distinct from all its non-trivial cyclic rotations), we conclude that $y_G = \alpha_G+i_G \cdot p$ for some $i_G\in [0\dd 3k]$. 
    The second claim is symmetric (with respect to reversal).
    \end{proof}

    Observe that $\str{F}[\alpha_F+i_F p\dd \beta_F-j_F p)=Q^{d_F}$ for $d_F := e-j_F -i_F$
    and, symmetrically, $\str{G}[\alpha_G+i_G p\dd \beta_G-j_G  p)=Q^{d_G}$ for $d_G := e-j_G -i_G$.
    We denote $d= \min(d_F,d_G)$, observe that $d\ge e-6k\ge 0$, and construct a tree alignment $\A'$ so that it
    \begin{itemize}
        \item aligns $\str{F}[0\dd \alpha_F+i_F  p)$ with $\str{G}[0\dd \alpha_G+i_G  p)$ in the same way as $\A$ does;
        \item matches $\str{F}[\alpha_F+i_F p\dd \alpha_F+(i_F+d)p)=Q^d$ with $\str{G}[\alpha_G+i_G p\dd \alpha_G+(i_G+d)p)=Q^d$;
        \item deletes $\str{F}[\alpha_F+(i_F+d)p \dd \beta_F-j_Fp)=Q^{d_F-d}$ and $\str{G}[\alpha_G+(i_G+d)p \dd \beta_G-j_Gp)=Q^{d_G-d}$;
        \item aligns $\str{F}[\beta_F-j_Fp\dd |\str{F}|)$ with $\str{G}[\beta_G-j_Gp\dd |\str{G}|)$ in the same way as $\A$ does.
    \end{itemize}
    Note that $\A'$ is a tree alignment: for any node of $F$, the corresponding parentheses are either both outside $\str{F}[\alpha_F\dd \beta_F)$ (and then they are handled as in $\A$) or both contained in a single copy of $Q$ (which is either deleted or matched perfectly to a copy of $Q$ in $\str{G}$).
    Moreover, the cost of $\A'$ does not exceed the cost of $\A$: the two alignments only differ in how they align $Q^{d_F}$ with $Q^{d_G}$,
    and $\A'$ provides an optimum alignment of these fragments. 

    Now, if the exponent $e$ of $Q^e=\str{F}[\alpha_F\dd \beta_F)=\str{G}[\alpha_G\dd \beta_G)$ is modified to $e'\ge 6k$,
    we can interpret this as modifying exponent $d$ of the fragments $Q^d$ matched perfectly by $\A'$ to $d'=d+e'-e \ge 0$.
    Thus, $\A'$ can be trivially adapted without modifying its cost and hence $\ted(F',G')\le \ted_{\A'}(F,G)=\ted(F,G)$.
    The converse inequality follows by symmetry between $(F,G)$ and $(F',G')$.
\end{proof}

From the previous lemma, it is clear that reducing a long run does not affect edit distance. Utilizing this idea, we finally prove that the forests outputted by \cref{alg:SyncReductions} have the same edit distance as the input forests and avoid $2k$-synchronized runs with an exponent more than $14k$ without changing the edit distance.

\begin{definition}[Synchronized horizontal periodicity]
    We say that forests $F,G$ \emph{avoid synchronized horizontal $k$-periodicity} if there is no non-empty balanced string $Q$ of length $|Q|\le 4k$
    such that $Q^{18k}$ has $2k$-synchronized occurrences in $\str{F},\str{G}$.
\end{definition}

\begin{proposition}[Avoiding Horizontal $k$-Periodicity]\label{prp:hor}
    There exists an $\Oh(n)$-time algorithm that, given labeled forests $F,G$ of total size $n$ and an integer $k\in \Zp$, produces labeled forests $F',G'$
    that satisfy $\ted_{\le k}(F,G)=\ted_{\le k}(F',G')$ and, moreover, avoid synchronized horizontal $k$-periodicity.
\end{proposition}

\begin{proof}
    Let $s_F, s_G = \mathsf{SyncReductions}(F, G)$.  By Lemma~\ref{lem:SyncOccur}, for any triple $(i, p', e') \in S$ of \cref{alg:SyncReductions}, there is  $2k$-synchronized occurrences in $F, G$ of some $Q^{e'}$ with $|Q| = p$ satisfying the constraints of Lemma~\ref{lem:horizreductions} within runs $r_1 = (i_1, j_1, p)$ and $r_2 = (i_2, j_2, p)$ of $F$ and $G$, respectively. In steps 7 and 8, we move the start of the synchronized occurrences forward by $(e' - 14k)p$ and do not copy the skipped over indices to $s_F$ and $s_G$.  This is equivalent to reducing the exponent of $r_1$ and $r_2$ by $e' - 14k$ since $\frac{j_1- (i_1 + (e'-14k)p)}{p} = \frac{j_1 - i_1}{p} - (e' - 14k)$. Furthermore, since the $2k$-synchronized occurrences with exponent $e'$ lies completely within $r_1$ and $r_2$, steps 9 and 10 are actually equivalent to reducing the synchronized occurrences of $Q^{e'}$ to $Q^{14k}$ in each forest. Note that since $Q$ is balanced, there exist forests $F', G'$ such that $\str{F'} = s_F$ and $\str{G'} = s_G$. Therefore, by Lemma~\ref{lem:horizreductions}, we have that $\ted_{\le k}(F,G)=\ted_{\le k}(F',G')$.
    
    Now, we must show that $F', G'$ avoid synchronized horizontal $k$-periodicity.  Assume for contradiction that there exists $2k$-synchronized occurrences of some $Q^{18k}$ with balanced string period $Q$ such that $|Q| \leq 4k$ in $F'$ and $G'$.  If $Q^{18k}$ was initially found in $2k$-synchronized occurrences in $F$ and $G$, then since it is periodic, it must be the case that it was contained in runs $(i_F, j_F, p_F), (i_G, j_G, p_G)$, in $F$ and $G$ respectively, with $p_F, p_G \leq |Q| \leq 4k$ and $\frac{j_F -i_F}{p_F}, \frac{j_G - i_G}{p_G} \geq 18k$. Without loss of generality, let $j_F < j_G$. Since $Q^{18k}$ is a $2k$-synchronized occurrence, it must be that the overlap of these runs is at least $18k - 2k = 16k$ and so, $j_F > i_G + 8k$. Let $(i_G', j_G', p')$ be the run in $R_G$ preceding $(i_G, j_G, p_G)$.  Note that by Corollary~\ref{cor:boundedperiodoverlap}, $j_G' < i_G + 8k < j_F$. Furthermore, we only increment the run counter of the forest whose current run ends before the other forest's current run. Since $j_G' < j_F$ and $j_F < j_G$, we will always reach an iteration of \cref{alg:SyncOccur} that considers the pair of runs $(i_F, j_F, p_F), (i_G, j_G, p_G)$. Since the overlap of the two runs is at least $16k$, \cref{alg:SyncOccur} will therefore add triple $(\max(i_F, i_G), p_F, e - 2k)$ to $S$ for some $e \geq 16k$. Then \cref{alg:SyncReductions} reduces the $2k$-synchronized occurrences containing $Q^{16k}$ to an exponent of at most $14k$ in steps 9 and 10. Therefore, if any $2k$-synchronized occurrences with exponent at least $18k$ and period at most $4k$ is present in $F$ and $G$, it will be reduced to an exponent of at most $14k$ in $F', G'$. 
    
    Now, if $F', G'$ do not avoid synchronized horizontal $k$-periodicity it must be the case that by reducing the exponent of overlapping runs, we created new $2k$-synchronized occurrences $Q^{18k}$ where $|Q| \leq 4k$. Clearly, $Q^{18k}$ must lie in some new run $r_Q' = (i_Q', j_Q', p_Q')$ in $F'$ since it is periodic. We will show that this is not possible for such a run to form due to the small overlap between periodic substrings with period at most $4k$. By Corollary~\ref{cor:boundedperiodoverlap}, the overlap between $r_Q'$ and any other reduced runs is at most $8k$. Since $r_Q'$ has length at least $18k$, $r_Q'$ may overlap at most two reduced runs. We refer to the two reduced runs as $r_1' = (i_1', j_1', p_1')$ and $r_2' = (i_2', j_2', p_2')$ and without loss of generality assume they are in $\str{F'}$. Let $\str{F'}[j_1' - \beta_1 \dd j_1']$ be the overlap between $r_1'$ and $r_Q'$, and similarly let $\str{F'}[i_2' \dd i_2' + \beta_2]$ be the overlap between $r_2'$ and $r_Q'$. Since $r_1'$ and $r_2'$ are reduced, there are two corresponding runs $r_1 = (i_1, j_1, p_1), r_2 = (i_1, j_2, p_1)$ in $\str{F}$, such that $\str{F}[j_1 - \beta_1 \dd j_1) = \str{F'}[j_1'-\beta_1 \dd j_1')$ and $\str{F}[i_2 \dd i_2 + \beta_2) = \str{F'}[i_2' \dd i_2' + \beta_2)$. Note that since we do not change any characters between any runs, we also know that the middle substrings $\str{F'}[j_1' \dd i_2'] = \str{F}[j_1 \dd i_2]$ are equal as well.  Combining these three substrings we have that $\str{F'}[i_Q' \dd j_Q') = \str{F}[j_1 - \beta_1 \dd i_2 + \beta_2)$, which implies that $Q^{18k}$ is a periodic substring of $\str{F}$. In other words, $Q^{18k}$ is contained in a run in $\str{F}$ before any reductions occur, which is a contradiction.
    
    Finally, we discuss the runtime of \cref{alg:SyncReductions}.  First \cref{alg:SyncReductions} calls \cref{alg:SyncOccur} in step~\ref{ln:3.1}, and by Lemma~\ref{lem:hoccurRuntime}, this takes time $\Oh(n)$. Now, we consider the loop in steps~\ref{ln:3.4}--\ref{ln:3.7}. We copy substrings from $\str{F}$ and $\str{G}$ of the start of one pair of synchronized occurrences to the start of the next pair of synchronized occurrences.  Note that synchronized occurrences we copy have periods at most $4k$ and exponents at least $14k$, and so by \cref{cor:boundedperiodoverlap}, we know that no two synchronized occurrences will start at the same index. Therefore, we only copy each character of $\str{F}$ and $\str{G}$ at most once across the entire algorithm, and so \cref{alg:SyncReductions} takes time $\Oh(n + |\str{F}| + |\str{G}|) = \Oh(n)$.    
\end{proof}

%% file: figs/horizontal_reduction.tex
\tikzset{every picture/.style={line width=0.75pt}} 

\begin{tikzpicture}[x=0.75pt,y=0.75pt,yscale=-1,xscale=1]

\draw  [draw opacity=0][fill={rgb, 255:red, 184; green, 233; blue, 134 }  ,fill opacity=1 ] (40,264.68) -- (100,264.68) -- (100,274.68) -- (40,274.68) -- cycle ;
\draw  [draw opacity=0][fill={rgb, 255:red, 80; green, 227; blue, 194 }  ,fill opacity=1 ] (40.67,210.33) -- (110.67,210.33) -- (110.67,220.68) -- (40.67,220.68) -- cycle ;
\draw  [draw opacity=0][fill={rgb, 255:red, 249; green, 108; blue, 74 }  ,fill opacity=1 ] (123.99,214.67) .. controls (123.99,208.41) and (129.07,203.33) .. (135.33,203.33) .. controls (141.59,203.33) and (146.67,208.41) .. (146.67,214.67) .. controls (146.67,220.93) and (141.59,226.01) .. (135.33,226.01) .. controls (129.07,226.01) and (123.99,220.93) .. (123.99,214.67) -- cycle ;
\draw  [draw opacity=0][fill={rgb, 255:red, 155; green, 155; blue, 155 }  ,fill opacity=1 ] (154.66,214.01) .. controls (154.66,207.74) and (159.73,202.67) .. (165.99,202.67) .. controls (172.26,202.67) and (177.33,207.74) .. (177.33,214.01) .. controls (177.33,220.27) and (172.26,225.34) .. (165.99,225.34) .. controls (159.73,225.34) and (154.66,220.27) .. (154.66,214.01) -- cycle ;
\draw  [draw opacity=0][fill={rgb, 255:red, 249; green, 108; blue, 74 }  ,fill opacity=1 ] (184.66,214.01) .. controls (184.66,207.74) and (189.73,202.67) .. (195.99,202.67) .. controls (202.26,202.67) and (207.33,207.74) .. (207.33,214.01) .. controls (207.33,220.27) and (202.26,225.34) .. (195.99,225.34) .. controls (189.73,225.34) and (184.66,220.27) .. (184.66,214.01) -- cycle ;
\draw  [draw opacity=0][fill={rgb, 255:red, 155; green, 155; blue, 155 }  ,fill opacity=1 ] (215.32,213.34) .. controls (215.32,207.08) and (220.4,202) .. (226.66,202) .. controls (232.92,202) and (238,207.08) .. (238,213.34) .. controls (238,219.6) and (232.92,224.68) .. (226.66,224.68) .. controls (220.4,224.68) and (215.32,219.6) .. (215.32,213.34) -- cycle ;
\draw  [draw opacity=0][fill={rgb, 255:red, 249; green, 108; blue, 74 }  ,fill opacity=1 ] (244.66,213.34) .. controls (244.66,207.08) and (249.73,202) .. (255.99,202) .. controls (262.26,202) and (267.33,207.08) .. (267.33,213.34) .. controls (267.33,219.6) and (262.26,224.68) .. (255.99,224.68) .. controls (249.73,224.68) and (244.66,219.6) .. (244.66,213.34) -- cycle ;
\draw  [draw opacity=0][fill={rgb, 255:red, 155; green, 155; blue, 155 }  ,fill opacity=1 ] (275.32,212.67) .. controls (275.32,206.41) and (280.4,201.33) .. (286.66,201.33) .. controls (292.92,201.33) and (298,206.41) .. (298,212.67) .. controls (298,218.93) and (292.92,224.01) .. (286.66,224.01) .. controls (280.4,224.01) and (275.32,218.93) .. (275.32,212.67) -- cycle ;
\draw  [draw opacity=0][fill={rgb, 255:red, 249; green, 108; blue, 74 }  ,fill opacity=1 ] (305.32,212.67) .. controls (305.32,206.41) and (310.4,201.33) .. (316.66,201.33) .. controls (322.92,201.33) and (328,206.41) .. (328,212.67) .. controls (328,218.93) and (322.92,224.01) .. (316.66,224.01) .. controls (310.4,224.01) and (305.32,218.93) .. (305.32,212.67) -- cycle ;
\draw  [draw opacity=0][fill={rgb, 255:red, 155; green, 155; blue, 155 }  ,fill opacity=1 ] (335.99,214.01) .. controls (335.99,207.74) and (341.07,202.67) .. (347.33,202.67) .. controls (353.59,202.67) and (358.67,207.74) .. (358.67,214.01) .. controls (358.67,220.27) and (353.59,225.34) .. (347.33,225.34) .. controls (341.07,225.34) and (335.99,220.27) .. (335.99,214.01) -- cycle ;
\draw  [draw opacity=0][fill={rgb, 255:red, 249; green, 108; blue, 74 }  ,fill opacity=1 ] (364.66,214.67) .. controls (364.66,208.41) and (369.73,203.33) .. (375.99,203.33) .. controls (382.26,203.33) and (387.33,208.41) .. (387.33,214.67) .. controls (387.33,220.93) and (382.26,226.01) .. (375.99,226.01) .. controls (369.73,226.01) and (364.66,220.93) .. (364.66,214.67) -- cycle ;
\draw  [draw opacity=0][fill={rgb, 255:red, 155; green, 155; blue, 155 }  ,fill opacity=1 ] (395.32,214.01) .. controls (395.32,207.74) and (400.4,202.67) .. (406.66,202.67) .. controls (412.92,202.67) and (418,207.74) .. (418,214.01) .. controls (418,220.27) and (412.92,225.34) .. (406.66,225.34) .. controls (400.4,225.34) and (395.32,220.27) .. (395.32,214.01) -- cycle ;
\draw  [draw opacity=0][fill={rgb, 255:red, 249; green, 108; blue, 74 }  ,fill opacity=1 ] (425.32,214.01) .. controls (425.32,207.74) and (430.4,202.67) .. (436.66,202.67) .. controls (442.92,202.67) and (448,207.74) .. (448,214.01) .. controls (448,220.27) and (442.92,225.34) .. (436.66,225.34) .. controls (430.4,225.34) and (425.32,220.27) .. (425.32,214.01) -- cycle ;
\draw  [draw opacity=0][fill={rgb, 255:red, 155; green, 155; blue, 155 }  ,fill opacity=1 ] (455.99,213.34) .. controls (455.99,207.08) and (461.07,202) .. (467.33,202) .. controls (473.59,202) and (478.67,207.08) .. (478.67,213.34) .. controls (478.67,219.6) and (473.59,224.68) .. (467.33,224.68) .. controls (461.07,224.68) and (455.99,219.6) .. (455.99,213.34) -- cycle ;
\draw  [draw opacity=0][fill={rgb, 255:red, 249; green, 108; blue, 74 }  ,fill opacity=1 ] (485.32,213.34) .. controls (485.32,207.08) and (490.4,202) .. (496.66,202) .. controls (502.92,202) and (508,207.08) .. (508,213.34) .. controls (508,219.6) and (502.92,224.68) .. (496.66,224.68) .. controls (490.4,224.68) and (485.32,219.6) .. (485.32,213.34) -- cycle ;
\draw  [draw opacity=0][fill={rgb, 255:red, 155; green, 155; blue, 155 }  ,fill opacity=1 ] (515.99,212.67) .. controls (515.99,206.41) and (521.07,201.33) .. (527.33,201.33) .. controls (533.59,201.33) and (538.67,206.41) .. (538.67,212.67) .. controls (538.67,218.93) and (533.59,224.01) .. (527.33,224.01) .. controls (521.07,224.01) and (515.99,218.93) .. (515.99,212.67) -- cycle ;
\draw  [draw opacity=0][fill={rgb, 255:red, 249; green, 108; blue, 74 }  ,fill opacity=1 ] (545.99,212.67) .. controls (545.99,206.41) and (551.07,201.33) .. (557.33,201.33) .. controls (563.59,201.33) and (568.67,206.41) .. (568.67,212.67) .. controls (568.67,218.93) and (563.59,224.01) .. (557.33,224.01) .. controls (551.07,224.01) and (545.99,218.93) .. (545.99,212.67) -- cycle ;
\draw  [draw opacity=0][fill={rgb, 255:red, 249; green, 108; blue, 74 }  ,fill opacity=1 ] (109.32,270.67) .. controls (109.32,264.41) and (114.4,259.33) .. (120.66,259.33) .. controls (126.92,259.33) and (132,264.41) .. (132,270.67) .. controls (132,276.93) and (126.92,282.01) .. (120.66,282.01) .. controls (114.4,282.01) and (109.32,276.93) .. (109.32,270.67) -- cycle ;
\draw  [draw opacity=0][fill={rgb, 255:red, 155; green, 155; blue, 155 }  ,fill opacity=1 ] (139.99,270.01) .. controls (139.99,263.74) and (145.07,258.67) .. (151.33,258.67) .. controls (157.59,258.67) and (162.67,263.74) .. (162.67,270.01) .. controls (162.67,276.27) and (157.59,281.34) .. (151.33,281.34) .. controls (145.07,281.34) and (139.99,276.27) .. (139.99,270.01) -- cycle ;
\draw  [draw opacity=0][fill={rgb, 255:red, 249; green, 108; blue, 74 }  ,fill opacity=1 ] (169.99,270.01) .. controls (169.99,263.74) and (175.07,258.67) .. (181.33,258.67) .. controls (187.59,258.67) and (192.67,263.74) .. (192.67,270.01) .. controls (192.67,276.27) and (187.59,281.34) .. (181.33,281.34) .. controls (175.07,281.34) and (169.99,276.27) .. (169.99,270.01) -- cycle ;
\draw  [draw opacity=0][fill={rgb, 255:red, 155; green, 155; blue, 155 }  ,fill opacity=1 ] (200.66,269.34) .. controls (200.66,263.08) and (205.73,258) .. (211.99,258) .. controls (218.26,258) and (223.33,263.08) .. (223.33,269.34) .. controls (223.33,275.6) and (218.26,280.68) .. (211.99,280.68) .. controls (205.73,280.68) and (200.66,275.6) .. (200.66,269.34) -- cycle ;
\draw  [draw opacity=0][fill={rgb, 255:red, 249; green, 108; blue, 74 }  ,fill opacity=1 ] (229.99,269.34) .. controls (229.99,263.08) and (235.07,258) .. (241.33,258) .. controls (247.59,258) and (252.67,263.08) .. (252.67,269.34) .. controls (252.67,275.6) and (247.59,280.68) .. (241.33,280.68) .. controls (235.07,280.68) and (229.99,275.6) .. (229.99,269.34) -- cycle ;
\draw  [draw opacity=0][fill={rgb, 255:red, 155; green, 155; blue, 155 }  ,fill opacity=1 ] (260.66,268.67) .. controls (260.66,262.41) and (265.73,257.33) .. (271.99,257.33) .. controls (278.26,257.33) and (283.33,262.41) .. (283.33,268.67) .. controls (283.33,274.93) and (278.26,280.01) .. (271.99,280.01) .. controls (265.73,280.01) and (260.66,274.93) .. (260.66,268.67) -- cycle ;
\draw  [draw opacity=0][fill={rgb, 255:red, 249; green, 108; blue, 74 }  ,fill opacity=1 ] (290.66,268.67) .. controls (290.66,262.41) and (295.73,257.33) .. (301.99,257.33) .. controls (308.26,257.33) and (313.33,262.41) .. (313.33,268.67) .. controls (313.33,274.93) and (308.26,280.01) .. (301.99,280.01) .. controls (295.73,280.01) and (290.66,274.93) .. (290.66,268.67) -- cycle ;
\draw  [draw opacity=0][fill={rgb, 255:red, 155; green, 155; blue, 155 }  ,fill opacity=1 ] (321.32,270.01) .. controls (321.32,263.74) and (326.4,258.67) .. (332.66,258.67) .. controls (338.92,258.67) and (344,263.74) .. (344,270.01) .. controls (344,276.27) and (338.92,281.34) .. (332.66,281.34) .. controls (326.4,281.34) and (321.32,276.27) .. (321.32,270.01) -- cycle ;
\draw  [draw opacity=0][fill={rgb, 255:red, 249; green, 108; blue, 74 }  ,fill opacity=1 ] (349.99,270.67) .. controls (349.99,264.41) and (355.07,259.33) .. (361.33,259.33) .. controls (367.59,259.33) and (372.67,264.41) .. (372.67,270.67) .. controls (372.67,276.93) and (367.59,282.01) .. (361.33,282.01) .. controls (355.07,282.01) and (349.99,276.93) .. (349.99,270.67) -- cycle ;
\draw  [draw opacity=0][fill={rgb, 255:red, 155; green, 155; blue, 155 }  ,fill opacity=1 ] (380.66,270.01) .. controls (380.66,263.74) and (385.73,258.67) .. (391.99,258.67) .. controls (398.26,258.67) and (403.33,263.74) .. (403.33,270.01) .. controls (403.33,276.27) and (398.26,281.34) .. (391.99,281.34) .. controls (385.73,281.34) and (380.66,276.27) .. (380.66,270.01) -- cycle ;
\draw  [draw opacity=0][fill={rgb, 255:red, 249; green, 108; blue, 74 }  ,fill opacity=1 ] (410.66,270.01) .. controls (410.66,263.74) and (415.73,258.67) .. (421.99,258.67) .. controls (428.26,258.67) and (433.33,263.74) .. (433.33,270.01) .. controls (433.33,276.27) and (428.26,281.34) .. (421.99,281.34) .. controls (415.73,281.34) and (410.66,276.27) .. (410.66,270.01) -- cycle ;
\draw  [draw opacity=0][fill={rgb, 255:red, 155; green, 155; blue, 155 }  ,fill opacity=1 ] (441.32,269.34) .. controls (441.32,263.08) and (446.4,258) .. (452.66,258) .. controls (458.92,258) and (464,263.08) .. (464,269.34) .. controls (464,275.6) and (458.92,280.68) .. (452.66,280.68) .. controls (446.4,280.68) and (441.32,275.6) .. (441.32,269.34) -- cycle ;
\draw  [draw opacity=0][fill={rgb, 255:red, 249; green, 108; blue, 74 }  ,fill opacity=1 ] (470.66,269.34) .. controls (470.66,263.08) and (475.73,258) .. (481.99,258) .. controls (488.26,258) and (493.33,263.08) .. (493.33,269.34) .. controls (493.33,275.6) and (488.26,280.68) .. (481.99,280.68) .. controls (475.73,280.68) and (470.66,275.6) .. (470.66,269.34) -- cycle ;
\draw  [draw opacity=0][fill={rgb, 255:red, 155; green, 155; blue, 155 }  ,fill opacity=1 ] (501.32,268.67) .. controls (501.32,262.41) and (506.4,257.33) .. (512.66,257.33) .. controls (518.92,257.33) and (524,262.41) .. (524,268.67) .. controls (524,274.93) and (518.92,280.01) .. (512.66,280.01) .. controls (506.4,280.01) and (501.32,274.93) .. (501.32,268.67) -- cycle ;
\draw  [draw opacity=0][fill={rgb, 255:red, 249; green, 108; blue, 74 }  ,fill opacity=1 ] (531.32,268.67) .. controls (531.32,262.41) and (536.4,257.33) .. (542.66,257.33) .. controls (548.92,257.33) and (554,262.41) .. (554,268.67) .. controls (554,274.93) and (548.92,280.01) .. (542.66,280.01) .. controls (536.4,280.01) and (531.32,274.93) .. (531.32,268.67) -- cycle ;
\draw   (93.33,194.67) -- (241.33,194.67) -- (241.33,290.84) -- (93.33,290.84) -- cycle ;
\draw    (154,194.68) .. controls (151.37,172.35) and (150.69,137.73) .. (204.18,134.15) ;
\draw [shift={(206.67,134.01)}, rotate = 177.27] [fill={rgb, 255:red, 0; green, 0; blue, 0 }  ][line width=0.08]  [draw opacity=0] (8.93,-4.29) -- (0,0) -- (8.93,4.29) -- cycle    ;
\draw [color={rgb, 255:red, 245; green, 166; blue, 35 }  ,draw opacity=1 ][line width=0.75]    (245.33,82.01) -- (271.33,82.01) ;
\draw [shift={(274.33,82.01)}, rotate = 180] [fill={rgb, 255:red, 245; green, 166; blue, 35 }  ,fill opacity=1 ][line width=0.08]  [draw opacity=0] (8.93,-4.29) -- (0,0) -- (8.93,4.29) -- cycle    ;
\draw [shift={(242.33,82.01)}, rotate = 360] [fill={rgb, 255:red, 245; green, 166; blue, 35 }  ,fill opacity=1 ][line width=0.08]  [draw opacity=0] (8.93,-4.29) -- (0,0) -- (8.93,4.29) -- cycle    ;
\draw  [draw opacity=0][fill={rgb, 255:red, 74; green, 74; blue, 74 }  ,fill opacity=1 ] (155.65,143.55) -- (149.2,138.5) .. controls (149.87,139.04) and (150.82,139.07) .. (151.31,138.58) .. controls (151.8,138.09) and (151.65,137.27) .. (150.97,136.74) -- (157.43,141.79) .. controls (158.11,142.32) and (158.26,143.14) .. (157.77,143.63) .. controls (157.28,144.12) and (156.33,144.08) .. (155.65,143.55) -- cycle ;
\draw  [draw opacity=0][fill={rgb, 255:red, 74; green, 74; blue, 74 }  ,fill opacity=1 ,even odd rule] (141.47,133.29) .. controls (141.47,130.27) and (143.97,127.82) .. (147.05,127.82) .. controls (150.12,127.82) and (152.62,130.27) .. (152.62,133.29) .. controls (152.62,136.32) and (150.12,138.77) .. (147.05,138.77) .. controls (143.97,138.77) and (141.47,136.32) .. (141.47,133.29)(139,133.29) .. controls (139,128.9) and (142.6,125.34) .. (147.05,125.34) .. controls (151.49,125.34) and (155.09,128.9) .. (155.09,133.29) .. controls (155.09,137.68) and (151.49,141.24) .. (147.05,141.24) .. controls (142.6,141.24) and (139,137.68) .. (139,133.29) ;
\draw  [draw opacity=0][fill={rgb, 255:red, 128; green, 128; blue, 128 }  ,fill opacity=1 ] (146.33,129.5) -- (147.76,129.5) -- (147.76,132.72) -- (150.98,132.72) -- (150.98,134.2) -- (147.76,134.2) -- (147.76,137.42) -- (146.33,137.42) -- (146.33,134.2) -- (143.11,134.2) -- (143.11,132.72) -- (146.33,132.72) -- cycle ;
\draw  [draw opacity=0][fill={rgb, 255:red, 80; green, 227; blue, 194 }  ,fill opacity=1 ] (580.67,207) -- (650.67,207) -- (650.67,217.34) -- (580.67,217.34) -- cycle ;
\draw  [draw opacity=0][fill={rgb, 255:red, 184; green, 233; blue, 134 }  ,fill opacity=1 ] (566.67,264.68) -- (650.67,264.68) -- (650.67,275.34) -- (566.67,275.34) -- cycle ;
\draw   (606.5,97.84) -- (218.5,97.84) -- (218.5,160.84) -- (606.5,160.84) -- cycle ;
\draw   (113.33,303) .. controls (113.32,307.67) and (115.64,310.01) .. (120.31,310.02) -- (203.06,310.31) .. controls (209.73,310.34) and (213.05,312.68) .. (213.03,317.35) .. controls (213.05,312.68) and (216.39,310.36) .. (223.06,310.38)(220.06,310.37) -- (305.81,310.67) .. controls (310.48,310.69) and (312.82,308.37) .. (312.83,303.7) ;
\draw  [color={rgb, 255:red, 74; green, 144; blue, 226 }  ,draw opacity=1 ][fill={rgb, 255:red, 74; green, 144; blue, 226 }  ,fill opacity=0.43 ] (336.5,197.84) -- (576.5,197.84) -- (554.5,286.84) -- (314.5,286.84) -- cycle ;
\draw  [color={rgb, 255:red, 74; green, 144; blue, 226 }  ,draw opacity=1 ] (324.5,302.7) .. controls (324.47,307.37) and (326.79,309.71) .. (331.46,309.74) -- (427.96,310.23) .. controls (434.63,310.26) and (437.95,312.61) .. (437.93,317.28) .. controls (437.95,312.61) and (441.29,310.3) .. (447.96,310.33)(444.96,310.31) -- (544.46,310.82) .. controls (549.13,310.84) and (551.47,308.52) .. (551.5,303.85) ;
\draw  [draw opacity=0][fill={rgb, 255:red, 74; green, 144; blue, 226 }  ,fill opacity=1 ] (459.4,228.44) -- (463.62,232.67) -- (453.45,242.84) -- (463.62,253.02) -- (459.9,256.74) -- (449.72,246.57) -- (439.55,256.74) -- (435.32,252.52) -- (445.5,242.34) -- (435.32,232.17) -- (439.05,228.44) -- (449.22,238.62) -- cycle ;

\draw (224.67,99.33) node [anchor=north west][inner sep=0.75pt]   [align=left] {{\fontfamily{pcr}\selectfont \textbf{\textcolor[rgb]{0.31,0.89,0.76}{.....}\textcolor[rgb]{0.98,0.42,0.29}{(([])\{\})}\textcolor[rgb]{0.29,0.29,0.29}{(([])\{\})}\textcolor[rgb]{0.98,0.42,0.29}{(([])\{\})}\textcolor[rgb]{0.29,0.29,0.29}{(([])\{\})}\textcolor[rgb]{0.31,0.89,0.76}{....}}}};
\draw (224,133.33) node [anchor=north west][inner sep=0.75pt]   [align=left] {{\fontfamily{pcr}\selectfont \textbf{\textcolor[rgb]{0.72,0.91,0.53}{.}\textcolor[rgb]{0.98,0.42,0.29}{(([])\{\})}\textcolor[rgb]{0.29,0.29,0.29}{(([])\{\})}\textcolor[rgb]{0.98,0.42,0.29}{(([])\{\})}\textcolor[rgb]{0.29,0.29,0.29}{(([])\{\})}\textcolor[rgb]{0.98,0.42,0.29}{(([]}\textcolor[rgb]{0.72,0.91,0.53}{....}}}};
\draw (244.33,49) node [anchor=north west][inner sep=0.75pt]  [color={rgb, 255:red, 245; green, 166; blue, 35 }  ,opacity=1 ]  {$2k\text{\mbox{-}synchronized}$};
\draw (202.67,322.67) node [anchor=north west][inner sep=0.75pt]    {$6k$};
\draw (419.3,320) node [anchor=north west][inner sep=0.75pt]  [color={rgb, 255:red, 74; green, 144; blue, 226 }  ,opacity=1 ]  {$e-6k$};

\end{tikzpicture}

%% file: src/vert_per.tex
In order to remove periodicity from regions of $\str{F}, \str{G}$ which may be unbalanced, we consider a second type of periodicity, \emph{vertical} periodicity, in addition to horizontal periodicity.  Avoiding horizontal periodicity allows us to reduce large powers of repeated balanced substrings; in this section, we essentially aim to reduce large powers of pairs of periodic substrings which are balanced together but may be unbalanced separately. We do so by finding paths of nodes in forests $F$ and $G$ such that the children to the left and right of the path of each node in the path are the same, which we call vertical periodicity. At the bottom of the path, we may no longer have vertical periodicity, and so in a parentheses representation of the forest we will have two separate periodic substrings, one to the left of the path and one to the right of the path. We define some useful notation for vertical periodicity as follows:

\newcommand{\cnt}[1]{\langle#1\rangle}
\begin{definition}[Context]
We define a \emph{context} as a pair $C=(C_L,C_R)$ such that $C_L\cdot C_R=\str{T}$ for some labeled forest $T$.
\end{definition}

\begin{definition}[Vertical composition]
For a context $C=(C_L,C_R)$ and a labeled forest $F$, we denote by $C\cnt{F}$ the labeled forest $H$ such that $\str{H}=C_L\cdot \str{F}\cdot C_R$.
Similarly, for two contexts $C=(C_L,C_R)$ and $D=(D_L,D_R)$, we denote $C\cnt{D}=(C_LD_L,D_RC_R)$.
\end{definition}

Observe that the vertical composition of contexts is associative.
For a context $C$ and an integer $e\in \Zp$, we define $C^e$ as the context obtained by vertical composition of $e$ copies of $C$.
We say a forest context $C$ \emph{occurs} at node $u$ of a labeled forest $F$
if the subtree of $F$ rooted at node $u$ is of the form $C\cnt{H}$ for some labeled forest $H$.

We say that context $C$ has $s$-synchronized occurrences in labeled forests $F,G$
if $C$ occurs at a node $u$ of $F$ and at a node $v$ of $G$ such that $|o_F(u)-o_G(v)|\le s$ and $|c_F(u)-c_G(v)|\le s$.

\begin{definition}[Synchronized vertical periodicity]
We say that forests $F,G$ \emph{avoid synchronized vertical $k$-periodicity} if there is no context $C=(C_L,C_R)$ with $|C_L|,|C_R|\le 4k$
such that $C^{16k}$ has $2k$-synchronized occurrences in $F,G$.
\end{definition}

To avoid synchronized vertical $k$-periodicity, we first compute periodic contexts which occur in forests $F$ and $G$ individually without concern for synchronicity. Note that if a forest has such a periodic context $C^{16k} = (C_L, C_R)^{16k}$, then that forest has two separate periodic substrings we want to identify, namely $C_L^{16k}$ and $C_R^{16k}$. Additionally, from the definition of context it is clear that while $C_L$ and $C_R$ do not need to be balanced parentheses strings, $C_L \cdot C_R$ does have to be balanced. Therefore, $C_L$ cannot begin with a closing parenthesis and $C_R$ cannot end with an opening parenthesis since such parentheses would have no match in $C_L \cdot C_R$. For this reason, we will want to find periodic substrings starting at opening parentheses in $\str{F}, \str{G}$ as well as periodic substrings ending at closing parentheses 

For a node $u$ in a forest $F$,
consider $q_L \in [1\dd  4k]$ and $e \in \mathbb{Q}$ such that $q_L$ is a period of $\str{F}[o(u) \dd o(u)+q_L \cdot e)$ and $e$ is maximized. 
In other words, we want to find the longest periodic substring starting at $o(u)$ with period at most $4k$ and exponent at least $16k$. If $e \geq 16k$, we may have vertical periodicity that we want to avoid, and so we define an array $Q_F$ to store these values. Let $Q_F[o(u)] := (q_L, o(u) + q_Le)$. If $e < 16k$, we do not need to worry about reducing any substring starting at $o(u)$ and so, we set a default value $Q_F[o(u)] := (1, o(u))$. Since we want to find periodic substrings ending in closing parentheses as well, we define $Q_F[c(u)] := (q_R, i)$ where $\str{F}(c(u) - i \dd c(u)]$ is the longest periodic substring ending at $c(u)$ with a period $q_R \leq 4k$ and exponent at least $16k$.  Again if the exponent is less than $16k$, we define $Q_F[c(u)] := (1, c(u))$.

\begin{lemma}\label{lem:computeQ}
    Given a forest $F$, $Q_F$ can be computed in time $\Oh(n)$.
\end{lemma}

\begin{proof}
    Initially, we set $Q_F[\ell] := (1, \ell)$ as the default value for all $\ell\in [0\dd 2|F|)$
    and use \cref{alg:FilterRuns} to compute the set $R_F$ of runs $(i,j,p)$ of $\str{F}$ 
    with period $p\le 4k$ and exponent $\frac{j-i}{p}\ge 16k$; as discussed in \cref{sec:hor}, this costs $\Oh(n)$ time.
    For every run $(i, j, p)\in R_F$, take any index $\ell \in [i\dd j)$ such that $\str{F}[\ell]$ is an opening parenthesis. 
    If the exponent of the substring $\str{F}[\ell \dd j)$, that is, $\frac{j-\ell}{p}$, is at least $16k$, we update $Q_F[\ell]$ to $(p, j)$. 
    By iterating over all indices in each run, we will find the longest periodic substring starting at every opening parenthesis in $Q_F$ with exponent at least $16k$. Note that any two runs with period at most $4k$ has less than $8k$ overlapping characters by \cref{cor:boundedperiodoverlap}, and therefore, we only consider each index $\ell$ twice.  Furthermore, we update each value $Q_F[\ell]$ at most once because the indices in the overlap of the runs cannot have an exponent of at least $16k$ in each run (otherwise the overlap would exceed $8k$).
    
    To finish the computation of $Q_F$, we do the analogous steps for closing parentheses.  For every run $(i,j,p)\in R_F$,
    we take any index $\ell\in [i\dd j)$. If the exponent of the substring $\str{F}[i\dd \ell]=\str{F}(i-1\dd \ell]$, that is, $\frac{\ell-(i-1)}{p}$, is at least $16k$,
    we update $Q_F[\ell]$ to $(p, i-1)$.
    Now, note that by Corollary~\ref{cor:boundedperiodoverlap}, runs with period at most $4k$ can only overlap in at most $8k$ indices. Furthermore, since we only update $Q_F$ if the exponent of a run is at least $16k = 2(8k)$, any index $\ell \leq |\str{F}|$ can be contained in at most two runs of period at most $4k$ and exponent at least $16k$. Therefore, iterating through all such runs of $\str{F}$ and computing $Q_F$ for each index takes time $\Oh(n)$.
\end{proof}

 \begin{algorithm}
    $R_F \gets \textsf{FilterRuns}(\str{F})$\;
    $Q_F[\ell] \gets (1, \ell)\quad\forall{\ell \in [|\str{F}|]}$\;
    \For{$(i, j, p) \in R_F$}{
        \For{$\ell \in [i\dd j)$}{
            \If{$\str{\F}[\ell]$ is an opening parenthesis \KwSty{and} $\frac{j - \ell}{p} \geq 16k$}{
                $Q_F[\ell] \gets (p, j)$\;
            }            
            \If{$\str{\F}[\ell]$ is a closing parenthesis \KwSty{and} $\frac{\ell - (i-1)}{p} \geq 16k$}{
                $Q_F[\ell] \gets (p, i-1)$\;
            }
        }
    }    
     \caption{$\mathsf{ComputeQ}(F)$. }
     \label{alg:ComputeQ}
 \end{algorithm}


We give \cref{alg:ComputeQ} for pseudocode on computing $Q_F$ from forest $F$. Now that we have computed $Q_F$, we can iterate over all nodes of forest $F$ to compute the maximal periodic context that occur at each node $u \in F$. Let $Q_F[o(u)] = (q_L, j_L)$ and $Q_F[c(u)] = (q_R, j_R)$, and denote the string periods by $P_L = \str{F}[o(u) \dd o(u)+q_L)$ and $P_R = \str{F}(c(u) - q_R \dd c(u)]$. Assume that there is a periodic context $C^e$ that occurs at $u$ such that $C = (C_L, C_R)$, $e$ is maximized, and the subtree rooted at $u$ is of form $C^e\langle H \rangle$. This means that for two positive integers $r_L, r_R \in \mathbb{Z}_+$, $|C_L| = r_Lq_L$ and $|C_R|=r_Rq_R$ since $C_L$ and $C_R$ could be periodic themselves. To find the correct coefficients $r_L$ and $r_R$, we need to compute the number of unmatched opening parentheses in $P_L$ and the number of unmatched closing parentheses in $P_R$, denoted by $d_L$ and $d_R$ respectively. Values $d_L$ and $d_R$ correspond to the depth that $P_L$ and $P_R$ go down in the tree. The highest node in which the runs starting at $o(u)$ and ending at $c(u)$ synchronize is located at depth $d:=\lcm(d_L,d_R)$. By synchronizing we mean that the unmatched opening and closing parentheses are matched so that $C_L\cdot C_R$ is balanced. Hence, $C_L = P_L^{d/d_L}$ and $C_R = P_R^{d/d_R}$. After this step we need to check whether the conditions $|C_L| \leq 4k$, $|C_R| \leq 4k$, and $e \geq 16k$ still hold. 

The next step is to find the maximum exponent $e$ such that $C^e$ occurs at node $u$. There is no guarantee that $\str{F}[j_L\dd j_R]$ forms a balanced substring. For example, the left run could finish earlier than the right run, or the two runs could \emph{diverge} as illustrated in Figure~\ref{fig:vertical_diverge}. To fix this issue, we first find the vertices $v_L$ and $v_R$ corresponding to indices $j_L$ and $j_R$ in $\str{F}$ respectively. If $v^*$ is the \emph{Lowest Common Ancestor} (LCA) of $v_L$ and $v_R$, then \[e = \left\lfloor{\min\left\{\tfrac{j_L-o(u)}{|C_L|},\tfrac{c(u)-j_R}{|C_R|},\tfrac{c(u)-o(u)+1}{|C_L|+|C_R|}, \tfrac{D[o(v^*)]-D[o(u)]+1}{d} \right\}}\right\rfloor\] where the value $D[o(v)] = D[c(v)]$
is the depth of node $v$ (i.e., the distance to the root of the corresponding tree). The arguments of the $\min$ function indicate that the left run is long enough, the right run is long enough, the subtree of $u$ is large enough (which is needed if the left run overlaps the right run), and that the runs do not diverge too early.
Again, we need to check whether $e \geq 16k$.

\begin{figure}
    \centering
    \input{figs/vertical_diverge}
    \caption{An example tree in which the left run starting at $o(u)$ and the right run ending at $c(u)$ diverge. The depth of the runs are $d_L=2$ and $d_R=1$. Therefore, the depth of the context is $\lcm(d_L, d_R)=2$, and the left and right runs end at $v_L$ and $v_R$ respectively. In this case, the LCA $v^*$ of $v_L$ and $v_R$ is at distance $5$ from $u$, which indicates the exponent $e=\lfloor{\frac{5+1}{2}}\rfloor=3$ of context $C^e$.
    }
    \label{fig:vertical_diverge}
\end{figure}

\begin{definition}
    Given a forest $F$, let $\mathcal{C}(F)$ be a set of quadruples $(u, q_L, q_R, e)$ such that 
    \begin{enumerate}
        \item $u$ is a node in $F$
        \item $q_L, q_R \leq 4k$, $e \geq 16k$.
        \item $C_L = \str{F}[o(u) \dd o(u) + q_L), C_R = \str{F}(c(u) - q_R \dd c(u)]$ form context $C = (C_L, C_R)$ in the subtree rooted at $u$.
        \item $C^e$ is a maximal context, i.e., $C^{e+1}$ does not occur at $u$. 
    \end{enumerate}
\end{definition}

\begin{lemma}
    $\mathcal{C}(F)$ can be computed in $\Oh(n\log n)$ time.
    \label{lem:computeC}
\end{lemma}

 \begin{algorithm}
    $Q_F \gets \textsf{ComputeQ}(\str{F})$\;
    $C \gets \textrm{empty set}$\;
    \For{$u \in V_F$}{ 
        $(q_L,j_L) \gets Q_F[o(u)], (q_R,j_R) \gets Q_F[c(u)]$\;
        \tcp{Check for dummy $Q_F$ values and make sure $o(u)$ and $c(u)$ are not in the same period}
        \If{$j_L \ne o(u)$ \KwSty{and} $j_R \ne c(u)$ \KwSty{and} $(c(u) - o(u)) \geq \max\{q_L, q_R\}$}{
        $d_L \gets D[o(u)+q_L]-D[o(u)]$\;
        $d_R \gets D[c(u)-q_R]-D[c(u)]$\;
        $d \gets \lcm(d_L,d_R)$\;
        $C_L \gets \str{F}[o(u) \dd o(u)+q_L)^{d/d_L}$\;
        $C_R \gets \str{F}(c(u)-q_R \dd c(u)]^{d/d_R}$\;
        \If{$|C_L| > 4k$ \KwSty{or} $|C_R| > 4k$}{
            \KwSty{continue}\;
        }
        \tcp{Let $v_L, v_R$ be the nodes corresponding to $j_L, j_R$}
        $v^* \gets \textsf{LCA}(v_L, v_R)$\;
        $e \gets \lfloor{\min\{\frac{j_L-o(u)}{|C_L|},\frac{c(u)-j_R}{|C_R|},\frac{c(u)-o(u)+1}{|C_L|+|C_R|}, \frac{D[o(v^*)]-D[o(u)]+1}{d} \}}\rfloor$\;
        \If{$e \geq 16k$}{
            $C.\textsf{insert}\big((u, |C_L|, |C_R|, e)\big)$\;
        }
        }
    }    
    $\textsf{return}\; C$\;
     \caption{$\mathsf{ComputeContexts}(F, D)$. }
     \label{alg:ComputeContext}
 \end{algorithm}

\begin{proof}
Given forest $F$, we showed that we can find $Q_F$ in $\Oh(n)$ time in Lemma~\ref{lem:computeQ}.
In the rest of the algorithm, we iterate over the nodes of $F$. For each node $u \in V_F$, we check $Q_F[o(u)]$ and $Q_F[c(u)]$ to find a potential periodic context rooted at vertex $u$. If the desired context exists with period at most $4k$ and exponent at least $16k$, it satisfies the required conditions to enter $\mathcal{C}(F)$. 

In order to find the period of a maximal context, we first define array $D[0\dd 2|F|)$ so that, for each node $u\in V_F$, the value $D[o(u)]=D[c(u)]$
is the depth of node $u$ (i.e., the distance to the root of the corresponding tree). Recall that $d_L$ is the number of unmatched opening parentheses in $P_L$, the shortest period starting from $o(u)$ which is stored in array $Q_F$, and similarly $d_R$ is the number of the number of unmatched closing parenthesis in $P_R$. The shortest period of a maximal context is equal to $\lcm(d_L, d_R)$, which can be found in $\Oh(\log k)$ time. 
\begin{align}
    d_L &= D[o(u)+q_L] - D[o(u)],\label{eq:dL}\\
    d_R &= D[c(u)-q_R] - D[c(u)].\label{eq:dR}
\end{align}

To find the correct exponent, we query the LCA of two nodes in the tree which indicate the end of $Q_F[o(u)]$ and $Q_F[c(u)]$;
this operation takes $\Oh(1)$ time after $\Oh(n)$-time preprocessing of the forest $F$~\cite{10.1137/0213024,Bender2005}. Once we have the LCA $v^*$, 
we compute the exponent in $\Oh(1)$ time using the $D$ array.
\end{proof}

Once we compute the set of contexts of each node in forests $F$ and $G$, we next find which contexts have $2k$-synchronized occurrences that we should reduce. We can use 2-D range queries to find such synchronized occurrences since we need to check both if the opening parentheses indices and closing parenthesis indices are within $2k$ of each other.

\begin{definition}
    Given a set of points $S$ and a query rectangle $Q$ in the plane, \emph{orthogonal range successor} ($\ORS$) is the problem of finding the point with smallest $y$-coordinate in $S \cap Q$. Given a quadruple $(u, q_L, q_R, e) \in \mathcal{C}(F)$ with context $C = (\str{F}[o(u) \dd o(u) + q_L), \str{F}(c(u) - q_R \dd c(u)])$ we let $\ORS_\mathcal{C}(u)$ denote an ORS query which returns a node $v \in G$ such that $(v, q_L, q_R, e') \in \mathcal{C}(G)$, $(\str{G}[o(v) \dd o(v) + q_L), \str{G}(c(v) - q_R \dd c(v)]) = C$ and $|o(u) - o(v)| \leq 2k, |c(u) - c(v)| \leq 2k$. If no such query exists, we return $(0, 0, 0, 0)$.
\end{definition}

\begin{theorem}[Linear time $\ORS$ queries~\cite{DBLP:conf/esa/Gao0N20}]
    Given a set of $n$ points in the plane, a data structure can be computed in time $\Oh(n \sqrt{\log n})$ to answer $\ORS$ queries in $\Oh(\lg \lg n)$ time.
    \label{thm:linearors}
\end{theorem}

\begin{definition}
    Given a string $s$, a \emph{fingerprint hash} of $s$ is a function $h_s : s \rightarrow [|s|^3]$ such that $h_s(s[i_1 \dd j_1]) = h_s(s[i_2 \dd j_2])$ if and only if $s[i_1 \dd j_1 ] = s[i_2 \dd j_2]$.
\end{definition}

\begin{theorem}[Linear time fingerprint hash~\cite{DBLP:conf/esa/Gawrychowski11} ]
    Given a string $s$ with $n = |s|$, after $\Oh(n)$ preprocessing time, fingerprint hash values can be computed in constant time for substrings of $s$.
    \label{thm:linearfingerprint}
\end{theorem}


\begin{lemma}
    Given forests $F, G$ and quadruple $(u, q_L, q_R, e) \in \mathcal{C}(F)$ with context \\ $C = (\str{F}[o(u) \dd o(u) + qL), \str{F}(c(u) - q_R \dd c(u)])$, $\ORS_\mathcal{C}(u) \in G$ can be computed in $\Oh(\lg \lg n)$ time with $\Oh(n \sqrt{\lg n})$ preprocessing.
    \label{lem:ORS}
\end{lemma}

\begin{proof}
    We will build a separate $\ORS$ data structure $\mathcal{D}_C$ for each context $C = (C_L, C_R)$ that corresponds to a quadruple of $\mathcal{C}(F)$. Any node $v \in G$ will be contained in data structure $\mathcal{D}_C$ if $(v, q_L, q_R, e) \in \mathcal{C}(G)$ corresponds has context $C$. To determine which point belongs to which data structure, we first preprocess $\str{F}, \str{G}$ and compute a fingerprint hash $h_{F, G}$ in $\Oh(n)$ time by \cref{thm:linearfingerprint}. Then, given a quadruple $(v, q_L, q_R, e) \in \mathcal{C}(G)$, determining which data structure $v$ belongs to can be done in constant time using $h_{F, G}$. Since there is only one quadruple per node of $F$ and $G$ in $\mathcal{C}(F), \mathcal{C}(G)$ and $n$ nodes total per forest, we can construct all data structures of contexts in $\mathcal{C}(F)$ in time $\Oh(n \sqrt{\lg n})$ according to \cref{thm:linearors}. Input queries to these data structures will be the $2k \times 2k$ rectangle surrounding a point $(o(u), c(u))$ such that $(u, q_L, q_R, e) \in \mathcal{C}(F)$ and outputs will be quadruples $(v, q_L, q_R, e') \in \mathcal{C}(G)$ where $(o(v), c(v))$ is the orthogonal range successor of $(o(u), c(u))$ in $\mathcal{D}_C$. Again, by \cref{thm:linearors}, these queries can be answered in time $\Oh(\lg \lg n)$.
\end{proof}

\begin{algorithm}
    $i \leftarrow -1$\;
    $S \leftarrow \emptyset$\;
    $\mathcal{C}(F) \leftarrow \mathsf{SortByStartingIndex}(\mathcal{C}(F))$\;
    \For{$(u_F, q_L^F, q_R^F, e_F) \in \mathcal{C}(F)$}{
        \If{$o(u_F) > i$}
        {
            $C \leftarrow (\str{F}[o(u_F) \dd o(u_F) + q_L^F), \str{F}(c(u_F) - q_R^F \dd c(u_F)])$\;
            $(v, q_L^F, q_R^G, e_G) \leftarrow \ORS_{C}(u)$\;\label{ln:6.7}
            \If{$e_G \ne 0$}{            
                $e' \leftarrow \min(e_F, e_G)$\;\label{ln:6.9}
                $S \leftarrow S \cup (u_F, u_G, q_L^F, q_R^F, e')$\;\label{ln:6.10}
                $i \leftarrow o(u) + (e' - 8k)q_L^F$\;\label{ln:6.11}
            }

        }
    }
    \KwRet{$S$}\;
    \caption{$\mathsf{VertPeriods}(F, G)$}
    \label{alg:VertSyncOccurrences}
\end{algorithm}


\cref{alg:VertSyncOccurrences} and \cref{alg:VertSyncReductions} identify and output forests $F', G'$ with reduced vertical periodicity. \Cref{alg:VertSyncOccurrences} simply iterates through all nodes $u$ of forest $F$ that are the beginning of a vertical period with power at least $16k$. Then using $\ORS$ queries, we find any nodes $v$ in $G$ with a matching context to that of $u$ and add the $2k$-synchronized occurrence to set $S$ to be reduced in \cref{alg:VertSyncReductions}. \Cref{alg:VertSyncReductions} is fairly straightforward and simply outputs $\str{F'}$ and $\str{G'}$ by copying all characters of $\str{F}, \str{G}$ and skipping all but $14k$ repetitions of any $2k$-synchronized vertical periods of set $S$. We now show that these algorithms do correctly identify vertical periods and that reducing vertical period powers does not change tree edit distance. Most of the following proofs mimic the same structures as \cref{sec:hor} with more details since each context $C^e = (C_L, C_R)^e$ has a left $C_L^e$ and right $C_R^e$ part we must consider rather than a single substring $Q^e$.

\begin{lemma}
    Given forests $F, G$ and $S = \mathsf{VertPeriods}(F, G)$, then for any $(u, v, q_L, q_R, e) \in S$ there is $2k$-synchronized occurrences of $C^e$ at nodes $u$ in $F$, $v$ in $G$ where $C = (\str{F}[o(u) \dd o(u) + q_L), \str{F}(c(u) - q_R \dd c(u)])$ and $e \geq 16k$.
    \label{lem:vertperset}
\end{lemma}

\begin{proof}
    The proof of this lemma is fairly straightforward from the steps of \cref{alg:VertSyncOccurrences}. First, note that for any quintuple $(u, v, q_L, q_R, e')$ added to $S$ in step~\ref{ln:6.10}, there must be some quadruple $(u, q_L, q_R, e_F) \in \mathcal{C}(F)$ where by definition of $\mathcal{C}(F)$, $e_F \geq 16k$ and $C^{e_F}$ occurs at node $u$ where context $C = (\str{F}[o(u) \dd o(u) + q_L), \str{F}(c(u) - q_R \dd c(u)])$. Furthermore, by definition of $\mathsf{ORS}_C(u)$ any quadruple $(v, q_L^G, q_R^G, e_G)$ returned in step~\ref{ln:6.7} must also have an occurrence of $C^{e_G}$ at node $v$ where $|o(u) - o(v)| \leq 2k, |c(u) - c(v)| \leq 2k$. Moreover, $(v, q_L, q_R, e_G)$ must be in $\mathcal{C}(G)$, and so, $e_G \geq 16k$. So, since in step~\ref{ln:6.9} we set $e' = \min(e_F, e_G)$, we have that $e' \geq 16k$ and there is a $2k$-synchronized occurrence of $C^e$ at nodes $u$ in $F$ and $v$ in $G$.
\end{proof}

\begin{algorithm}
    $S \leftarrow \mathsf{VertPeriods}(F, G)$\;\label{ln:7.1}
    $S' \leftarrow \emptyset$\;
    \For{$(u_F, u_G, q_L, q_R, e) \in S$}{
        $S' \leftarrow S' \cup \{(o(u_F), o(u_G), q_L, e),(c(u_F) - q_R\cdot e+1, c(u_G) - q_R\cdot e+1, q_R, e)\}\label{ln:7.4} $ \;
    }
    $S' \leftarrow \textsf{SortByStartingIndex}(S')$\;\label{ln:7.5}
    $s_F, s_G \leftarrow \varepsilon$\;
    $i_F, i_G \leftarrow 0$\;
    \For{$(\ell_F, \ell_G, q, e) \in S'$}{
        $s_F \leftarrow s_F \cdot \str{F}[i_F \dd \ell_F)$\;
        $s_G \leftarrow s_G \cdot \str{G}[i_G \dd \ell_G)$\;
        $i_F \leftarrow \ell_F + q (e - 14k)$\;\label{ln:7.11}
        $i_G \leftarrow \ell_G + q (e - 14k)$\;\label{ln:7.12}
    }
    $s_F \leftarrow s_F \cdot \str{F}[i_F \dd |\str{F}|)$ \;
    $s_G \leftarrow s_G \cdot \str{G}[i_G \dd |\str{G}|)$ \; 
    \KwRet{$s_F, s_G$}\;
    \caption{$\mathsf{VertSyncReductions}(F, G)$}
    \label{alg:VertSyncReductions}
\end{algorithm}

\begin{lemma}
    Consider forests $F, G$ and a context $C = (C_L, C_R)$ such that $0<|C_L|, |C_R| \leq 4k$ and, for some integer exponent $e\ge 10k$,
    the context $C^e$ has $2k$-synchronized occurrences in $F, G$ at nodes $u$ and $v$, respectively. This means that
    \begin{align*}
    \str{F} = \str{F}[0 \dd o_F(u)) \cdot C_L^e \cdot \str{F}[o_F(u)+e|C_L| \dd c_F(u)-e|C_R|] \cdot C_R^e \cdot \str{F}(c_F(u) \dd |\str{F}|), \\ 
    \str{G} = \str{G}[0 \dd o_G(v)) \cdot C_L^e \cdot \str{G}[o_G(v)+e|C_L| \dd c_G(v)-e|C_R|] \cdot C_R^e \cdot \str{G}(c_G(v) \dd |\str{G}|).
    \end{align*}
    We define $F', G'$ so that the following holds for some exponent $e' \ge 10k$:
    \begin{align*}
    \str{F'} = \str{F}[0 \dd o_F(u)) \cdot C_L^{e'} \cdot \str{F}[o_F(u)+e|C_L| \dd c_F(u)-e|C_R|] \cdot C_R^{e'} \cdot \str{F}(c_F(u) \dd |\str{F}|), \\ \str{G'} = \str{G}[0 \dd o_G(v)) \cdot C_L^{e'} \cdot \str{G}[o_G(v)+e|C_L| \dd c_G(v)-e|C_R|] \cdot C_R^{e'} \cdot \str{G}(c_G(v) \dd |\str{G}|). 
    \end{align*}
    Then, $\ted_{\le k}(F, G) = \ted_{\le k}(F', G')$.
    \label{lem:vertred}
\end{lemma}

\begin{proof}
    We assume without loss of generality that $C$ is primitive (if $C$ can be expressed as an integer power of a smaller context, we should consider that context instead) and denote $q_L := |C_L| \leq 4k$ and $q_R := |C_R| \leq 4k$. For $i\in [0\dd e)$, let $u_i$ be the node of $F$
    with $o_F(u_i)=o_F(u)+iq_L$ (and $c_F(u_i)=c_F(u)-iq_R$) and let $v_i$ be the node of $G$ with $o_G(v_i)=o_G(v)+iq_L$ (and $c_G(v_i)=c_G(v)-iq_R$).
    Moreover, let $\A$ be an optimal tree alignment such that $\ted(F, G) = \ted_{\A} (F, G) \le k$.
   
    \begin{claim}
        There exist $i_F,i_G\in [0\dd 5k]$ such that \[(o_F(u)+i_F\cdot q_L, o_G(v)+i_G\cdot q_L),(1+c_F(u)-i_F\cdot q_R, 1+c_G(v)-i_G\cdot q_R) \in \A.\]
        Moreover, there exist $j_F,j_G\in [0\dd 5k]$ such that \[(o_F(u)+(e-j_F)q_L, o_G(v)+(e-j_G)q_L),(1+c_F(u)-(e-j_F)q_R, 1+c_G(v)-(e-j_G)q_R)\in \A.\]
    \end{claim}
    \begin{proof}
        Let $(x_F,x_G)\in \A$ be the leftmost element of $\A$ such that $x_F \ge o_F(u)$ and $x_G \ge o_G(v)$.
        By symmetry between $F$ and $G$, we may assume without loss of generality that $x_F = o_F(u)$.
        The context $C$ occurs at each of the nodes $u_0,\ldots,u_k$ and, since the occurrences are disjoint, the alignment $\A$ (of cost at most $k$) must match one of these occurrences perfectly.
        We pick the index $i_F \in [0\dd k]$ of one such perfectly matched occurrence and denote the node matched to $u_{i_F}$ by $w$.
        In particular, $\str{F}[o_F(u_{i_F})\dd o_F(u_{i_F})+q_L)\simeq_{\A} \str{G}[o_G(w)\dd o_G(w)+q_L)$
        and $\str{F}(c_F(u_{i_F})-q_R\dd c_F(u_{i_F})]\simeq_{\A} \str{G}(c_G(w)-q_R\dd c_G(w)]$.
        Since $(x_F,x_G)\in \A$, we must have $o_G(w)\ge o_G(v)$ by the non-crossing property of $\A$.
        At the same time, $o_G(w)\le o_F(u_{i_F}) + 2k \le o_F(u)+kq_L + 2k \le o_G(v)+kq_L + 4k \le o_G(v)+5kq_L$.
        Similarly, $c_G(w) \ge c_G(v)-5kq_R$, which also implies $c_G(w) \le c_G(v)$.

        Our next goal is to show that $w=v_{i_G}$ for some $i_G\in [0\dd 5k]$.
        For a proof by contradiction, suppose that $o_G(v_i)<o_G(w)<o_G(v_{i+1})$ for some $i\in [0\dd 5k)$.
        Due to $c_G(w)> c_G(v)-5kq_R$, this also implies that $c_G(v_i)>c_G(w)>c_G(v_{i+1})$,
        i.e., that $w$ is a node on the path between $v_i$ and $v_{i+1}$.
        Suppose that the length of this path is $\ell$ and the node $w$ is at distance $\ell'$ from $v_i$.
        Hence,  $\str{G}[o_G(v_i)\dd o_G(w))$ has $\ell'$ unmatched opening parentheses out of the $\ell$  unmatched opening parentheses in $C_L$. 
        Moreover, $\str{G}[o_G(v_i)\dd o_G(w))\cdot \str{G}[o_G(w)\dd o_G(v_{i+1}))=C_L = \str{G}[o_G(w)\dd o_G(v_{i+1}))\cdot \str{G}[o_G(v_i)\dd o_G(w))$,
        and thus there is a primitive string $Q_L$ such that $\str{G}[o_G(w)\dd o_G(v_{i+1}))$ and $\str{G}[o_G(v_i)\dd o_G(w))$ are both powers of $Q_L$.
        The number of unmatched opening parentheses is $Q_L$ must be a common divisor of $\ell$ and $\ell'$,
        i.e., $C_L$ can be expressed as a string power with exponent $\ell/\gcd(\ell,\ell')$.
        A symmetric argument shows that $C_R$ can be expressed as a string power with exponent $\ell/\gcd(\ell,\ell')$.
        Overall, we conclude that $C$ can be expressed as a context power with exponent $\ell/\gcd(\ell,\ell')$,
        contradicting the primitivity of $C$.
        Hence, $w=v_{i_G}$ for some $i_G\in [0\dd 5k]$ holds as claimed and, in particular,
        $o_G(w) = o_G(v)+i_G q_L$ and $c_G(w)=c_G(v)-i_Gq_R$.

        The proof of the second part of the claim is analogous.
    \end{proof}

    Observe that $\str{F}[o_F(u)+i_Fq_L\dd o_F(u)+(e-j_F)q_L)=C_L^{d_F}$ and $\str{F}(c_F(u)-(e-j_F)q_R\dd c_F(u)-i_F q_L]=C_R^{d_F}$
    for $d_F := e-j_F-i_F$. Symmetrically, 
    $\str{G}[o_G(v)+i_Gq_L\dd o_G(v)+(e-j_G)q_L)=C_L^{d_G}$ and $\str{G}(c_G(v)-(e-j_G)q_R\dd c_G(v)-i_G q_L]=C_R^{d_G}$ for $d_G := e-j_G-i_G$.
    We denote $d=  \min(d_F,d_G)$, observe that $d \ge e-10k \ge 0$, and construct a tree alignment $\A'$ so that it:
    \begin{itemize}
        \item aligns $\str{F}[0\dd o_F(u)+i_Fq_L)$ with $\str{G}[0\dd o_G(v)+i_Gq_L)$ in the same way as $\A$ does;
        \item matches $\str{F}[o_F(u)+i_Fq_L\dd o_F(u)+(i_F+d)q_L)=C_L^d$ with $\str{G}[o_G(v)+i_Gq_L\dd o_G(v)+(i_G+d)q_L)=C_L^d$;
        \item deletes $\str{F}[o_F(u)+(i_F+d)q_L\dd o_F(u)+(e-j_F)q_L)=C_L^{d_F-d}$ and $\str{G}[o_G(v)+(i_G+d)q_L\dd o_G(v)+(e-j_G)q_L)=C_L^{d_G-d}$;
        \item aligns $\str{F}[o_F(u)+(e-j_F)q_L\dd c_F(u)-(e-j_F)q_R]$ with $\str{G}[o_G(v)+(e-j_G)q_L\dd c_G(v)-(e-j_G)q_R]$ in the same way as $\A$ does;
        \item deletes $\str{F}(c_F(u)-(e-j_F)q_R\dd c_F(u)-(i_F+d)q_R]=C_R^{d_F-d}$ and $\str{G}(c_G(v)-(e-j_G)q_R\dd c_G(v)-(i_G+d)q_R]=C_R^{d_G-d}$;
        \item matches $\str{F}(c_F(u)-(i_F+d)q_R\dd c_F(u)-i_Fq_R]=C_R^{d}$ with $\str{G}(c_G(v)-(i_G+d)q_R\dd c_G(v)-i_Gq_R]=C_R^{d}$;
        \item aligns $\str{F}(c_F(u)-i_Fq_R\dd |\str{F}|)$ with $\str{G}(c_G(v)-i_Gq_R\dd |\str{G}|)$ in the same way as $\A$ does.
    \end{itemize}
    This definition makes it clear that $\A'$ is an edit-distance (string) alignment.
    In terms of the forests $F$ and $G$, the alignment $\A'$ can be interpreted so that it:
    \begin{itemize}
        \item perfectly matches the occurrences of $C$ at nodes $u_{i_F},\ldots, u_{i_F+d-1}$ to the occurrences of $C$ at nodes $v_{i_G},\ldots v_{i_G+d-1}$, respectively;
        \item deletes the occurrences of $C$ at nodes $u_{i_F+d},\ldots u_{i_F+d_F-1}$ and $v_{i_G+d},\ldots v_{i_G+d_G-1}$;
        \item handles the remaining parts of $F$ and $G$ in the same way as $\A$ does.
    \end{itemize}
    This interpretation makes it clear that $\A'$ is a tree alignment.
    Moreover, $\A$ and $\A'$ only differ in how they align $C^{d_F}$ to $C^{d_G}$, and $\A'$ provides an optimum alignments of these contexts.

    Now, if the exponent $e$ of the context $C^e$ is modified to $e'\ge 10k$, we can interpret this as modifying the exponent $d$ of the contexts $C^d$ matched perfectly by $\A'$ to $d' = d+e'-e\ge 0$. Thus, $\A'$ can be trivially adapted without modifying its cost, and hence
    $\ted(F',G')\le \ted_{\A'}(F,G)= \ted(F,G)$. The converse inequality follows by symmetry between $(F,G)$ and $(F',G')$.
\end{proof}

\begin{proposition}[Avoiding Vertical $k$-Periodicity]\label{prp:ver}
There exists an $\Oh(n\log n)$-time algorithm that, given labeled forests $F,G$ of total size $n$ and an integer $k\in \Zp$, produces forests labeled forests $F',G'$
that satisfy $\ted_{\le k}(F,G)=\ted_{\le k}(F',G')$ and, moreover, avoid both synchronized horizontal $k$-periodicity and synchronized vertical $k$-periodicity.
\end{proposition}

\begin{proof}
    We consider \cref{alg:VertSyncReductions} to prove this proposition. First we find $S = \mathsf{VertPeriods(F, G)}$ in step~\ref{ln:7.1}; by Lemma \ref{lem:vertperset}, any $(u, v, q_L, q_R, e) \in S$ represents a $2k$-synchronized occurrence of $C^e$ where $C_L = \str{F}[o(u) \dd o(u) + q_L), C_R = \str{G}(c(u) - q_R \dd c(u)], C = (C_L, C_R)$ at nodes $u$ and $v$. In step~\ref{ln:7.4}, we identify the start of the specific periodic substrings we want to reduce the exponent of, and construct a set $S'$ with quadruples containing these starting indices, the length of the periods, and the exponent to be reduced. In step~\ref{ln:7.5} we sort these quadruples by their indices and iterate through them in order from left to right. For each quadruple $(\ell_F, \ell_G, q, e)$, in step~\ref{ln:7.11} and step~\ref{ln:7.12}, we move the start of these synchronized occurrences ahead by $q(e - 14k)$. Clearly, this is equivalent to shortening the power $e$ of the synchronized occurrences of $C$ to a power of $14k$ exactly since $$\frac{(\ell_F + qe) - (\ell_F + q(e - 14k))}{q} = 14k.$$  We also note that by construction in \cref{alg:VertSyncOccurrences}, no two quintuples of $S$ will represent overlapping synchronized occurrences with more than $8k \leq 14k$ overlapping characters in either the left or right part of the context (see step~\ref{ln:6.11}).  Therefore, 
    since $C_L \cdot C_R$ is a balanced string, we can construct forests $F', G'$ such that $\str{F'} = s_F, \str{G'} = s_G$ and by Lemma \ref{lem:vertred}, $\ted_{\le k}(F, G) = \ted_{\le k}(F', G')$.
    
    Now, we want to show that $F', G'$ avoid vertical $k$-periodicity. Assume for contradiction that there is some $2k$-synchronized occurrence $C^{16k}$ in $F', G'$. If the occurrence already was present in $F', G'$, there must have been some maximal $2k$-synchronized occurrence $C^e$ for $e \geq 16k$ between nodes $u, w \in F$ and $v, x \in G$. \Cref{alg:VertSyncOccurrences} iterates over all nodes of $F$ with a vertical period and finds any synchronized occurrences of the period in $G$.  Note that when a synchronized occurrence with power $e'$ and context length $q_L^F$ is found, we skip the $(e' - 8k)q_L^F$ nodes, i.e. we still check the last $8k$ nodes of every vertical period. By \cref{cor:boundedperiodoverlap}, there is at most an $8k$ overlap between any two maximal runs of period length at most $4k$, we know that we will check node $u$ since it is the start of a maximal vertical period. Therefore, as shown previously, we will find the occurrence of $C^e$ and reduce its power to $14k$ and so no such $C^{16k}$ will occur in $F', G'$.
    
    We have shown that if a $2k$-synchronized $C^{16k}$ occurs in $F, G$, we will reduce it.  However, it is possible that such an occurrence comes as a result of other reductions. In this case, by \cref{cor:boundedperiodoverlap}, such a synchronized occurrence in $F', G'$ can overlap at most two reduced vertical periods $C_1^{14k}, C_2^{14k}$ since both contexts had powers of at least $16k$ in $F, G$. Since the overlap with each context is at most $8k$ and we leave $14k$ repetitions of $C_1$ and $C_2$ in $F', G'$, any context $C^{16k}$ in $F', G'$ must have also occurred in $F, G$. As shown previously, $C^{16k}$ must have been reduced to $C^{14k}$, which is a contradiction. Therefore, all $2k$-synchronized occurrences of vertical periods have been reduced in $F', G'$ to a power of $14k$. Note that by the same argument, no new horizontal periodicity can occur in $F', G'$ as well.
    
   \cref{alg:VertSyncOccurrences} takes $\Oh(n\log n)$ time to compute $S$ by Lemmas \ref{lem:computeC} and \ref{lem:ORS}. We add one quintuple to $S$ per node in $F$, so $|S| \leq n$. Therefore, \cref{alg:VertSyncReductions} takes time $\Oh(n\log n)$ to compute forests $F', G'$.
\end{proof}

%% file: figs/vertical_diverge.tex
\tikzset{every picture/.style={line width=0.75pt}} 

\begin{tikzpicture}[x=0.75pt,y=0.75pt,yscale=-1,xscale=1]

\draw   (270,70.74) .. controls (270,65.09) and (274.59,60.5) .. (280.25,60.5) .. controls (285.91,60.5) and (290.5,65.09) .. (290.5,70.74) .. controls (290.5,76.4) and (285.91,80.98) .. (280.25,80.98) .. controls (274.59,80.98) and (270,76.4) .. (270,70.74) -- cycle ;
\draw   (270,110.24) .. controls (270,104.59) and (274.59,100) .. (280.25,100) .. controls (285.91,100) and (290.5,104.59) .. (290.5,110.24) .. controls (290.5,115.9) and (285.91,120.48) .. (280.25,120.48) .. controls (274.59,120.48) and (270,115.9) .. (270,110.24) -- cycle ;
\draw   (269.5,150.24) .. controls (269.5,144.59) and (274.09,140) .. (279.75,140) .. controls (285.41,140) and (290,144.59) .. (290,150.24) .. controls (290,155.9) and (285.41,160.48) .. (279.75,160.48) .. controls (274.09,160.48) and (269.5,155.9) .. (269.5,150.24) -- cycle ;
\draw   (270,191.24) .. controls (270,185.59) and (274.59,181) .. (280.25,181) .. controls (285.91,181) and (290.5,185.59) .. (290.5,191.24) .. controls (290.5,196.9) and (285.91,201.48) .. (280.25,201.48) .. controls (274.59,201.48) and (270,196.9) .. (270,191.24) -- cycle ;
\draw   (270,230.74) .. controls (270,225.09) and (274.59,220.5) .. (280.25,220.5) .. controls (285.91,220.5) and (290.5,225.09) .. (290.5,230.74) .. controls (290.5,236.4) and (285.91,240.98) .. (280.25,240.98) .. controls (274.59,240.98) and (270,236.4) .. (270,230.74) -- cycle ;
\draw   (269.5,270.74) .. controls (269.5,265.09) and (274.09,260.5) .. (279.75,260.5) .. controls (285.41,260.5) and (290,265.09) .. (290,270.74) .. controls (290,276.4) and (285.41,280.98) .. (279.75,280.98) .. controls (274.09,280.98) and (269.5,276.4) .. (269.5,270.74) -- cycle ;
\draw [color={rgb, 255:red, 245; green, 166; blue, 35 }  ,draw opacity=1 ][line width=2.25]    (283.25,80.48) -- (283.25,99.5) ;
\draw [color={rgb, 255:red, 235; green, 51; blue, 51 }  ,draw opacity=0.97 ][line width=2.25]    (283.75,120.48) -- (283.25,140) ;
\draw [color={rgb, 255:red, 245; green, 166; blue, 35 }  ,draw opacity=1 ][line width=2.25]    (283.25,160.48) -- (283.75,181) ;
\draw [color={rgb, 255:red, 235; green, 51; blue, 51 }  ,draw opacity=1 ][line width=2.25]    (283.75,201.48) -- (283.75,220.5) ;
\draw [color={rgb, 255:red, 245; green, 166; blue, 35 }  ,draw opacity=1 ][line width=2.25]    (283.25,240.98) -- (282.75,260.5) ;
\draw   (240,90.24) .. controls (240,84.59) and (244.59,80) .. (250.25,80) .. controls (255.91,80) and (260.5,84.59) .. (260.5,90.24) .. controls (260.5,95.9) and (255.91,100.48) .. (250.25,100.48) .. controls (244.59,100.48) and (240,95.9) .. (240,90.24) -- cycle ;
\draw   (210.5,110.24) .. controls (210.5,104.59) and (215.09,100) .. (220.75,100) .. controls (226.41,100) and (231,104.59) .. (231,110.24) .. controls (231,115.9) and (226.41,120.48) .. (220.75,120.48) .. controls (215.09,120.48) and (210.5,115.9) .. (210.5,110.24) -- cycle ;
\draw [color={rgb, 255:red, 74; green, 144; blue, 226 }  ,draw opacity=1 ][line width=2.25]    (228,103.57) -- (241.5,95.57) ;
\draw [color={rgb, 255:red, 74; green, 144; blue, 226 }  ,draw opacity=1 ][line width=2.25]    (258,83.57) -- (271,75.57) ;
\draw   (240,170.74) .. controls (240,165.09) and (244.59,160.5) .. (250.25,160.5) .. controls (255.91,160.5) and (260.5,165.09) .. (260.5,170.74) .. controls (260.5,176.4) and (255.91,180.98) .. (250.25,180.98) .. controls (244.59,180.98) and (240,176.4) .. (240,170.74) -- cycle ;
\draw   (210.5,190.74) .. controls (210.5,185.09) and (215.09,180.5) .. (220.75,180.5) .. controls (226.41,180.5) and (231,185.09) .. (231,190.74) .. controls (231,196.4) and (226.41,200.98) .. (220.75,200.98) .. controls (215.09,200.98) and (210.5,196.4) .. (210.5,190.74) -- cycle ;
\draw [color={rgb, 255:red, 184; green, 233; blue, 134 }  ,draw opacity=1 ][line width=2.25]    (228,184.07) -- (241.5,176.07) ;
\draw [color={rgb, 255:red, 184; green, 233; blue, 134 }  ,draw opacity=1 ][line width=2.25]    (258,164.07) -- (271,156.07) ;
\draw   (239.5,248.74) .. controls (239.5,243.09) and (244.09,238.5) .. (249.75,238.5) .. controls (255.41,238.5) and (260,243.09) .. (260,248.74) .. controls (260,254.4) and (255.41,258.98) .. (249.75,258.98) .. controls (244.09,258.98) and (239.5,254.4) .. (239.5,248.74) -- cycle ;
\draw   (210,268.74) .. controls (210,263.09) and (214.59,258.5) .. (220.25,258.5) .. controls (225.91,258.5) and (230.5,263.09) .. (230.5,268.74) .. controls (230.5,274.4) and (225.91,278.98) .. (220.25,278.98) .. controls (214.59,278.98) and (210,274.4) .. (210,268.74) -- cycle ;
\draw [color={rgb, 255:red, 74; green, 144; blue, 226 }  ,draw opacity=1 ][line width=2.25]    (227.5,262.07) -- (241,254.07) ;
\draw [color={rgb, 255:red, 74; green, 144; blue, 226 }  ,draw opacity=1 ][line width=2.25]    (257.5,242.07) -- (270.5,234.07) ;
\draw   (244,309.74) .. controls (244,304.09) and (248.59,299.5) .. (254.25,299.5) .. controls (259.91,299.5) and (264.5,304.09) .. (264.5,309.74) .. controls (264.5,315.4) and (259.91,319.98) .. (254.25,319.98) .. controls (248.59,319.98) and (244,315.4) .. (244,309.74) -- cycle ;
\draw   (243.5,349.74) .. controls (243.5,344.09) and (248.09,339.5) .. (253.75,339.5) .. controls (259.41,339.5) and (264,344.09) .. (264,349.74) .. controls (264,355.4) and (259.41,359.98) .. (253.75,359.98) .. controls (248.09,359.98) and (243.5,355.4) .. (243.5,349.74) -- cycle ;
\draw [color={rgb, 255:red, 74; green, 144; blue, 226 }  ,draw opacity=1 ][line width=2.25]    (271.5,278.07) -- (259,300.57) ;
\draw [color={rgb, 255:red, 184; green, 233; blue, 134 }  ,draw opacity=1 ][line width=2.25]    (254.25,319.98) -- (253.75,339.5) ;
\draw   (213.5,327.74) .. controls (213.5,322.09) and (218.09,317.5) .. (223.75,317.5) .. controls (229.41,317.5) and (234,322.09) .. (234,327.74) .. controls (234,333.4) and (229.41,337.98) .. (223.75,337.98) .. controls (218.09,337.98) and (213.5,333.4) .. (213.5,327.74) -- cycle ;
\draw   (184,347.74) .. controls (184,342.09) and (188.59,337.5) .. (194.25,337.5) .. controls (199.91,337.5) and (204.5,342.09) .. (204.5,347.74) .. controls (204.5,353.4) and (199.91,357.98) .. (194.25,357.98) .. controls (188.59,357.98) and (184,353.4) .. (184,347.74) -- cycle ;
\draw [color={rgb, 255:red, 184; green, 233; blue, 134 }  ,draw opacity=1 ][line width=2.25]    (201.5,341.07) -- (215,333.07) ;
\draw [color={rgb, 255:red, 184; green, 233; blue, 134 }  ,draw opacity=1 ][line width=2.25]    (231.5,321.07) -- (244.5,313.07) ;
\draw   (288.5,310.74) .. controls (288.5,305.09) and (293.09,300.5) .. (298.75,300.5) .. controls (304.41,300.5) and (309,305.09) .. (309,310.74) .. controls (309,316.4) and (304.41,320.98) .. (298.75,320.98) .. controls (293.09,320.98) and (288.5,316.4) .. (288.5,310.74) -- cycle ;
\draw   (309.5,348.74) .. controls (309.5,343.09) and (314.09,338.5) .. (319.75,338.5) .. controls (325.41,338.5) and (330,343.09) .. (330,348.74) .. controls (330,354.4) and (325.41,358.98) .. (319.75,358.98) .. controls (314.09,358.98) and (309.5,354.4) .. (309.5,348.74) -- cycle ;
\draw [color={rgb, 255:red, 235; green, 51; blue, 51 }  ,draw opacity=1 ][line width=2.25]    (285.75,279.09) -- (294.75,301.09) ;
\draw [color={rgb, 255:red, 235; green, 51; blue, 51 }  ,draw opacity=1 ][line width=2.25]    (290.25,273.59) -- (311.25,276.59) ;
\draw   (300,90.24) .. controls (300,84.59) and (304.59,80) .. (310.25,80) .. controls (315.91,80) and (320.5,84.59) .. (320.5,90.24) .. controls (320.5,95.9) and (315.91,100.48) .. (310.25,100.48) .. controls (304.59,100.48) and (300,95.9) .. (300,90.24) -- cycle ;
\draw   (300,130.74) .. controls (300,125.09) and (304.59,120.5) .. (310.25,120.5) .. controls (315.91,120.5) and (320.5,125.09) .. (320.5,130.74) .. controls (320.5,136.4) and (315.91,140.98) .. (310.25,140.98) .. controls (304.59,140.98) and (300,136.4) .. (300,130.74) -- cycle ;
\draw   (300,169.74) .. controls (300,164.09) and (304.59,159.5) .. (310.25,159.5) .. controls (315.91,159.5) and (320.5,164.09) .. (320.5,169.74) .. controls (320.5,175.4) and (315.91,179.98) .. (310.25,179.98) .. controls (304.59,179.98) and (300,175.4) .. (300,169.74) -- cycle ;
\draw   (300.5,212.24) .. controls (300.5,206.59) and (305.09,202) .. (310.75,202) .. controls (316.41,202) and (321,206.59) .. (321,212.24) .. controls (321,217.9) and (316.41,222.48) .. (310.75,222.48) .. controls (305.09,222.48) and (300.5,217.9) .. (300.5,212.24) -- cycle ;
\draw   (305,245.24) .. controls (305,239.59) and (309.59,235) .. (315.25,235) .. controls (320.91,235) and (325.5,239.59) .. (325.5,245.24) .. controls (325.5,250.9) and (320.91,255.48) .. (315.25,255.48) .. controls (309.59,255.48) and (305,250.9) .. (305,245.24) -- cycle ;
\draw   (310.5,280.74) .. controls (310.5,275.09) and (315.09,270.5) .. (320.75,270.5) .. controls (326.41,270.5) and (331,275.09) .. (331,280.74) .. controls (331,286.4) and (326.41,290.98) .. (320.75,290.98) .. controls (315.09,290.98) and (310.5,286.4) .. (310.5,280.74) -- cycle ;
\draw [color={rgb, 255:red, 245; green, 166; blue, 35 }  ,draw opacity=1 ][line width=2.25]    (288.75,76.59) -- (302.75,83.59) ;
\draw [color={rgb, 255:red, 235; green, 51; blue, 51 }  ,draw opacity=0.97 ][line width=2.25]    (288.25,117.09) -- (302.25,124.09) ;
\draw [color={rgb, 255:red, 245; green, 166; blue, 35 }  ,draw opacity=1 ][line width=2.25]    (288.25,156.09) -- (302.25,163.09) ;
\draw [color={rgb, 255:red, 235; green, 51; blue, 51 }  ,draw opacity=1 ][line width=2.25]    (289.25,197.59) -- (303.25,204.59) ;
\draw [color={rgb, 255:red, 245; green, 166; blue, 35 }  ,draw opacity=1 ][line width=2.25]    (289.75,236.09) -- (306.25,240.59) ;
\draw [color={rgb, 255:red, 245; green, 166; blue, 35 }  ,draw opacity=1 ][line width=2.25]    (304.75,319.09) -- (315.75,339.59) ;
\draw [color={rgb, 255:red, 245; green, 166; blue, 35 }  ,draw opacity=1 ][line width=2.25]    (309.25,313.59) -- (330.25,316.59) ;
\draw   (329.5,320.74) .. controls (329.5,315.09) and (334.09,310.5) .. (339.75,310.5) .. controls (345.41,310.5) and (350,315.09) .. (350,320.74) .. controls (350,326.4) and (345.41,330.98) .. (339.75,330.98) .. controls (334.09,330.98) and (329.5,326.4) .. (329.5,320.74) -- cycle ;
\draw [color={rgb, 255:red, 74; green, 144; blue, 226 }  ,draw opacity=1 ][line width=2.25]    (278.25,80.48) -- (278.25,99.5) ;
\draw [color={rgb, 255:red, 74; green, 144; blue, 226 }  ,draw opacity=1 ][line width=2.25]    (278.75,120.48) -- (278.25,140) ;
\draw [color={rgb, 255:red, 184; green, 233; blue, 134 }  ,draw opacity=1 ][line width=2.25]    (278.25,160.48) -- (278.75,181) ;
\draw [color={rgb, 255:red, 184; green, 233; blue, 134 }  ,draw opacity=1 ][line width=2.25]    (278.75,201.48) -- (278.75,220.5) ;
\draw [color={rgb, 255:red, 74; green, 144; blue, 226 }  ,draw opacity=1 ][line width=2.25]    (278.25,240.98) -- (277.75,260.5) ;
\draw    (200.75,64.09) -- (200.75,136.59) ;
\draw [shift={(200.75,139.59)}, rotate = 270] [fill={rgb, 255:red, 0; green, 0; blue, 0 }  ][line width=0.08]  [draw opacity=0] (8.93,-4.29) -- (0,0) -- (8.93,4.29) -- cycle    ;
\draw [shift={(200.75,61.09)}, rotate = 90] [fill={rgb, 255:red, 0; green, 0; blue, 0 }  ][line width=0.08]  [draw opacity=0] (8.93,-4.29) -- (0,0) -- (8.93,4.29) -- cycle    ;
\draw    (345,62.34) -- (345,96.84) ;
\draw [shift={(345,99.84)}, rotate = 270] [fill={rgb, 255:red, 0; green, 0; blue, 0 }  ][line width=0.08]  [draw opacity=0] (8.93,-4.29) -- (0,0) -- (8.93,4.29) -- cycle    ;
\draw [shift={(345,59.34)}, rotate = 90] [fill={rgb, 255:red, 0; green, 0; blue, 0 }  ][line width=0.08]  [draw opacity=0] (8.93,-4.29) -- (0,0) -- (8.93,4.29) -- cycle    ;
\draw [color={rgb, 255:red, 245; green, 166; blue, 35 }  ,draw opacity=1 ]   (330.73,62.59) -- (330.27,137.09) ;
\draw [shift={(330.25,140.09)}, rotate = 270.36] [fill={rgb, 255:red, 245; green, 166; blue, 35 }  ,fill opacity=1 ][line width=0.08]  [draw opacity=0] (8.93,-4.29) -- (0,0) -- (8.93,4.29) -- cycle    ;
\draw [shift={(330.75,59.59)}, rotate = 90.36] [fill={rgb, 255:red, 245; green, 166; blue, 35 }  ,fill opacity=1 ][line width=0.08]  [draw opacity=0] (8.93,-4.29) -- (0,0) -- (8.93,4.29) -- cycle    ;

\draw (273,265.84) node [anchor=north west][inner sep=0.75pt]  [font=\small]  {$v^{*}$};
\draw (274.5,66.34) node [anchor=north west][inner sep=0.75pt]  [font=\small]  {$u$};
\draw (140.5,86.34) node [anchor=north west][inner sep=0.75pt]    {$d_{L} =2$};
\draw (351,65.34) node [anchor=north west][inner sep=0.75pt]    {$d_{R} =1$};
\draw (339,107.34) node [anchor=north west][inner sep=0.75pt]  [color={rgb, 255:red, 245; green, 166; blue, 35 }  ,opacity=1 ]  {$\text{lcm}( d_{L} ,d_{R}) =2$};
\draw (247,345.84) node [anchor=north west][inner sep=0.75pt]  [font=\footnotesize]  {$v_{L}$};
\draw (312.5,344.84) node [anchor=north west][inner sep=0.75pt]  [font=\footnotesize]  {$v_{R}$};

\draw   (101.5,439.7) .. controls (101.5,435.03) and (99.17,432.7) .. (94.5,432.7) -- (89,432.7) .. controls (82.33,432.7) and (79,430.37) .. (79,425.7) .. controls (79,430.37) and (75.67,432.7) .. (69,432.7)(72,432.7) -- (63.5,432.7) .. controls (58.83,432.7) and (56.5,435.03) .. (56.5,439.7) ;
\draw   (157.9,439.94) .. controls (157.9,435.27) and (155.57,432.94) .. (150.9,432.94) -- (145.4,432.94) .. controls (138.73,432.94) and (135.4,430.61) .. (135.4,425.94) .. controls (135.4,430.61) and (132.07,432.94) .. (125.4,432.94)(128.4,432.94) -- (119.9,432.94) .. controls (115.23,432.94) and (112.9,435.27) .. (112.9,439.94) ;
\draw   (443.5,440.74) .. controls (443.5,436.07) and (441.17,433.74) .. (436.5,433.74) -- (431,433.74) .. controls (424.33,433.74) and (421,431.41) .. (421,426.74) .. controls (421,431.41) and (417.67,433.74) .. (411,433.74)(414,433.74) -- (405.5,433.74) .. controls (400.83,433.74) and (398.5,436.07) .. (398.5,440.74) ;
\draw   (497.5,439.38) .. controls (497.5,434.71) and (495.17,432.38) .. (490.5,432.38) -- (485,432.38) .. controls (478.33,432.38) and (475,430.05) .. (475,425.38) .. controls (475,430.05) and (471.67,432.38) .. (465,432.38)(468,432.38) -- (459.5,432.38) .. controls (454.83,432.38) and (452.5,434.71) .. (452.5,439.38) ;
\draw    (209.6,415.14) -- (209.6,441.94) ;
\draw    (209.6,428.54) -- (360,428.74) ;
\draw    (360,428.74) -- (360,442.14) ;
\draw  [dash pattern={on 0.84pt off 2.51pt}]  (268.27,414.54) -- (268.27,441.34) ;
\draw  [dash pattern={on 0.84pt off 2.51pt}]  (304.6,414.21) -- (304.6,441.01) ;

\draw (49,445.33) node [anchor=north west][inner sep=0.75pt]  [font=\fontsize{1.03em}{1.2em}\selectfont]  [align=left] {{\fontfamily{pcr}\selectfont \textbf{\textcolor[rgb]{0.29,0.56,0.89}{([[]](}\textcolor[rgb]{0.72,0.91,0.53}{([[]](}\textcolor[rgb]{0.29,0.56,0.89}{([[]](}\textcolor[rgb]{0.72,0.91,0.53}{([[]](}\textcolor[rgb]{0,0,0}{))((}\textcolor[rgb]{0.82,0.01,0.11}{)}\textcolor[rgb]{0.96,0.65,0.14}{[])}\textcolor[rgb]{0.82,0.01,0.11}{[])}\textcolor[rgb]{0.96,0.65,0.14}{[])}\textcolor[rgb]{0.82,0.01,0.11}{[])}\textcolor[rgb]{0.96,0.65,0.14}{[])}\textcolor[rgb]{0.82,0.01,0.11}{[])}\textcolor[rgb]{0.96,0.65,0.14}{[])}}}};
\draw (67.2,399.54) node [anchor=north west][inner sep=0.75pt]    {$C_{L}$};
\draw (123.2,400.34) node [anchor=north west][inner sep=0.75pt]    {$C_{L}$};
\draw (462.8,399.94) node [anchor=north west][inner sep=0.75pt]    {$C_{R}$};
\draw (407.6,400.34) node [anchor=north west][inner sep=0.75pt]    {$C_{R}$};
\draw (203.2,390.94) node [anchor=north west][inner sep=0.75pt]    {$v^{*}$};
\draw (260.53,393.28) node [anchor=north west][inner sep=0.75pt]    {$v_{L}$};
\draw (295.2,393.61) node [anchor=north west][inner sep=0.75pt]    {$v_{R}$};



\end{tikzpicture}

%% file: src/full_per.tex
\begin{lemma}\label{lem:fur}
Let $\lab$ be a joint labeling of forests $\F,\G$ resulting in labeled forests $F,G$,
and let $\hlab = \CR{\look{\lab}{8k}}{2k}$ for some $k\in \Zp$.
If $F,G$ avoid both synchronized horizontal $k$-periodicity and synchronized vertical $k$-periodicity,
then $\str[\hlab]{\F}$ and $\str[\hlab]{\G}$ avoid $2k$-synchronized $(20k+2)$-powers with root at most $4k$.
\end{lemma}
\begin{proof}
For a proof by contradiction, suppose that there is a string $\hQ$ of length $\hq\in (0\dd 4k]$ such that $\hQ^{20k+2}=\str[\hlab]{\F}[x\dd x+(20k+2)\hq)=
\str[\hlab]{\G}[y\dd y+(20k+2)\hq)$ for some positions $x,y$ with $|x-y|\le 2k$.
Note that $\hlab$ is a refinement of $\look{\lab}{8k}$,
which, in turn, is a refinement of $\lab':= \look{\lab}{4k}$.
In particular, $\str[\lab']{\F}[x\dd x+(20k+2)\hq)=\str[\lab']{\G}[y\dd y+(20k+2)\hq)$ is a string with period $\hq$.
Let $q$ be the shortest period of this string and let $Q$ be the underlying string period.

If $Q$ contains the same number of opening and closing parentheses, then, since $Q^2$ occurs in a balanced string $\str[\lab']{\F}$,
we conclude that some cyclic rotation of $Q$ is balanced, and therefore $F,G$ do not avoid synchronized horizontal $k$-periodicity (because $\lab'$ is a refinement of $\lab$).
Thus, by symmetry, we may assume without loss of generality that $Q$ contains more opening than closing parentheses.
Suppose that $Q[\delta]$ is the leftmost unmatched opening parenthesis within $Q$
and define nodes $u_0,\ldots,u_{20k+1}\in V_\F$ such that $o_\F(u_i)  = x +\delta + iq$
and nodes $v_0,\ldots,v_{20k+1}\in V_\G$ such that $o_\G(v_i) = y+\delta+iq$.
Observe that $c_\F(u_0) > \cdots > c_\F(u_{20k+1}) > o_\F(u_{20k+1})$
and $c_\G(v_0) > \cdots > c_\G(v_{20k+1}) > o_\G(v_{20k+1})$.

Next, we prove two claims regarding $2k$-compatibility with respect to $\look{\lab}{8k}$, 
which we simply refer to as compatibility in the remainder of the proof.
\begin{claim}\label{clm:one}
If there exists $i\in [0\dd 20k]$ and a node $u'\in V_\F$ with $o_\F(u')\in [o_\F(u_0)\dd o_\F(u_{20k})]$
such that $\hlab(u_i)=\hlab(u')$, then $u'=u_{i'}$ for some $i'\in [0\dd 20k]$.
\end{claim}
\begin{proof}
Note that $\hlab(u_i)=\hlab(u')$ implies $\str[\lab]{\Sub_{< 8k}(u')}=\str[\lab]{\Sub_{< 8k}(u_i)}$ and $\str[\lab']{\Sub_{< 4k}(u')}=\str[\lab']{\Sub_{< 4k}(u_i)}$. We thus have $\str[\lab']{\F}[o_\F(u')\dd o_\F(u')+q)=\str[\lab']{\F}[o_\F(u_i)\dd o_\F(u_i)+q)$ because the common size of $\Sub_{< 4k}(u')$ and $\Sub_{< 4k}(u_i)$ is at least $4k\ge q$.
Since $q$ is the shortest period of $\str[\lab']{\F}[o_\F(u_0)\dd o_\F(u_{20k})+q)$, we conclude that $o_\F(u')-o_\F(u_0)$ is a multiple of $q$,
and hence $u'=u_{i'}$ for some $i'\in [0\dd 20k]$.
\end{proof}

\begin{claim}\label{clm:two}
For all $i\in [0\dd 20k)$, we have $c_\F(u_i)\le c_\F(u_{i+1})+4k$ 
and $c_\G(v_i)\le c_\G(v_{i+1})+4k$ 
\end{claim}
\begin{proof}
We focus on the claim regarding $\F$ (the claim regarding $\G$ is symmetric).
For a proof by contradiction, suppose that $c_\F(u_i)>  c_\F(u_{i+1})+4k$ holds for some $i\in [0\dd 20k)$.
By construction, the node $u_i$ shares the $\hlab$-label with its descendant whose opening parenthesis is located at position $o_\F(u_i)+\hq$.
Consider the underlying path in the compatibility graph, let $u'\in V_\F$ be the last node on this path that is an ancestor of $u_i$ (possibly $u' = u_i$), and let $u''\in V_\F$ be the subsequent node of $\F$ on this path.
By definition of compatibility, we have $|o_\F(u')-o_\F(u'')|\le 4k$ and $|c_\F(u')-c_\F(u'')|\le 4k$.
Moreover, since $u'$ is an ancestor of $u_i$, we have $o_\F(u')\le o_\F(u_i)$ and $c_\F(u')\ge c_\F(u_i)$.
We conclude the proof by deriving a contradiction for every possible location of $o_\F(u'')$.
\begin{itemize}
    \item If $o_\F(u'')\le o_\F(u_i)$, then either $c_\F(u'')\le o_\F(u_i)$, which implies $c_\F(u'') \le o_\F(u_i) < o_\F(u_{i+1}) < c_\F(u_{i+1}) < c_\F(u_i)-4k \le c_\F(u')-4k \le c_\F(u'')$ (a contradiction) or $c_\F(u'')\ge c_\F(u_i)$, which means that $u''$ is an ancestor of $u_i$ and contradicts the choice of $u'$.
    \item If $o_\F(u'')\in (o_\F(u_i)\dd o_\F(u_{i+1}))$, then a contradiction follows from \cref{clm:one}

    \item If $o_\F(u'')\in [o_\F(u_{i+1})\dd c_\F(u_{i+1})]$, then $u''$ is an ancestor of $u_{i+1}$, which means that $c_\F(u'')\le c_\F(u_{i+1})
    < c_\F(u_i)-4k \le c_\F(u')-4k \le c_\F(u'')$ (a contradiction).
    \item If $o_\F(u'')> c_\F(u_{i+1})$, then $o_\F(u'')>c_\F(u_{i+1}) > o_\F(u_{i+1})+4k > o_\F(u_i)+4k \ge o_\F(u')+4k \ge o_\F(u'')$, which is also a contradiction.\qedhere
\end{itemize}
\end{proof}

The nodes $(u_i)_{i\in [0\dd 20k]}$ and $(v_i)_{i\in [0\dd 20k]}$ share the same $\lab'$-label, so the subtrees $\str[\lab]{\Sub_{< 4k}(u_i)}$ and $\str[\lab]{\Sub_{< 4k}(v_i)}$ are all isomorphic.
Consequently, by \cref{clm:two}, the context $C:=(C_L,C_R):=(\str[\lab]{\F}[o_\F(u_0)\dd o_\F(u_1)), \str[\lab]{\F}(c_\F(u_1)\dd c_\F(u_0)])$ satisfies $0 < |C_L|,|C_R| < 4k$ and occurs at all nodes $(u_i)_{i\in [0\dd 20k)}$ and $(v_i)_{i\in [0\dd 20k)}$. Its power $C^{20k}$ occurs in $F$ at node $u_0$ and in $G$ at node $v_0$.

Observe that, for every $i\in [0\dd 20k]$, the node $v_{i}$ must be compatible with some $u'\in V_\F$.
If  $i\in [4k\dd 16k]$, then $o_\F(u')\ge o_\G(v_{i})-2k \ge o_\F(u_i)-4k \ge o_\F(u_{i-4k})$ and $o_\F(u')\le o_\G(v_{i})+2k \le o_\F(u_i)+4k \le o_\F(u_{i+4k})$, so $o_\F(u')\in [o_\F(u_{i-4k})\dd o_\F(u_{i+4k})]$.
By \cref{clm:one}, this means that $u'=u_j$ for some $j\in [i-4k\dd i+4k]$.
In particular, $|o_\F(u_i)-o_\G(v_{j})|\le 2k$ and $|c_\F(u_i)-c_\G(v_{j})|\le 2k$.
Since $(o_\F(u_t))_{t\in [0\dd 20k)}$ and $(o_\G(v_t))_{t\in [0\dd 20k)}$ form arithmetic progressions with difference $|C_L|$
and $(c_\F(u_t))_{i\in [0\dd 20k)}$ and $(c_\G(v_i))_{t\in [0\dd 20k)}$ form arithmetic progressions with difference $-|C_R|$,
we conclude that there exists $\delta:=j-i\in [-4k\dd 4k]$ such that $|o_\F(u_t)-o_\G(v_{t+\delta})|\le 2k$
and $|c_\F(u_t)-c_\G(v_{t+\delta})|\le 2k$ hold whenever $t,t+\delta\in [0\dd 20k]$.
This means that $C^{16k}$ has $2k$-synchronized occurrences in $F$ and $G$ (at nodes $u_0, v_{\delta}$ if $\delta \ge 0$, and at nodes $u_{-\delta}, v_0$ otherwise), contradicting the assumption that $F,G$ avoid synchronized vertical $k$-periodicity.
\end{proof}

\begin{proposition}\label{prp:fur}
    There exists a randomized algorithm that, given forests $F,G$ and a threshold $k\in \Zp$,
    produces forests $F',G'$ and an alignment $\A: \str{F'}\leadsto \str{G'}$ such that:
    \begin{itemize}
        \item $\ted_{\le k}(F,G)=\ted_{\le k}(F',G')$, and 
        \item $|\A \triangle \B| \le 4928k^4$ for every alignment $\B\in \ta_k(F',G')$.
    \end{itemize}
The running time is $\Oh(n\log n+k^3)$ and the algorithm is correct w.h.p.
\end{proposition}
\begin{proof}
The forests $F',G'$ are produced by horizontal (\cref{prp:hor}) and vertical (\cref{prp:ver}) periodicity reduction;
this guarantees $\ted_{\le k}(F,G)=\ted_{\le k}(F',G')$.
Let $\F',\G'$ be the underlying unlabeled forests and let $\lab'$ be their joint labeling.
We use \cref{lem:buildLookahead,lem:buildC} to construct $\hlab = \CR{\look{\lab'}{8k}}{2k}$ and strings $\str[\hlab]{\F'}$, $\str[\hlab]{\G'}$.
Note that the width of any $\B\in \ta_k(F',G')$ does not exceed $2k$ and, by \cref{lem:boundLookahead,obs:boundCompatibility}, we have $\ed_{\B}(\str[\hlab]{\F'},\str[\hlab]{\G'})\le 2k\cdot 8k \le 16k^2$. Thus, $\B\in \aa_{16k^2,2k}(\str[\hlab]{\F'},\str[\hlab]{\G'})$. 
We construct $\A$ in $\Oh(n+k^3)$ time using \cref{lem:build_greedy} as an arbitrary alignment in $\ga_{16k^2,2k}(\str[\hlab]{\F'},\str[\hlab]{\G'})$; if there is no such alignment, then $\aa_{16k^2,2k}(\str[\hlab]{\F'},\str[\hlab]{\G'})=\emptyset$ and hence $\ta_k(F',G')=\emptyset$.
By \cref{lem:fur,lem:sheep}, we have $|\A \triangle \B| \le 7 \cdot 2k \cdot 16k^2 \cdot (20k+2) \le 4928k^4$.
\end{proof}

%% file: src/main_alg.tex
\subsection{Partial Forest Matching}\label{lem:partial}

Let $F$ and $G$ be labeled forests. 
We say that a set  $M \sub V_F\times V_G$ is \emph{non-crossing} if the set $\bigcup_{(u,v)\in M}\{(o_F(u),o_F(v)),\allowbreak (c_F(u),c_F(v))\}\sub \Zz\times \Zz$ is non-crossing. Furthermore, we say that $M\sub V_F\times V_G$ is a \emph{non-crossing matching} of forests $F,G$ if 
$\bigcup_{(u,v)\in M}\{(o_F(u),o_F(v)),\allowbreak (c_F(u),c_F(v))\}$ is a non-crossing matching of $\str{F},\str{G}$.

We say that a tree alignment $\A\in \ta(F,G)$ \emph{aligns} $u\in V_F$ with $v\in V_G$, denoted $u \sim_{\A} v$, if $\str{F}[o_F(u)] \sim_{\A} \str{G}[o_G(v)]$
and  $\str{F}[c_F(u)] \sim_{\A} \str{G}[c_G(v)]$. If, additionally, $u$ and $v$ have the same labels,
we say that $\A$ \emph{matches} $u$ and $v$, denoted $u \simeq_{\A} v$.
For a set $M\sub V_F\times V_G$, we define $\ta^M(F,G) = \{\A \in \ta(F,G): u \simeq_{\A} v \text{ for each }(u,v)\in M\}$;
note that $\ta^M(F,G)=\emptyset$ unless $M$ is a non-crossing matching.
Moreover, we denote $\ted^M(F,G) = \min_{\A \in \ta^M(F,G)} \ted_{\A}(F,G)$.

In this section, we consider the problem of computing $\ted^M(F,G)$ given labeled forests $F,G$ and a non-crossing matching $M\sub V_F\times V_G$.
The following lemma provides a reduction that restricts the heights of  $F,G$ and lets us assume that $M\sub L_F\times L_G$,
where $L_F\sub V_F$ is the set of leaves of $F$ and $L_G\sub V_G$ is the set of leaves of $G$.

\begin{lemma}\label{lem:reduceHeight}
    There exists a linear-time algorithm that, given labeled forests $F,G$ and a non-crossing matching $M\sub  V_F\times V_G$,
    produces labeled forests $F',G'$ and a non-crossing matching $M'\sub L_{F'}\times L_{G'}$ 
    such that:
    \begin{itemize}
        \item $\ted^{M'}(F',G') = \ted^{M}(F,G)$;
        \item $|F'|= |M| + |F|$, $|G'|= |M|+|G|$, and $|M'|= 2|M|$;
        \item if the height of $F'$ ($G'$) is $h> 1$, then there is an $(h-1)$-node top-down path in $F$ (respectively, $G$) avoiding nodes participating in $M$.
    \end{itemize}
\end{lemma}
\begin{figure}[ht]
    \centering
    \scalebox{0.63}{\input{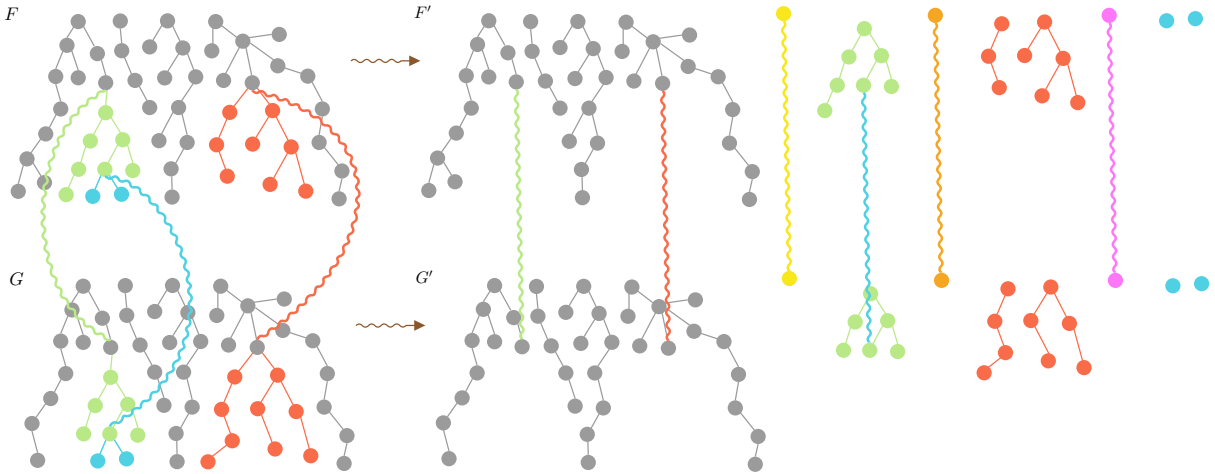}}
    \caption{An illustration of the transformation implemented in \cref{lem:reduceHeight}. The non-crossing matching is represented using colorful wavy lines. In the input forests $F$ and $G$, each node inherits its color from the nearest marked ancestor.}%
    \label{fig:partial_forest}
\end{figure}

\begin{proof}
Let us denote $M = \{(u_i,v_i) : i\in [1\dd m]\}$
and \emph{mark} all nodes participating in $M$. 
We decompose $V_F = \bigcup_{i=0}^m V_{F_i}$ so that $u\in V_{F_0}$ if $u$ does not have any proper marked ancestor
and, for every $i\in [1\dd m]$, we have $u\in V_{F_i}$ if $u_i$ is the nearest proper marked ancestor of $u$.
We further define forests $F_0,\ldots, F_m$ so that $F_i$ is obtained from $F$ by deleting all vertices in $V_F \sm V_{F_i}$.

As far as the implementation is concerned, a top-down traversal of $F$ allows classifying each node $u\in V_F$ into one of the classes $V_{F_i}$.
This is because every root node belongs to $V_{F_0}$, every node with a marked parent $u_i$ belongs to the class $V_{F_i}$,
and every other node belongs to the same class as its unmarked parent.
Consequently, the forest $F$ can be decomposed into the forests $F_0,\ldots,F_m$ in linear time.
Symmetrically, we decompose the forest $G$ into analogously defined forests $G_0,\ldots,G_m$.

Due to the fact that $M$ is non-crossing, this yields a decomposition $M = \bigcup_{i=0}^m M_i$, where $M_i = M \cap (V_{F_i}\times V_{G_i})$. 
Furthermore, if $u\sim_{\A} v$ holds for an alignment $\A \in \ta^M(F,G)$,
then the set $M \cup \{(u,v)\}\sub V_F\times V_G$ is non-crossing, and thus there exists $i\in [0\dd m]$ such that $u\in V_{F_i}$ and $v \in V_{G_i}$.
Consequently, we have $\ted^M(F,G)=\sum_{i=0}^m \ted^{M_i}(F_i,G_i)$.

In the second step of the algorithm, we create nodes $\hu_1,\ldots,\hu_m$ and $\hv_1,\ldots,\hv_m$, all sharing the same label.
The forest $F'$ is obtained as a (horizontal) concatenation $F' := F_0 \cdot \hu_1 \cdot F_1  \cdots  \hu_m \cdot F_m$,
where each node $\hu_i$ is interpreted as a single-node forest.
Symmetrically, $G' := G_0 \cdot \hv_1 \cdot G_1  \cdots  \hv_m \cdot G_m$.
Finally, we set $M' := M \cup \{(\hu_i,\hv_i) : i\in [1\dd m]\}$.
The forests $F',G'$ and the set $M'$ can be easily constructed in linear time.

The construction trivially yields $|F'|=|F|+|M|$, $|G'|=|G|+|M|$, and  $|M'|=2|M|$.
Since the separator nodes are included in $M'$, we further have $\ted^{M'}(F',G')=\sum_{i=0}^m \ted^{M_i}(F_i,G_i)=\ted^M(F,G)$.
The separator nodes are leaves and $M_i \sub L_{F_i}\times L_{G_i}$ holds for each $i\in [0\dd m]$ (by definition of the decompositions of $F$ and $G$),
so we have $M' \sub L_{F'}\times L_{G'}$.
Finally, observe that if the height of $F'$ is $h>1$,
then one of the forests $F_i$ has an $h$-node root-to-leaf path.
Trimming the leaf, we obtain an $(h-1)$-node top-down path avoiding marked nodes.
By construction of $F_i$, this path is also present in $F$.
A symmetric reasoning lets us characterize the height of $G'$.
\end{proof}

The following reduction prunes unnecessary leaves.
This is useful if $M\sub  L_F\times L_G$ is very large.

\begin{lemma}\label{lem:reduce}
    There exists a linear-time algorithm that, given labeled forests $F,G$ and a non-crossing matching $M\sub  L_F\times L_G$
    produces  labeled forests $F',G'$ and a non-crossing matching $M'\sub L_{F'}\times L_{G'}$ 
    such that:
    \begin{itemize}
        \item $F'$ and $G'$ are obtained by deleting some leaves in $F$ and $G$, respectively;
        \item $M' = M\cap (V_{F'}\times V_{G'})=M\cap (V_{F}\times V_{G'})=M\cap (V_{F'}\times V_{G})$ satisfies $|M'|\le \frac25(|F'|+|G'|+1)$;
        \item $\ted^{M'}(F',G')=\ted^{M}(F,G)$.
    \end{itemize}
\end{lemma}
\begin{proof}
    The algorithm identifies a subset $\hM$ of \emph{redundant leaf pairs}
    and constructs $(F',G',M')$ so that $M'=M\sm \hM$
    whereas the forests $F'$ and $G'$ are obtained be deleting all nodes participating in $\hM$.
    Specifically, $\hM$ contains all pairs $(\hu,\hv)\in M$ such that $\hu$ and $\hv$ have immediate left siblings $u$
    and $v$, respectively, such that $(u,v)\in M$.
    The construction trivially yields $\ted^{M'}(F',G')\le \ted^{M}(F,G)$.
    As for the inverse inequality, we shall prove that every alignment $\A'\in \ted^{M'}(F',G')$,
    can be converted to an alignment $\A\in \ted^{M}(F,G)$ of the same cost.
    For this, we simply extend $\A'$ so that $\hu \simeq_{\A} \hv$ holds for all $(\hu,\hv)\in \hM$.
    A simple inductive argument (reintroducing $(\hu,\hv)\in \hM$ in the left-to-right order)
    shows that this does not introduce any crossings among the aligned pairs of vertices.
    This is because any existing pair crossing $(\hu,\hv)$ would also cross the pair $(u,v)$ of immediate left siblings.

    Next, we note that if $(u',v')\in M'$, then one of the following cases holds:
    \begin{itemize}
        \item $u'$ or $v'$ is the leftmost root of $F'$ (resp. $G'$);
        \item $u'$ or $v'$ is the leftmost child of its parent (and the parent is not participating in $M'$);
        \item $u'$ or $v'$ has an immediate left sibling that is not participating in $M'$.
    \end{itemize}
    The number of nodes not participating in $M'$ is $|F'|+|G'|-2|M'|$,
    and each of them can be charged twice (by its right sibling and by its leftmost child).
    Hence, $|M'|\le 2(1+|F'|+|G'|-2|M'|)$, and this simplifies to $|M'|\le \frac25(|F'|+|G'|+1)$.

    Finally, we observe that $F'$, $G'$, and $M'$ can be easily constructed in linear time.
\end{proof}

The last result of this section reduces computing $\ted^M_{\le k}(F,G)$ to computing $\ted_{\le k}(F',G')$
for some $F'$~and~$G'$.

\begin{lemma}\label{lem:gadget}
    There exists an algorithm that, given labeled forests $F,G$, a non-crossing matching $M\sub  L_F\times L_G$, and an integer $k\in \Zz$,
    produces labeled forests $F',G'$ satisfying the following conditions:
    \begin{itemize}
        \item $\ted_{\le k}(F',G')=\ted^{M}_{\le k}(F,G)$;
        \item $|F'|= |F|+(k+1)|M|$ and $|G'|= |G|+(k+1)|M|$;
        \item The height of $F'$ ($G'$) does not exceed the height of $F$ (resp., $G$) plus one.
    \end{itemize}
The running time of the algorithm is $\Oh(|F'|+|G'|)$.
\end{lemma}
\begin{proof}
    Let $M = \{(u_i,v_i) : i\in [1\dd m]\}$.
    For each $i\in [1\dd m]$, we attach $k+1$ new nodes, $u_{i,0},\ldots,u_{i,k}$,
    as children of $u_i$ and $k+1$ new nodes, $v_{i,0},\ldots,v_{i,k}$, as children of $v_i$.
    The nodes $u_{i,j}$ and $v_{i,j}$ share a unique label $\$_{i,j}$ that is not present anywhere
    else in the constructed forests $F',G'$.
    This completes the description of the constructed forests $F',G'$.
    It is easy to see that $F',G'$ can be constructed in time proportional to their sizes,
    which are $|F|+(k+1)|M|$ and $|G|+(k+1)|M|$, respectively.
    Moreover, the height of $F'$ (resp. $G'$) may increase compared to the height of $F$ ($G$) by at most one
    (since we attach new leaves only to existing nodes).

    Thus, it remains to prove that $\ted^{M}_{\le k}(F,G)=\ted_{\le k}(F',G')$.
    For this, let us first observe that $\ted(F',G')\le \ted^M(F,G)$.
    This is because, any $\A \in \ta^M(F,G)$
    can be extended to an alignment $\A' \in \ta(F',G')$ such that $u_{i,j}\simeq_{\A'} v_{i,j}$ holds for all $i\in [1\dd m]$
    and $j\in [0\dd k]$. If any of the newly matched pairs $(u_{i,j},v_{i,j})$ crossed some other pair of vertices aligned by $\A'$,
    already $(u_i,v_i)$, with $u_i \simeq_{\A} v_i$, would cross that pair.
    As for the converse inequality, consider an alignment $\A' \in \ta(F',G')$ of cost at most $k$.
    By \cref{lem:optIsGreedy}, we can choose $\A'$ so that it can be interpreted as a greedy alignment with respect to (unbounded) look-ahead labels.
    We shall prove that, for each $i\in [1\dd m]$, we have $u_i \simeq_{\A'} v_i$ and $u_{i,j}\simeq_{\A'} v_{i,j}$ for $j\in [0\dd k]$.
    As a result, a matching $\A \in \ta^M(F,G)$ of the same cost can be obtained by restricting $\A'$ to the nodes of $F$ and $G$.

\newcommand{\jh}{\hat{\jmath}}
    Let us fix $i\in [1\dd m]$. As the cost of $\A'$ is at most $k$, this alignment must match the node $u_{i,\jh}$ for some $\jh\in [0\dd k]$.
    The only vertex in $G'$ sharing the label with $u_{i,\jh}$ is $v_{i,\jh}$, so we must have $u_{i,\jh} \simeq_{\A'} v_{i,\jh}$
    and, in particular, $\str{F'}[c_{F'}(u_{i,\jh})]\simeq_{\A'} \str{G'}[c_{G'}(v_{i,\jh})]$.
    Furthermore, observe that the nodes $(u_i,v_i)$ and the nodes $(u_{i,j},v_{i,j})$ for $j\in [0\dd k]$ not only share their regular labels,
    but also their (unbounded) look-ahead labels. Consequently, the greedy nature of $\A'$ yields 
    $\str{F'}[c_{F'}(u_{i,\jh})\dd c_{F'}(u_i)]\simeq_{\A'} \str{G'}[c_{G'}(v_{i,\jh})\dd c_{G'}(v_i)]$.
    Since $\A'$ is a tree alignment, $\str{F'}[c_{F'}(u_i)]\simeq_{\A'} \str{G'}[c_{G'}(v_i)]$
    implies $\str{F'}[o_{F'}(u_i)]\simeq_{\A'} \str{G'}[o_{G'}(v_i)]$.
    Using the greedy nature of $\A'$ once again, we finally conclude that
    $\str{F'}[o_{F'}(u_i)\dd c_{F'}(u_i)]\simeq_{\A'} \str{G'}[o_{G'}(v_i)\dd c_{G'}(v_i)]$.
    Thus, $u_i \simeq_{\A'} v_i$ and $u_{i,j}\simeq_{\A'} v_{i,j}$ for $j\in [0\dd k]$ hold as claimed.
\end{proof}

We conclude with a corollary that applies \cref{lem:reduceHeight,lem:reduce,lem:gadget} one after another.

\begin{corollary}\label{cor:partial}
    There exists an algorithm that, given labeled forests $F,G$, a non-crossing matching $M\sub  V_F\times V_G$,
    and an integer $k\in \Zz$, produces labeled forests $F',G'$ such that:
    \begin{itemize}
        \item $\ted_{\le k}(F',G') = \ted^{M}_{\le k}(F,G)$;
        \item $|F'|+|G'| = \Oh(\min(|F|+|G|+k|M|, 1+(k+1)(|F|+|G|-2|M|)))$;
        \item if the height of $F'$ (or $G'$) is $h> 2$, then there is an $(h-2)$-node top-down path in $F$ (respectively, $G$) avoiding nodes participating in $M$.
    \end{itemize}
    The running time of the algorithm is $\Oh(|F|+|G|+|F'|+|G'|)$.
\end{corollary}
\begin{proof}
    As hinted above, the algorithm behind \cref{cor:partial} simply chains the procedures underlying \cref{lem:reduceHeight,lem:reduce,lem:gadget}.
    Let us denote the intermediate results by $(\hF,\hG,\hM)$ and $(\bF,\bG,\bM)$, respectively.
    Due to $\ted_{\le k}(F',G')=\ted^{\bM}_{\le k}(\bF,\bG)$
    and $\ted^{\bM}(\bF,\bG)=\ted^{\hM}(\hF,\hG)=\ted^M(F,G)$, we conclude that $\ted_{\le k}(F',G') = \ted^{M}_{\le k}(F,G)$
    holds as claimed.
    As for the heights, observe that if the height of $F'$ is $h>2$,
    then the heights of $\bF$ are $\hF$ are at least $h-1>1$. By \cref{lem:reduceHeight}, this means that $F$ contains a top-down path of at least $h-2$ nodes none of which participate in $M$. A symmetric argument can be used to bound the height of $G'$.

    It remains to bound $|F'|+|G'|$.
    Applying \cref{lem:gadget}, \cref{lem:reduce}, and \cref{lem:reduceHeight} in subsequent steps,
    we obtain
    \begin{multline*}
        |F'|+|G'| \le |\bF|+|\bG|+2(k+1)|\bM| \le |\hF|+|\hG|+2(k+1)|\hM| \le |F|+|M|+|G|+|M|+4(k+1)|M| \\= \Oh(|F|+|G|+k|M|).
    \end{multline*}
    However, we further have
    \begin{multline*}
    2|\bM|\le |\bF|+|\bG| \le |\bF|+|\bG| - 10|\bM|+10|\bM| \le |\bF|+|\bG| - 10|\bM|+4(|\bF|+|\bG|+1) \\ =
     5(|\bF|+|\bG|-2|\bM|)+4 = 5(|\hF|+|\hG|-2|\hM|)+4 = 5(|F|+|G|-2|M|)+4,
    \end{multline*}
    and thus 
    \[ |F'|+|G'| \le |\bF|+|\bG|+2(k+1)|\bM| \le (k+2)(|\bF|+|\bG|) = \Oh(1+(k+1)(|F|+|G|-2|M|)).\]
    Consequently, we indeed have  $|F'|+|G'| = \Oh(\min(|F|+|G|+k|M|, 1+(k+1)(|F|+|G|-2|M|)))$.
\end{proof}

\subsection{Tree Edit Distance of Shallow Forests}\label{sec:smallh}
\newcommand{\MP}{\ensuremath{M_{\mathsf{P}}}}

\begin{theorem}\label{thm:smallh}
    There exists a randomized algorithm that, given forests $F,G$ of height at most $h\in \Zp$
    and a threshold $k\in \Zp$, computes $\ted_{\le k}(F,G)$ in $\Oh(n\log n+h^2 k^7 \log(hk))$ time correctly with high probability.
    \end{theorem}
    \begin{proof}
    Let $\F$ and $\G$ be the underlying unlabeled forests and let $\lab$ be their joint labeling. 
    Using \cref{prp:hor}, at the cost of $\Oh(n)$ time, we can assume without loss of generality
    that $F,G$ avoid synchronized horizontal $k$-periodicity,
    i.e., there is no balanced string $Q$ of length $|Q|\le 4k$ such $Q^{18k}$ has $2k$-synchronized occurrences in $\str[\lab]{\F}$ and $\str[\lab]{\G}$.
    In the next step, we construct a labeling $\hlab$ equivalent to $\look{\lab}{h}$ (using \cref{lem:buildLookahead})
    as well as the strings $\str[\hlab]{\F}$ and $\str[\hlab]{\G}$.
    Since $\hlab$ is a refinement of $\lab$, we conclude 
    that there is no balanced string $Q$ of length $|Q|\le 4k$ such $Q^{18k}$ has $2k$-synchronized occurrences in $\str[\hlab]{\F}$ and $\str[\hlab]{\G}$.
    Moreover, $\hlab$ assigns distinct labels to nodes in any root-to-leaf path; 
    thus, for any unbalanced string $Q$, even $Q^2$ cannot occur in  $\str[\hlab]{\F}$ or $\str[\hlab]{\G}$.
    Consequently, $\str[\hlab]{\F}$ or $\str[\hlab]{\G}$ 
    avoid $2k$-synchronized $18k$-powers of length at most $4k$.
    This lets us apply \cref{lem:almostEverythingIsCommon} to build a set $\MP$ of size $|\MP| \ge 2|F| - 15\cdot (18k)\cdot (2hk)^2 \cdot 2k \ge  2|F|-\Oh(h^2 k^4)$ such that $\MP\sub \mtch_\A(\str[\hlab]{\F},\str[\hlab]{\G})$ holds for every $\A \in \ga_{2hk,2k}(\str[\hlab]{\F},\str[\hlab]{\G})$; the construction of $\MP$ costs $\Oh(n+hk^2)$ time.
    \begin{claim}\label{clm:shallow}
    Let $M = \{(u,v)\in V_F\times V_G: (o_F(u),o_G(v))\in \MP\text{ and }(c_F(u),c_G(v))\in \MP\}$. 
    If $\ted(F,G)\le k$, then  $M$ is a non-crossing matching of size $|M|\ge |F|-\Oh(h^2 k^4)$ and, moreover, $\ted^M(F,G)=\ted(F,G)$.
    \end{claim}
    \begin{proof}
    By \cref{lem:optIsGreedy}, there is an optimum alignment $\A \in \ta(F,G)$
    of cost $\ted_{\A}(F,G)=\ted(F,G)$ such that $\A \in \ga(\str[\hlab]{\F},\str[\hlab]{\G})$.
    If $\ted(F,G)\le k$, then the width of $\A$ is at most $2k$ and, by \cref{lem:boundLookahead}, its cost $\ed_{\A}(\str[\hlab]{\F},\str[\hlab]{\G})$ does not exceed $2hk$. Hence, $\A\in \ga_{2hk,2k}(\str[\hlab]{\F},\str[\hlab]{\G})$, which means that $\MP \sub \mtch_\A(\str[\hlab]{\F},\str[\hlab]{\G})$ and, in particular, $\A \in \ta^M(F,G)$. 
    Consequently, $\ted^M(F,G)=\ted(F,G)$ and the number of vertices $u\in V_\F$ not participating in $M$ does not exceed $2|F|-|\MP|=\Oh(h^2 k^4)$.
    \end{proof}
    In the light of \cref{clm:shallow}, we construct $M$ and check whether it is a non-crossing matching of size $|M| \ge |F|-\Oh(h^2k^4)$; if it is not, we report that $\ted(F,G)>k$.
    The remaining task is to compute $\ted^M_{\le k}(F,G)$.
    For this, we use \cref{cor:partial}, which reduces this problem to computing $\ted_{\le k}(F',G')$, where $|F'|+|G'|=\Oh(k\cdot h^2 k^4)$
    because $|F|+|G|-2|M| = \Oh(h^2k^4)$; this reduction costs $\Oh(n+h^2k^5)$ time.
    As for the final step of determining $\ted_{\le k}(F',G')$, we employ the algorithm of~\cite{DBLP:conf/icalp/AkmalJ21},
    whose running time is $\Oh(h^2k^5\cdot k^2 \cdot \log(h^2k^5)) = \Oh(h^2k^7 \log (hk))$.
    Overall, our procedure takes $\Oh(n + h^2 k^7 \log(hk))$ time.
    \end{proof}

\subsection{Level Sampling}\label{sec:sampl}

\main*

\SetKwFunction{TED}{TreeEditDistance}
\SetKwFunction{PR}{FullPeriodicityReduction}
\SetKwFunction{PFM}{PartialMatchingReduction}
\SetKwFunction{STED}{ShallowTreeEditDistance}
\begin{algorithm}
\KwIn{labeled forests $F,G$, integer threshold $k\in \Zp$}
\KwOut{$\ted_{\le k}(F,G)$}
$(F',G',\A):=\PR(F,G,k)$\tcp*{\cref{prp:fur}}
$d:=\infty$\;
$h:=19716k^4$\;
\For{$i:=0$ \KwSty{to} $\Theta(\log (|F|+|G|))$}{
    Pick $r_i \in [0\dd h)$ uniformly at random\;
    Mark nodes in $F',G'$ at depths $\equiv r_i \pmod{h}$\;
    $M_i:=\{(u,v)\in V_{F'}\times V_{G'} : \str{F'}[o_{F'}(u)]\simeq_{\A} \str{G'}[o_{G'}(v)]\text{, and }\hspace{4cm}\allowbreak\phantom{blablabla}\str{F'}[c_{F'}(u)]\simeq_{\A} \str{G'}[c_{G'}(v)]\text{, and $u$ or $v$ is marked}\}$\;
    \If{$|M_i|\le \frac{4}{h}(|F'|+|G'|)$ \KwSty{and} all marked nodes participate in $M_i$}{
    $(F_i,G_i):=\PFM(F',G',M_i,k)$\tcp*{\cref{cor:partial}}
    $d_i := \STED(F_i,G_i,k)$\tcp*{\cref{thm:smallh}}
    $d := \min(d, d_i)$\;
    }
}
\Return{$d$}\;
\caption{\protect\TED{$F,G,k$}}
\end{algorithm}

\begin{proof}
As a first step, we apply \cref{prp:fur}, which produces forests $F',G'$ and an alignment  $\A: \str{F'}\leadsto \str{G'}$.
Next, we pick $h:= 19716k^4$ and draw $s = \Theta(\log n)$ uniformly random values $r_1,\ldots,r_s\in [0\dd h)$.
As described below, each of these values $r_i$ is either discarded or results in an upper bound on $\ted_{\le k}(F',G')$.
We mark all nodes in $F'$ and $G'$ whose depths are congruent to $r_i$ modulo $h$.
We then attempt using $\A$ to construct a non-crossing matching $M_i$ in which all marked nodes participate.
For this, we add to $M_i$ every pair of nodes $(u,v)\in V_{F'}\times V_{G'}$ such that  $\str{F'}[o_{F'}(u)]\simeq_{\A} \str{G'}[o_{G'}(v)]$, $\str{F'}[c_{F'}(u)]\simeq_{\A} \str{G'}[c_{G'}(v)]$, and $u$ or $v$ is marked.
By construction, this guarantees that $M_i$ is a non-crossing matching.
However, we discard $M_i$ if  $|M_i|>\frac{4}{h}(|F'|+|G'|)$
or there exists a marked node that does not participate in $M_i$.
If $M_i$ is not discarded, we use \cref{cor:partial} to build forests $F_i,G_i$ such that $\ted_{\le k}(F_i,G_i)=\ted^{M_i}_{\le k}(F',G')$,
and then we compute $d_i:=\ted_{\le k}(F_i,G_i)$ using \cref{thm:smallh}.
The final answer $d$ is defined as the minimum among the values $d_i$ computed for non-discarded matchings $M_i$.

Let us analyze the correctness of this approach.
\cref{prp:fur} provides the following guarantees with high probability:
\begin{itemize}
    \item $\ted_{\le k}(F,G)=\ted_{\le k}(F',G')$,
    \item for every alignment $\B\in \ta_k(F',G')$, we have $|\A \triangle \B| \le 4928k^4$.
\end{itemize}
Due to $\ted^{M_i}(F',G')\ge \ted(F',G')$, the reported value $d$ must clearly satisfy $d \ge \ted_{\le k}(F',G')= \ted_{\le k}(F,G)$.
The main challenge is to prove the converse inequality.
This task is trivial if $\ted(F,G)>k$.
Otherwise, it boils down to showing that, with high probability,
there exists $i\in [1\dd s]$ such that $d_i \le \ted_{\le k}(F,G)$.
Let us fix an optimum alignment $\B\in \ta(F',G')$.
We shall first prove that $d_i \le \ted_{\le k}(F,G)$ holds conditioned on the following events:
\begin{itemize}
    \item $|M_i|\le \frac{4}{h}(|F'|+|G'|)$;
    \item the alignment $\B$ that does not make any edits on the marked nodes
    and the parentheses corresponding to the marked nodes do not contribute to $\A \triangle \B$.
\end{itemize}
    Since $\B$ does not make any edits on the marked nodes, all the marked nodes participate in the following non-crossing matching: $\{(u,v)\in V_{F'}\times V_{G'} : u \simeq_{\B} v \text{ and } u\text{ or }v\text{ is marked}\}$.
Furthermore, this matching is equal to $M_i$ because the parentheses corresponding to marked nodes do not contribute to $\A \triangle \B$.
Consequently, $\ted^{M_i}(F',G')\le \ted_{\B}(F',G') \le k$,
and thus $d_i = \ted_{\le k}(F_i,G_i) = \ted^{M_i}_{\le k}(F',G') = \ted_{\le k}(F',G') = \ted_{\le k}(F,G)$.
It remains to prove that the favorable events hold with high probability for at least one $i\in [1\dd s]$.
For this, we analyze the complementary bad events.
For a random remainder modulo $h$, each individual node in $V_{F'}\cup V_{G'}$ is marked with probability $\frac{1}{h}$.
Hence, in expectation, the number of marked nodes is $\frac{1}{h}(|F'|+|G'|)$.
Given that each pair in $M_i$ contains a marked node, this means that $\Exp[|M_i|]\le \frac{1}{h}(|F'|+|G'|)$.
By Markov's inequality, we conclude that $|M_i| > \frac{4}{h}(|F'|+|G'|)$ holds with probability at most $\frac14$.
Moreover, there are at most $2k$ nodes affected by the edits of $\B$
and at most $9856k^4$ further nodes may contribute to $\B\triangle \A$.
By the union bound, at least one of these nodes is marked with probability at most $\frac{2k+9856k^4}{h} \le \frac12$.
Overall, for each $i\in [1\dd s]$, the probability of the bad events does not exceed $\frac{1}{4}+\frac12 = \frac34$.
Due to $s=\Theta(\log n)$, the probability that the bad events hold for all $i\in [1\dd s]$ does not exceed $(\frac{3}{4})^s = n^{-\Theta(1)}$.
Thus, the algorithm is correct with high probability.

Let us complete the proof with the running time analysis.
Applying \cref{prp:fur} costs $\Oh(n\log n+k^3)$ time.
For each $i\in [1\dd s]$, the set $M_i$ can be constructed and verified in $\Oh(n)$ time.
If $r_i$ is not discarded, then $|M_i|=\Oh(\frac{1}{h}(|F'|+|G'|))=\Oh(\frac{n}{k^4})$
and neither $F'$ nor $G'$ contains an $h$-node top-down path avoiding nodes participating in $M_i$.
Consequently, the application of \cref{cor:partial} takes $\Oh(n)$ time and produces forests of size $\Oh(n)$
and height at most $h+1 = \Oh(k^4)$.
The final application of \cref{thm:smallh} thus costs $\Oh(n + k^{15} \log k)$ time.
Since all steps, except for the preprocessing of \cref{prp:fur}, are repeated $\Oh(\log n)$ times,
the overall time complexity of our algorithm is $\Oh(n\log n + k^{15}\log k\log n)$.
\end{proof}

%% file: src/conclusion.tex
This paper gives an $\tOh(n + \poly(k))$-time algorithm for bounded tree edit distance, solving an open problem posed in~\cite{DBLP:conf/icalp/AkmalJ21,Xiao21}.
Multiple natural improvements and extensions of this result might be feasible. An immediate direction for future work is to significantly reduce the polynomial dependency on~$k$, which currently stays at $k^{15}$, akin to the recent progress for Dyck edit distance~\cite{F:22}. Another open question is whether the weighted version of tree edit distance admits an $\tOh(n + \poly(k))$-time algorithm (assuming that the cost of each edit is at least one). To the best of our knowledge, no such algorithm is known even for weighted string edit distance.

%% file: src/acknowledgment.tex
Barna Saha and Tomasz Kociumaka were partly supported by NSF 1652303, 1909046, and HDR TRIPODS Phase II grant 2217058. MohammadTaghi Hajiaghayi and Jacob Gilbert were partly supported by NSF CCF grants 2114269 and 2218678.